\documentclass[12pt]{article}
\usepackage{amssymb}
\usepackage{amsmath}
\usepackage{amsfonts}
\usepackage{geometry,enumerate,bbm}
\usepackage{natbib}
\usepackage{epsfig}
\usepackage{float}
\usepackage[onehalfspacing]{setspace}
\usepackage{graphicx}
\usepackage{subcaption}
\usepackage{lscape}
\usepackage[flushleft]{threeparttable}
\usepackage{ulem}
\usepackage{scrextend}
\usepackage{xcolor}

\RequirePackage[colorlinks,citecolor=blue,urlcolor=blue,linkcolor=blue]{hyperref}

\newtheorem{theorem}{Theorem}

\newtheorem{lemma}{Lemma}
\newtheorem{proposition}{Proposition}
\newenvironment{proof}[1][Proof]{\noindent \textbf{#1.} }{\  \rule{0.5em}{0.5em}}
\definecolor{darkgreen}{rgb}{0.0, 0.5, 0.13}
\newcommand\IgnoreThisText[1]{}

\input{tcilatex}
\begin{document}

\title{{\bf Dynamic Ordered Panel Logit Models}\thanks{
This research was
supported by the Gregory C. Chow Econometric Research Program at Princeton
University, by the National Science Foundation (Grant Number SES-1530741),
by the Economic and Social Research Council through the ESRC Centre for
Microdata Methods and Practice (grant numbers RES-589-28-0001, RES-589-28-0002 and ES/P008909/1),
by the Social Sciences and Humanities Research Council of Canada (grant number IG 435-2021-0778),
and by the European Research Council grants ERC-2014-CoG-646917-ROMIA and
ERC-2018-CoG-819086-PANEDA. We are grateful to seminar participants at the Universities of Leuven and Toronto, to Paul Contoyannis for helpful conversations about the BHPS data used in \citet{contoyannis_dynamics_2004}, and to Riccardo D'Adamo, Geert Dhaene, Sharada Dharmasankar and Runtong Ding for useful comments and discussion. We also thank Limor Golan and four anonymous referees for constructive feedback and suggestions.}
}
\author{\setcounter{footnote}{2}Bo E. Honor{\'e}\thanks{%
Princeton University and  The Dale T. Mortensen
Centre at the University of Aarhus, \texttt{honore@princeton.edu} }
 \and 
Chris Muris%
\thanks{%
McMaster University, \texttt{muerisc@mcmaster.ca}} 
\and 
 Martin Weidner%
\thanks{%
University of Oxford, \texttt{martin.weidner@economics.ox.ac.uk} } }
\date{April 2024}

\maketitle
\thispagestyle{empty}
\setcounter{page}{0}

\begin{abstract}
\noindent 
This paper studies a dynamic ordered logit model for panel data with fixed effects. The main contribution of the paper is to construct a set of valid moment conditions that are free of the fixed effects. The moment functions can be computed using four or more periods of data, and the paper presents sufficient conditions for the moment conditions to identify the common parameters of the model, namely the regression coefficients, the autoregressive parameters, and the threshold parameters. The availability of moment conditions suggests that these common parameters can be estimated using the generalized method of moments, and the paper documents the performance of this estimator using Monte Carlo simulations and an empirical illustration to self-reported health status using the British Household Panel Survey.
\end{abstract}

\newpage

\section{Introduction}

Panel surveys routinely collect data on an ordinal scale. For example, many nationally representative surveys ask respondents to rate their health or life satisfaction on an ordinal scale.\footnote{One example is the British Household Panel Survey in our empirical application. Others include the U.S. Health and Retirement Study and  Medical Expenditure Panel Survey, the Canadian Longitudinal Study on Ageing and the National Longitudinal Survey of Children and Youth, the Australian Longitudinal Study on Women's Health, the European Union Statistics on Income and Living Conditions, the Survey on Health, Ageing, and Retirement in Europe, among many others.} Other examples include test results in longitudinal data sets gathered for studying education.

We are interested in regression models for ordinal outcomes that allow for lagged dependent variables as well as fixed effects. In the  model that we propose, the ordered outcome depends on a fixed effect, a lagged dependent variable, regressors, and a logistic error term. We study identification and estimation of the finite-dimensional parameters in this model when only a small number ($\geq 4$) of time periods is available.

For other types of outcome variables (continuous outcomes in linear models, binary and multinomial outcomes), results for regression models with fixed effects and lagged dependent variables are already available. Such results are of great importance for applied practice, as they allow researchers to distinguish unobserved heterogeneity from state dependence, and to control for both when estimating the effect of regressors. The demand for such methods is evidenced by the popularity of existing approaches for the linear model, such as those proposed by \cite{arellano_tests_1991} and \cite{blundell_initial_1998}. In contrast, for ordinal outcomes, almost no results are available.

The challenge of accommodating unobserved heterogeneity in nonlinear models is well understood, especially when the researcher also wants to allow for lagged dependent variables. For example, while recent developments (\citealt{kitazawa2021transformations} and \citealt{honore2020dynamic}) relax these requirements, early work on the dynamic binary logit model with fixed effects either assumed no regressors, or restricted their joint distribution (cf. \citealt{Chamberlain1985} and \citealt{honore2000panel}). The challenge of accommodating unobserved heterogeneity in the ordered logit model seems even greater than in the binary model. The reason is that even the static version of the model is not in the exponential family (\citealt{hahn_note_1997}). As a result, one cannot directly appeal to a sufficient statistic approach. An alternative approach in the static ordered logit model is to reduce it to a set of binary choice models (cf. \citealt{das_panel_1999}, \citealt{JohnsonThesis2004}, \citealt{Baetschmann2015}, \citealt{Muris2017}, and \citealt{bmp2021}). Unfortunately, the dynamic ordered logit model cannot be similarly reduced to a dynamic binary choice model (see \citealt{muris2020dynamic}). Therefore, a new approach is needed. The contribution of this paper is to develop such an approach.
To do this, we follow the functional differencing approach in \cite{bonhomme2012functional} to obtain moment conditions for the finite-dimensional parameters in this model, namely the autoregressive parameters (one for each level of the lagged dependent variable), the threshold parameters in the underlying latent variable formulation, and the regression coefficients. Our approach is closely related to \cite{honore2020dynamic}, and can be seen as the extension of their method to the case of an ordered response variable.

This paper contributes to the literature on dynamic ordered logit models. We are aware of only one paper that studies a fixed-$T$ version of this model while allowing for fixed effects. The approach in \cite{muris2020dynamic} builds on methods for dynamic binary choice models in \cite{honore2000panel} by restricting how past values of the dependent variable enter the model. In particular, in \cite{muris2020dynamic}, the lagged dependent variable $Y_{i,t-1}$ enters the model only via $\mathbbm{1}\{Y_{i,t-1}\geq k\}$ for some known $k$. We do not impose such a restriction, and allow the effect of $Y_{i,t-1}$ to vary freely with its level.
Other existing work on dynamic panel models for ordered outcomes uses a random effects approach (\citealt{contoyannis_dynamics_2004}, \citealt{albarran_estimation_2019}) or requires a large number of time periods for consistency (\citealt{carro_state_2014}, \citealt{fernandez-val_evaluating_2017}). An earlier version of \cite{aristodemou2018semiparametric} contained partial identification results for a dynamic ordered choice model without logistic errors. Our approach places no restrictions on the dependence between fixed effects and regressors, requires only four periods of data for consistency, and delivers point identification and estimates.

More broadly, this paper contributes to the literature on fixed-$T$ identification and estimation in nonlinear panel models with fixed effects (see  \citealt{honore_nonlinear_2002}, \citealt{arellano2003discrete},  and \citealt{ArellanoBonhomme2011} for overviews). The literature contains results for several models adjacent to ours. For example, the static panel ordered logit model with fixed effects was studied by \cite{das_panel_1999}, \cite{JohnsonThesis2004}, \cite{Baetschmann2015}, and \cite{Muris2017}; results for static and dynamic binomial and multinomial choice models are in \cite{chamberlain_analysis_1980}, \cite{honore2000panel}, \cite{magnac2000subsidised}, \cite{shi_estimating_2018}, \cite{aguirregabiriaSufficientStatisticsUnobserved2021},
\cite{aguirregabiriaIdentificationAverageMarginal},
\cite{PakesPorterShepardCalderWang2022} and \cite{khan_inference_2021}.

Our main contribution is to obtain novel moment conditions for the common parameters in the dynamic ordered logit model with fixed effects. Additionally, we  obtain conditions under which these moment conditions identify those parameters. Finally, we discuss the  implied generalized method of moments (GMM) estimator and demonstrate its performance in both a Monte Carlo study and an empirical application to self-reported health status in the British Household Panel Study.

The remainder of this paper is organized as follows. Section~\ref{sec:model} introduces the model and the moment conditions that are free of fixed effects. Section~\ref{sec:Identification} presents identification results for the common model parameters based on those moment conditions.
Section~\ref{Section: Implication for estimation} discusses how to use the moment conditions for estimation and inference.
Section~\ref{PracticalPerformance}  explores the practical performance of the resulting estimation method through Monte Carlo simulations, including comparison with a correlated random effects approach. 
That section also provides an empirical illustration to health data. Section~\ref{Conclusions} concludes the paper. The appendix provides proofs and further computational details.

\section{Model and moment conditions}
\label{sec:model}

In this section, we first describe the panel ordered logit model that 
is used throughout the paper, and then present moment
conditions for the model that can be used to estimate the common
parameters of the model without imposing any knowledge of the 
individual-specific effects.

\subsection{Model and notation}

We consider panel data with cross-sectional units  
$i=1,\ldots,n$ and time periods $t=0,\ldots,T$. 
For each pair $(i,t)$, 
we observe the discrete outcome $Y_{it}\in \{1,2,\ldots,Q\}$,
which can take $Q \in \{2,3,4,\ldots\}$ different values,
and the strictly exogenous regressors $X_{it}\in \mathbb{R}^{K}$.
We discuss unbalanced panels in Section~\ref{sec:MomentsTlarger3}, but for now, we assume a balanced panel where outcomes
are observed for all $t \geq 0$ and regressors for all $t \geq 1$.
Thus, the total number of time periods for which outcomes are
observed is $T+1$. For $t \geq 1$, the observed discrete outcomes depend on
an unobserved latent variable $Y^*_{it} \in \mathbb{R}$ as follows:
\begin{align}
     Y_{it} = \left\{ 
     \begin{array}{ll}
          1 & \text{if \; $  Y^*_{it}  \; \in \; (-\infty \,,\lambda_{1}]$,}
          \\
          2 & \text{if \; $  Y^*_{it}   \; \in \; (\lambda_{1} \, ,\lambda_{2}]$,}
          \\
          \vdots
          \\
          Q &   \text{if \; $ Y^*_{it}   \; \in \; (\lambda_{Q-1} \, , \infty)$,}
     \end{array}
     \right.
   \label{ModelOutcomes}     
\end{align}
where $\lambda=(\lambda_1,\ldots,\lambda_{Q-1}) \in \mathbb{R}^{Q-1}$ is a vector of unknown parameters with $\lambda_1 < \lambda_2 < \ldots < \lambda_{Q-1}$.
The  latent variable is generated by the model
\begin{align}
    Y^*_{it} = X_{it}^{\prime }\,\beta  +  \sum_{q=1}^{Q} \gamma_q  \, \mathbbm{1}\left\{ Y_{i,t-1} = q \right\} + A_i + \varepsilon_{it} ,
    \label{ModelYstar}
\end{align}
with unknown parameters $\beta \in \mathbb{R}^K$ and $\gamma = (\gamma_1,\ldots,\gamma_Q) \in \mathbb{R}^{Q}$. Here, $A_i \in \mathbb{R} \cup \{\pm \infty\}$ is an unobserved individual-specific effect
whose distribution is not specified, and $A_i$ is allowed to be arbitrarily correlated with the regressors $X_{it}$ and the initial conditions $Y_{i0}$. Let $X_{i}:=(X_{i1},\ldots ,X_{iT})$.
Conditional on $Y_{i0}$, $X_i$, 
and $A_i$, 
the idiosyncratic error term $ \varepsilon_{it}$ is assumed to 
be independent and identically distributed over $t$ with 
cumulative distribution function  $ \Lambda(\varepsilon) :=  [1+\exp(-\varepsilon)]^{-1}$.
Thus, $ \varepsilon_{it}$ is a logistic error term.
For the cross-sectional sampling, we assume that
$(Y_{i0}, X_{i1},\ldots ,X_{iT}, A_{i}, \varepsilon_{i1},\ldots ,\varepsilon_{iT})$ are independent and identically distributed across $i$.

The model described by \eqref{ModelOutcomes} and \eqref{ModelYstar} is a dynamic ordered panel logit model, where an arbitrary function $\gamma_{Y_{i,t-1}}$ of the lagged dependent variable $Y_{i,t-1}$ is allowed to enter additively into the latent variable
$Y^*_{it}$. 
This model strikes a balance between
a general functional form and a parsimonious parameter structure. 
We discuss possible generalizations of the 
model for $Y^*_{it}$ in Section~\ref{sec:GeneralModel}, but 
otherwise impose \eqref{ModelYstar} throughout the paper.\footnote{If the observed 
$Y_{it}$ is a discretized version of a continuous variable with a natural economic interpretation, then it would be more natural to model the state dependence in \eqref{ModelYstar} as $Y_{i,t-1}^{\star}\gamma$. A numerical investigation suggests that it is not possible to develop moment conditions for such a model. This suggests that the common parameters in this model are not point-identified, or are only point-identified under strong support assumptions on the covariates and hence not generally $\sqrt{n}$ estimable. 
We explore this alternative model in Appendix \ref{app:alternative_specification}.}

Our ultimate goal is to estimate the unknown parameters  $\theta = (\beta,\gamma,\lambda) \in \Theta := \mathbb{R}^{K+2Q-1}$ without imposing any assumptions on the individual-specific effect $A_i$. This requires  two normalizations, because 
 common additive shifts of all the parameters  $\gamma_q$ or  of all the parameters $\lambda_q$ can 
be absorbed into $A_i$. For example, we could impose 
the normalizations $\gamma_1=0$ and $\lambda_1=0$, but in this section there is no need to specify such normalizations.

It is convenient to define $\lambda_{0} := -\infty$, and $\lambda_{Q} := \infty$,
and  
\begin{align}
    z(Y_{i,t-1},X_{it},\theta) :=  X_{it}^{\prime }\,\beta  +  \sum_{q=1}^{Q} \gamma_q  \, \mathbbm{1}\left\{ Y_{i,t-1} = q \right\} .
    \label{DefSingleIndex}
\end{align}    
With this notation, the  model assumptions imposed so far imply that the distribution of $Y_{it}$ conditional on the regressors $X_{i}$,
past outcomes $Y_{i}^{t-1}=(Y_{i,t-1},Y_{i,t-2},\ldots )$, and fixed effects $A_i$,
is given by
\begin{equation}
 \mathrm{Pr}\left( Y_{it}=q \, \Big| \, Y_{i}^{t-1},X_{i},A_{i}, \theta
\right) 
 =  \Lambda\Big[ z(Y_{i,t-1},X_{it},\theta) + A_{i} - \lambda_{q-1}  \Big] -  \Lambda\Big[  z(Y_{i,t-1},X_{it},\theta) + A_{i} - \lambda_{q}  \Big] 
\label{model}
\end{equation}
for all $i\in \{1,\ldots ,n\}$, $t\in \{1,2,\ldots ,T\}$, and $q \in  \{1,2,\ldots,Q\}$. 
Let $Y_{i}=(Y_{i1},\ldots ,Y_{iT})$, and let the true model parameters be denoted by $\theta^0 = (\beta^0,\gamma^0,\lambda^0)$. In the following, all probabilistic statements are for the model distribution generated under $\theta^0$. For
example, we have $\mathrm{Pr}\big(Y_{i}=y_{i}\,\big|\,Y_{i0}=y_{i0},\,X_{i}=x_{i},%
\,A_{i}=\alpha _{i}\big)=p_{y_{i0}}(y_{i},x_{i},\theta^0,\alpha _{i})$, where
\begin{equation}
p_{y_{i0}}(y_{i},x_{i},\theta ,\alpha _{i}) :=\prod_{t=1}^{T}
\left\{  \Lambda\Big[ z(y_{i,t-1},x_{it},\theta) + \alpha_{i} - \lambda_{y_{it}-1}  \Big] -  \Lambda\Big[  z(y_{i,t-1},x_{it},\theta) + \alpha_{i} - \lambda_{y_{it}}  \Big]
\right\} .
 \label{DefProb}
\end{equation}%
Below, we drop the index $i$ until we discuss estimation;
instead of $Y_{i0}$, $Y_i$, $X_i$, $A_i$, we just write $Y_0$, $Y$, $X$, $A$ for those random variables and random vectors.

\subsection{Moment condition approach}
\label{subsec:momentApproach}

In the next subsection, we present moment functions for the ordered logit model discussed above. These are functions
$m:\{1,\ldots,Q\} \times \{1,\ldots,Q\}^T \times \mathbb{R}^{T \times K} \times \Theta
\rightarrow \mathbb{R}$ such that 
\begin{align}
\mathbb{E} \left[ m_{Y_0}(Y,X,\theta^0) \, \Big| \, Y_0=y_0,
\, X=x, \, A=\alpha \right] &= 0 
   \label{MomentFunctionIdea}
\end{align}
for all $y_0 \in \{1,\ldots,Q\}$, $x \in \mathbb{R}^{T \times K}$,
and $\alpha \in \mathbb{R} \cup \{\pm \infty\}$. We write the first argument $y_0$ of the moment function as an index, but that is purely
for notational convenience.
Conditional on 
$Y_0=y_0$, $X=x$, and $A=\alpha$, the distribution of 
$Y=(Y_1,\ldots,Y_T) \in \{1,\ldots,Q\}^T$ is given by
\eqref{DefProb}. The model assumptions in the last subsection are therefore sufficient to evaluate the conditional expectation
in \eqref{MomentFunctionIdea}.

If we can establish 
the conditional moment condition \eqref{MomentFunctionIdea} then,
by the law of iterated  expectations, we also have 
the unconditional moment conditions
\begin{align}
\mathbb{E} \left[ h(Y_0,X,\theta^0) \,  m_{Y_0}(Y,X,\theta^0)   \right] &= 0 
   \label{MomentFunctionUnconditional}
\end{align}
for any function $h : \{1,\ldots,Q\} \times \mathbb{R}^{T \times K} \times \Theta
\rightarrow \mathbb{R}$ such that the expectation is well-defined.
Those unconditional moment conditions can then be used to estimate
the model parameters $\theta^0$ by the generalized method of moments (GMM).
Such an estimation approach solves the incidental parameter
problem (\citealt{neyman1948consistent}), because the moment condition
\eqref{MomentFunctionUnconditional} does not feature the
individual-specific effect $A$ at all. No 
 assumptions are imposed on the distribution of those nuisance parameters,
and they need not be estimated. On the flip-side, this implies
that we do not learn anything about the distribution of $A$. Notice, however, that if one is interested in (functions of)
the individual-specific effects such as average partial effects, then the estimation of the common parameters $\theta$ will always be
a key first step in any inference procedure.

The moment condition approach just described eliminates
the individual-specific effect $A$ from the estimation, because
\eqref{MomentFunctionIdea} is assumed to hold for all $\alpha \in \mathbb{R} \cup \{\pm \infty\}$, but the moment function
$m_{Y_0}(Y,X,\theta^0)$ does not depend on $A$ at all.
The existence of moment functions with this property is quite 
miraculous:
for any given values of $Y_0=y_0$, $X=x$, and $\theta^0$, the moment function
$m_{y_0}(\cdot,x,\theta^0) : \{1,\ldots,Q\}^T \rightarrow \mathbb{R}$
can be viewed as a finite-dimensional vector (with $Q^T$ real numbers),
but \eqref{MomentFunctionIdea} imposes an infinite number of linear constraints -- one for each $\alpha \in \mathbb{R} \cup \{\pm \infty\}$.
The logistic assumption on $\varepsilon_{it}$ is important for finding
solutions 
of this infinite-dimensional linear system in a finite number
of variables, and for most choices of error distributions  (e.g.\ 
normally distributed error),
we do not expect such solutions to exist. It seems 
likely that for the error distributions in \cite{johnson2004identification} and \cite{davezies2022fixed},
and also for a mixture of logistics (briefly discussed in
\citealt{honore2020dynamic}), one could also find 
valid moment conditions, if a sufficient number of
time periods are available, but we focus purely
on logistic errors in this paper.

In the following, we present moment functions $m_{Y_0}(Y,X,\theta^0)$
that satisfy \eqref{MomentFunctionIdea}. 
We  derive those moment functions for the dynamic panel ordered
logit model analogously to the results for the dynamic panel
 binary choice logit model in \cite{honore2020dynamic}. 
 Indeed, for the binary choice case ($Q=2$), our moment functions below
 exactly coincide with those in \cite{honore2020dynamic}, and
 we refer to that paper for more details on the derivation, which is closely related to the functional differencing method in 
 \cite{bonhomme2012functional}.
Once we have obtained expressions for the moment functions,  their derivation is no longer relevant and we can focus on showing that they are valid -- i.e. that \eqref{MomentFunctionIdea} holds -- and on their implications for the identification and estimation of $\theta$. 

\subsection{Moment conditions for $T=3$}

We first introduce our moment functions for $T=3$. In our convention, this means that outcomes $Y_t$ are observed
for the four time periods $t =0,1,2,3$ (including the initial conditions $Y_0$).
We have verified numerically
that no moment functions satisfying \eqref{MomentFunctionIdea} 
for general parameter and regressor values
exist for $T<3$,
and for the binary choice case ($Q=2$)   a proof of this non-existence is given in \cite{honore2020dynamic}. Thus, $T=3$ is the smallest
number of time periods that we can consider.

We use lower case letters for generic arguments (as opposed to
random variables) of the moment function
$m_{y_0}(y,x,\theta)$, where $y_0 \in \{1,\ldots,Q\}$,
$y \in \{1,\ldots,Q\}^T$, $x \in \mathbb{R}^{T \times K}$
and $\theta = (\beta,\gamma,\lambda) \in \Theta$. 
The
$t$'th row of $x$ is denoted by $x_t' \in \mathbb{R}^K$,
and we define $x_{ts} := x_t - x_s$, $\gamma_{qr} := \gamma_q - \gamma_r$
and $\lambda_{qr} := \lambda_q - \lambda_r$.

We find multiple moment functions 
$m_{y_0,q_1,q_2,q_3} \allowbreak (y,x,\theta)$
which are distinguished by the additional indices $q_1 \in \{1,\ldots,Q-1\}$, $q_2 \in \{1,\ldots,Q\}$, $q_3 \in \{1,\ldots,Q-1\}$.
For the moment function labelled by $q_1$, $q_2$, $q_3$, the dependence
on $y$ is only through the coarser outcome $\widetilde y_{q_1,q_2,q_3}(y)
\in \{0,1\} \times \{1,2,3\} \times \{0,1\}$, which is a vector
with three components $ \widetilde y_{t,q_1,q_2,q_3}(y) $, $t \in \{1,2,3\}$,
given by
 \begin{align*}
        \widetilde y_{2,q_1,q_2,q_3}(y) &:= \left\{   
              \begin{array}{ll}  1 & \text{if } y_2 < q_2 ,
                           \\
                                     2 & \text{if } y_2 = q_2 ,
                           \\
                                     3 & \text{if } y_2 > q_2 .
           \end{array} \right.
   &
        \widetilde y_{t,q_1,q_2,q_3}(y) &:=   \mathbbm{1}\left\{  y_r > q_r \right\} ,
        \quad \text{for $t \in \{1,3\}$}.
\end{align*}
The moment functions presented below have the property that
\begin{align}
    m_{y_0,q_1,q_2,q_3}(y,x,\theta) = m_{y_0,q_1,q_2,q_3}(\widetilde y_{q_1,q_2,q_3}(y),x,\theta) .
    \label{BinarizationTrinarization}
\end{align}
Thus, for given $q_1$, $q_2$, $q_3$, it would be sufficient to observe the outcome
$\widetilde Y = \widetilde y_{q_1,q_2,q_3}(Y)$ to implement this moment function.
Notice that $ \widetilde y_{t,q_1,q_2,q_3}(y)$, for  time periods
$t=1$ and $t=3$ is just a binarization, as in
\cite{Muris2017} and \cite{muris2020dynamic}, but for
$t=2$ we crucially deviate from those existing papers, because
for $q_2 \in \{2,\ldots,Q-1\}$ the coarser outcome
$\widetilde y_{2,q_1,q_2,q_3}(y)$ is a trinarization of the second
period outcome, not a binarization. It turns out that this is a crucial 
extension to obtain all the valid moment conditions in our model.

The moment functions presented below were obtained
using Mathematica by following
 the methods described in Section~2 of
\cite{honore2020dynamic}, which builds on 
ideas in \cite{bonhomme2012functional}. 
Once derived, we can prove by hand that the moment functions are valid (proof of Theorem~\ref{th:momentsT3} in the appendix),
but we do not have
any useful explanation or intuition for the detailed functional
form of these moment functions. 
However, 
the binarized/trinarized outcome $\widetilde Y$
and equation \eqref{BinarizationTrinarization}
help to appreciate some aspects of the structure
of the moment functions (see also  Lemma~\ref{lemma:MASTER} in the appendix).
Furthermore, while the functional form of the 
moment functions is mysterious, one can show
the existence of valid moment functions much more easily,
see Appendix~\ref{app:LowerBound}.

For $q_1,q_3 \in \{1,\ldots,Q-1\}$ and $q_2 \in \{2,\ldots,Q-1\}$, we define
 \begin{align}
  &m_{y_0,q_1,q_2,q_3}(y,x,\theta) 
 \nonumber \\[10pt]
   &:= 
    \left\{  \begin{array}{ll}
       \text{\footnotesize $   \exp\left(   x_{13}^{\prime }\,\beta  +   \gamma_{y_0,q_2}  + \lambda_{q_3,q_1}   \right) 
             \frac     {\displaystyle  \exp(x_{32}' \, \beta + \gamma_{q_2,y_1} + \lambda_{q_2,q_3}) - 1   }
             {\displaystyle \exp\left(   \lambda_{q_2,q_2-1}  \right)  - 1  }  
             $}
             & \text{if  $ y_1 \leq q_1, \, y_2=q_2 , y_3 \leq q_3$}, 
            \\[6pt]
         \text{\footnotesize $  \exp\left(   x_{13}^{\prime }\,\beta  +   \gamma_{y_0,q_2}  + \lambda_{q_3,q_1}   \right) 
              \frac     {\displaystyle 1 - \exp(x_{23}' \, \beta + \gamma_{y_1,q_2} + \lambda_{q_3,q_2})  }
            {\displaystyle 1 -  \exp\left(   \lambda_{q_2-1,q_2}  \right)  } 
             $}
             & \text{if }   y_1 \leq q_1, \, y_2=q_2 , \, y_3 > q_3,            
            \\[6pt]
          \text{\footnotesize $     \exp\left(   x_{13}^{\prime }\,\beta  +   \gamma_{y_0,q_2}  + \lambda_{q_3,q_1}   \right) 
             $}
             & \text{if }   y_1 \leq q_1, \,  y_2>q_2 ,
            \\[6pt]
        \text{\footnotesize $  -1  $}
         & \text{if }   y_1 > q_1, \, y_2<q_2 ,                        
            \\[6pt]
 \text{\footnotesize $ - \,   \frac{\displaystyle      1 - \exp(x_{32}' \, \beta + \gamma_{q_2,y_1}  + \lambda_{q_2-1,q_3})   }
            {\displaystyle 1 -  \exp\left(   \lambda_{q_2-1,q_2}  \right)  } 
             $}
             & \text{if }   y_1 > q_1, \, y_2=q_2 , \, y_3 \leq q_3, 
            \\[6pt]
 \text{\footnotesize $ - \,   \frac{ \displaystyle       \exp( x_{23}' \, \beta + \gamma_{y_1,q_2}  + \lambda_{q_3,q_2-1}) - 1    }
             {\displaystyle  \exp\left(   \lambda_{q_2,q_2-1}    \right) - 1  } 
             $}
               & \text{if }   y_1 > q_1, \, y_2=q_2 , \,y_3 > q_3, 
            \\[6pt]         
             \text{\footnotesize $ 0 $}
            & \text{otherwise}.
       \end{array} \right.
     \label{MomentFunctionsT3}  
\end{align}
Any valid moment function satisfying \eqref{MomentFunctionIdea} 
can be multiplied by an arbitrary constant and remain a valid
moment function. In \eqref{MomentFunctionsT3}
we used that rescaling freedom to
 normalize the entry for the case 
 $(y_1 > q_1, y_2<q_2)$ to be equal to $-1$. 
If, alternatively, we  normalize the entry for $(y_1 \leq q_1, y_2>q_2)$ to be equal to $-1$, then we obtain the equally valid moment function
\begin{align*}
    \widetilde m_{y_0,q_1,q_2,q_3}(y,x,\theta) &= - \frac{ m_{y_0,q_1,q_2,q_3}(y,x,\theta)} {  \exp\left(   x_{13}^{\prime }\,\beta  +   \gamma_{y_0,q_2}  + \lambda_{q_3,q_1}   \right) } .
\end{align*}
This rescaling is interesting, because
if we reverse the order of the outcome labels (i.e.\ $Y_t \mapsto Q+1-Y_{t}$),  the model remains unchanged  except for the parameter
transformations $\beta \mapsto - \beta$, $\gamma_q \mapsto - \gamma_{Q+1-q}$, and $\lambda_q \mapsto - \lambda_{Q-q}$. 
Under this transformation, the moment function $m_{y_0,q_1,q_2,q_3}(y,x,\theta) $
becomes $  \widetilde m_{\widetilde y_0,\widetilde q_1,\widetilde q_2,\widetilde q_3}(y,x,\theta) $
with $\widetilde y_0 = Q+1- y_0$ and 
$(\widetilde q_1,\widetilde q_2,\widetilde q_3) = (Q-q_1, Q+1-q_2, Q-q_3)$. This transformation therefore does not deliver any new moment functions,
which are not already (up to rescaling) given by~\eqref{MomentFunctionsT3}.

Equation \eqref{MomentFunctionsT3} does not define $m_{y_0,q_1,q_2,q_3}(y,x,\theta)$ for $q_2=1$ and $q_2 = Q$.
If we  plug those values of $q_2$ into \eqref{MomentFunctionsT3}, then various undefined terms appear
since $\lambda_{0} = -\infty$ and $\lambda_{Q} = \infty$. However, 
if for $q_2=1$ we properly evaluate the limit of
$\widetilde m_{y_0,q_1,q_2,q_3}(y,x,\theta)$ as
$\lambda_{0}  \rightarrow -\infty$,
then we obtain
\begin{align}
  m_{y_0,q_1,1,q_3}(y,x,\theta)  &:= 
    \left\{  \begin{array}{ll}
             \exp\left(  x_{23}' \, \beta + \gamma_{y_1,1} + \lambda_{q_3,1}  \right)     -  1 & \text{if }  y_1 \leq q_1, \, y_2=1 , \, y_3>q_3, \\
             - 1 & \text{if }   y_1 \leq q_1, \, y_2>1,\\
                 \exp\left( x_{31}' \, \beta + \gamma_{1,y_0} + \lambda_{q_1,q_3}  \right)  & \text{if }  y_1 > q_1, \, y_2=1 , \,y_3\leq q_3,\\
                \exp\left( x_{21}' \, \beta + \gamma_{y_1,y_0} + \lambda_{q_1,1}  \right)  & \text{if } y_1 > q_1, \, y_2=1 , \, y_3>q_3,\\
             0 & \text{otherwise}.
       \end{array} \right.
     \label{MomentFunctionsT3b}  
\end{align}
Similarly, if for $q_2 = Q$ we properly evaluate the limit of  $m_{y_0,q_1,q_2,q_3}(y,x,\theta)$ as $\lambda_{Q}  \rightarrow \infty$,
then we obtain
\begin{align}
   m_{y_0,q_1,Q,q_3}(y,x,\theta)  &:= 
    \left\{  \begin{array}{ll}
           \exp\left( x_{12}' \, \beta + \gamma_{y_0,y_1} +   \lambda_{Q-1,q_1}  \right) & \text{if }   y_1 \leq q_1, \, y_2=Q , \, y_3\leq q_3, \\
              \exp\left(  x_{13}' \, \beta + \gamma_{y_0,Q} + \lambda_{q_3,q_1}  \right)   & \text{if }  y_1 \leq q_1, \, y_2=Q , \, y_3>q_3,\\
               -1  & \text{if } y_1 > q_1, \, y_2<Q,\\
              \exp\left( x_{32}' \, \beta + \gamma_{Q,y_1} + \lambda_{Q-1,q_3}    \right)  -1  & \text{if } y_1 > q_1, \, y_2=Q , \, y_3\leq q_3,\\
             0 & \text{otherwise}.
       \end{array} \right.
     \label{MomentFunctionsT3c}  
\end{align}

Notice that the moment functions for 
$q_2=1$ and $q_2=Q$ also satisfy \eqref{BinarizationTrinarization}, but
here $\widetilde y_{q_1,q_2,q_3}(y)$ corresponds to a binarization of
the outcome in all time periods, because $\widetilde y_{t,q_1,q_2,q_3}(y)$
also only takes two values for those values of $q_2$. Those moment
functions are therefore conceptually much closer
to \cite{Muris2017} and \cite{muris2020dynamic}, and they also 
incorporate the moment functions for the dynamic binary choice model
in \cite{honore2020dynamic} as special cases.

Together, the  formulas \eqref{MomentFunctionsT3}, \eqref{MomentFunctionsT3b}, and \eqref{MomentFunctionsT3c}
provide one moment function for every value of $(y_0,q_1,q_2,q_3) \in \{1,\ldots,Q\} \times \{1,\ldots,Q-1\} \times \{1,\ldots,Q\} \times  \{1,\ldots,Q-1\} $, and these constitute all our moment
functions for the dynamic ordered logit model with $T=3$.%
\footnote{By the limiting arguments ($\lambda_{0}  \rightarrow -\infty$ and $\lambda_{Q}  \rightarrow \infty$) described
above, all of those moment functions are already implicitly defined
via \eqref{MomentFunctionsT3} alone.}
The following theorem states that these are indeed valid moment functions for the dynamic panel ordered logit model, independent of the value
of the fixed effect $A$.

\begin{theorem}
\label{th:momentsT3}   
If the outcomes $Y=(Y_1,Y_2,Y_3)$ are generated
from model \eqref{model} with $Q \geq 2$, $T=3$ and true parameters $\theta^0 = (\beta^0,\gamma^0,\lambda^0)$, then we have for all
$y_0 \in \{1,\ldots,Q\}$, $q_1,q_3 \in \{1,\ldots,Q-1\}$, $q_2 \in \{1,\ldots,Q\}$, $x \in \mathbb{R}^{K
\times 3}$, and $\alpha \in \mathbb{R}  \cup \{\pm \infty\}$ that
\begin{align*}
\mathbb{E} \left[ m_{y_0,q_1,q_2,q_3}(Y,X,\theta^0) \, \big| \, Y_0=y_0,
\, X=x, \, A=\alpha \right] &= 0 .
\end{align*}
\end{theorem}
The proof of the theorem is given in the appendix.
For any fixed value of $Q$ one could, in principle, show by direct calculation that
\begin{align*}
\sum_{y \in \{ 1,2,\ldots,Q\}^3} \, p_{y_0}(y,x,\theta^0,\alpha) \;
m_{y_0,q_1,q_2,q_3}(y,x,\theta^0) = 0
\end{align*}
for the model probabilities $p_{y_0}(y,x,\theta^0,\alpha)$ given by \eqref{DefProb},
but our proof in the appendix does not rely on such a brute force calculation and is valid for any $Q \geq 2$.

For each initial condition $y_0$ we thus have $\ell = Q (Q-1)^2$ available moment conditions.
For example, for $Q=2,3,4,5$ there are respectively $\ell=2$, $12$, $36$, $80$ available moment conditions for each initial condition. 
For those values of $Q$ we have verified numerically that our $\ell$ moment conditions are linearly independent, and that they constitute all the valid moment conditions that are available for the dynamic panel ordered logit model with $T=3$, for generic values of $\gamma$.\footnote{If some of the parameters $\gamma_q$ are equal to each other, then additional moment conditions become available.}
 
We believe that this is true for all $Q \geq 2$, but a proof of this completeness result is beyond the scope of this paper.
For the special case of dynamic binary choice ($Q=2$), the moment conditions here are identical to those in \cite{honore2020dynamic}
and \cite{kitazawa2021transformations},
and the completeness of those binary choice moment conditions is discussed in \cite{kruiniger2020further}
and \cite{dobronyi2021identification}.

\subsection{Moment conditions for $T>3$}
\label{sec:MomentsTlarger3}

We now consider the case where the econometrician has data for more than three time periods (in addition to the period that gives the initial condition).
Obviously, all the moment conditions above for $T=3$ are still
valid when applied to three consecutive periods, but additional moment conditions become available
for $T>3$. We first consider moment conditions that are based on the outcome in three periods, where the last two are consecutive.
Let $z_t := z(y_{t-1},x_{t},\theta)$, with $z(\cdot,\cdot,\cdot)$
defined in \eqref{DefSingleIndex},
and define $z_{ts} := z_t - z_s$.
For $y_0 \in \{1,\ldots,Q\}$, $q_1,q_3 \in \{1,\ldots,Q-1\}$, $q_2 \in \{2,\ldots,Q-1\}$, 
and $t,s \in \{1,2,\ldots,T-1\}$ with $t<s$ we define  
 \begin{align}
  m^{(t,s,s+1)}_{y_0,q_1,q_2,q_3}(y,x,\theta)  &:= 
    \left\{  \begin{array}{ll}
           \exp\left(   z_{t,s+1} + \lambda_{q_3,q_1}  \right)  \frac     {\displaystyle \exp(z_{s+1,s} + \lambda_{q_2,q_3}) - 1   }
             {\displaystyle  \exp\left(   \lambda_{q_2,q_2-1}  \right) - 1   }  
             & \text{if }  y_t \leq q_1, \, y_s=q_2 , y_{s+1} \leq q_3, 
            \\[10pt]
           \exp\left(   z_{t,s+1} + \lambda_{q_3,q_1}  \right)     \frac     {\displaystyle 1 - \exp(z_{s,s+1} + \lambda_{q_3,q_2})  }
            {\displaystyle 1 -  \exp\left(   \lambda_{q_2-1,q_2}  \right)  } 
             & \text{if }   y_t \leq q_1, \, y_s=q_2 , \, y_{s+1} > q_3,            
            \\[10pt]
               \exp\left(   z_{t,s+1}  + \gamma_{y_s,q_2} + \lambda_{q_3,q_1} \right)  
             & \text{if }   y_t \leq q_1, \,  y_s>q_2 ,
            \\[10pt]
         -1  
         & \text{if }   y_t > q_1, \, y_s<q_2 ,                        
            \\[10pt]
 - \,   \frac{\displaystyle      1 - \exp(z_{s+1,s} + \lambda_{q_2-1,q_3})   }
            {\displaystyle 1 -  \exp\left(   \lambda_{q_2-1,q_2}  \right)  } 
             & \text{if }   y_t > q_1, \, y_s=q_2 , \, y_{s+1} \leq q_3, 
            \\[10pt]
 - \,  \frac{ \displaystyle       \exp( z_{s,s+1} + \lambda_{q_3,q_2-1})  -  1  }
             {\displaystyle  \exp\left(   \lambda_{q_2,q_2-1}  \right) - 1   } 
               & \text{if }   y_t > q_1, \, y_s=q_2 , \,y_{s+1} > q_3, 
            \\[10pt]          0 
            & \text{otherwise}.
       \end{array} \right.
     \label{MomentFunctionsTgeneral}  
\end{align}
For  $T=3$, $t=1$, and $s=2$, it is straightforward to verify that  $ m^{(t,s,s+1)}_{y_0,q_1,q_2,q_3}(y,x,\theta) $ in equation \eqref{MomentFunctionsTgeneral}
equals 
the moment function in equation \eqref{MomentFunctionsT3}.
For larger values of $T$,
the moment function in \eqref{MomentFunctionsTgeneral}
can be implemented as long as outcomes
are observed for the time periods $\{t-1,t,s-1,s,s+1\}$
and covariates are observed for time periods
$\{t,s,s+1\}$.

Since $\lambda_0= -\infty$ and $\lambda_Q=\infty$, equation \eqref{MomentFunctionsTgeneral} can not be used to  define a moment function when $q_2$ equals 1 or $Q$. We next define moment functions for these cases.
For $y_0 \in \{1,\ldots,Q\}$, $q_1,q_3 \in \{1,\ldots,Q-1\}$, $t,s,r\in \{1,2,\ldots ,T\}$, and $t<s<r$, we
define
\begin{align}
  m^{(t,s,r)}_{y_0,q_1,1,q_3}(y,x,\theta)  &:= 
    \left\{  \begin{array}{ll}
             \exp\left(  z_{sr} + \lambda_{q_3,1}  \right)     -  1 & \text{if }  y_t \leq q_1, \, y_s=1 , \, y_r>q_3, \\
             - 1 & \text{if }   y_t \leq q_1, \, y_s>1,\\
                 \exp\left( z_{rt} + \lambda_{q_1,q_3}  \right)  & \text{if }  y_t > q_1, \, y_s=1 , \,y_r\leq q_3,\\
                \exp\left( z_{st}   + \lambda_{q_1,1}  \right)  & \text{if } y_t > q_1, \, y_s=1 , \, y_r>q_3,\\
             0 & \text{otherwise},
       \end{array} \right.
\nonumber \\[10pt]
   m^{(t,s,r)}_{y_0,q_1,Q,q_3}(y,x,\theta)  &:= 
    \left\{  \begin{array}{ll}
           \exp\left( z_{ts}  +   \lambda_{Q-1,q_1}  \right) & \text{if }   y_t \leq q_1, \, y_s=Q , \, y_r\leq q_3, \\
              \exp\left(  z_{tr} + \lambda_{q_3,q_1}  \right)   & \text{if }  y_t \leq q_1, \, y_s=Q , \, y_r>q_3,\\
               -1  & \text{if } y_t > q_1, \, y_s<Q,\\
              \exp\left( z_{rs}  + \lambda_{Q-1,q_3}    \right)  -1  & \text{if } y_t > q_1, \, y_s=Q , \, y_r\leq q_3,\\
             0 & \text{otherwise}.
       \end{array} \right.
    \label{MomentFunctionsTgeneralBC}   
\end{align}
When $T$ equals 3 and $(t,s,r)=(1,2,3)$, these moment functions agree with the ones in equations \eqref{MomentFunctionsT3b} and \eqref{MomentFunctionsT3c}, where all the arguments were made explicit.
For $r=s+1$,
analogous to \eqref{MomentFunctionsT3b} and \eqref{MomentFunctionsT3c} for $T=3$, the two moment conditions in \eqref{MomentFunctionsTgeneralBC} for $T \geq 3$ can be obtained from 
\eqref{MomentFunctionsTgeneral} by setting $q_2=1$ and carefully evaluating the limit $\lambda_{0}  \rightarrow -\infty$
(after normalizing the value for $y_t \leq q_1$, $y_s>1$ to be $-1$),
or setting  $q_2=Q$ and taking the limit $\lambda_{Q}  \rightarrow \infty$. It is therefore 
appropriate to think of \eqref{MomentFunctionsTgeneral} as our master equation, which summarizes all
the moment conditions provided in this paper.
In \eqref{MomentFunctionsTgeneralBC} we can choose more general $r\geq s+1$,
but otherwise the structure of \eqref{MomentFunctionsTgeneralBC} can be derived from \eqref{MomentFunctionsTgeneral}.

It turns out that the moment functions with $r > s+1$ are not actually needed to span all possible valid moment functions of the dynamic ordered choice logit model (see our discussion of independence and completeness below). However, since implementation of these moment functions requires only that we observe three pairs $(y_{t-1},y_t)$, $(y_{s-1},y_s)$, $(y_{r-1},y_r)$ of consecutive outcomes, they may be empirically relevant for the case where observations for some time periods are (exogenously) missing.%
\footnote{For example, an estimator that allows for selection to be correlated with $(Y_{i0},X_i,A_i)$ can be constructed using the results on GMM estimation with incomplete data in \cite{muris_incomplete}.}
We also include $r > s+1$ in our discussion
here to ensure that our results in this paper
contain those for the dynamic binary choice logit model studied in \cite{honore2020dynamic} as a special case
--- notice that for $Q=2$ we always have $q_2=1$
or $q_2=Q$, that is, for the binary choice case all available moment functions are stated in \eqref{MomentFunctionsTgeneralBC}.

The following theorem establishes that the moment functions in \eqref{MomentFunctionsTgeneral} and \eqref{MomentFunctionsTgeneralBC} do indeed deliver valid moment conditions.

\begin{theorem}
\label{th:momentsTgeneralA}   
 If the outcomes $Y=(Y_1,\ldots,Y_T)$ are generated
from model \eqref{model} with $Q \geq 2$, $T\geq 3$ and true parameters $\theta^0 = (\beta^0,\gamma^0,\lambda^0)$, then we have for all 
$t,s,r \in \{1,2,\ldots,T\}$ with $t<s<r$, 
$y_0 \in \{1,\ldots,Q\}$,
$q_1,q_3 \in \{1,\ldots,Q-1\}$, $x\in \mathbb{R}^{K\times T}$, $\alpha
\in \mathbb{R}  \cup \{\pm \infty\}$, and $w:\{1,\ldots,Q\}^{t-1}\rightarrow \mathbb{R}$   that
\begin{align*}
\mathbb{E}\left[ w(Y_{1},\ldots ,Y_{t-1})\,m^{(t,s,s+1)}_{y_0,q_1,q_2,q_3}(Y,X,\theta^0)\,\big|\,Y_{0}=y_{0},\,X=x,\,A=\alpha \right] & =0,
\quad \text{for $q_2 \in \{2,\ldots,Q-1\}$,}
\\
\mathbb{E}\left[ w(Y_{1},\ldots ,Y_{t-1})\,m^{(t,s,r)}_{y_0,q_1,q_2,q_3}(Y,X,\theta)\,\big|\,Y_{0}=y_{0},\,X=x,\,A=\alpha \right] & =0,
\quad \text{for $q_2 \in \{1,Q\}$.}
\end{align*}
\end{theorem}
The proof is provided in the appendix. Notice that for $q_2 \in \{1,Q\}$ we can choose the time indices $t<s<r$ freely. 
By contrast, for $q_2 \in \{2,\ldots,Q-1\}$ we can only choose $t<s$ freely, but the third time index that appears in the definition
of the moment function needs to be equal to $s+1$, otherwise we do not obtain a valid moment function for those values of $q_2$.

This distinction between $q_2 \in \{1,Q\}$ and $q_2 \in \{2,\ldots,Q-1\}$ is also reflected in the proof of Theorem~\ref{th:momentsTgeneralA}. 
The moment functions in \eqref{MomentFunctionsTgeneralBC} for $q_2 \in \{1,Q\}$ 
only depend on $Y_1$, $Y_2$, $Y_3$ through the binarized variables
$\widetilde Y_1 =   \mathbbm{1}\left\{  Y_1 > q_1 \right\}$,
$\widetilde Y_2 =  \mathbbm{1}\left\{  Y_2 = q_2 \right\}$,
$\widetilde Y_3  =  \mathbbm{1}\left\{  Y_3 > q_3 \right\}$,
and the proof relies on Lemma~\ref{lemma:MASTER2} in the appendix, which
provides a general set of valid moment functions for such binary variables,
very closely related to the dynamic binary choice results in 
\cite{honore2020dynamic}.
By contrast, the moment functions in \eqref{MomentFunctionsTgeneral}
for $q_2 \in \{2,\ldots,Q-1\}$ cannot be expressed through binarized variables only,
because there the dependence on $Y_2$ requires distinguishing three cases
($Y_s < q_2 $, $Y_s = q_2$, $Y_s > q_2$). The proof, in this case, relies on
Lemma~\ref{lemma:MASTER} in the appendix which is completely novel to the
current paper. However, that proof strategy for $q_2 \in \{2,\ldots,Q-1\}$
does not work for $s>r+1$, and we have also numerically verified that our
moment conditions for $q_2 \in \{2,\ldots,Q-1\}$ indeed do not
generalize to $s>r+1$.

\subsubsection*{Conjecture on the completeness of the moment conditions}

Theorem~\ref{th:momentsTgeneralA} states that  the moment functions
in \eqref{MomentFunctionsTgeneral} and \eqref{MomentFunctionsTgeneralBC}
are valid, but it is natural to ask whether they are also linearly
independent, and whether they constitute all possible valid moment 
functions of the dynamic panel ordered logit model. 
We do not aim to formally prove such a linear independence and completeness result in this paper, and the 
following statement should accordingly be understood as a conjecture, which we have numerically confirmed for various combinations of $Q$ and $T$  and for
many different numerical values of the regressors and model parameters:

Let the outcomes $Y=(Y_1,\ldots,Y_T)$ be generated
from model \eqref{model} with $Q \geq 2$, $T\geq 3$, and let the true parameters
$\theta^0 = (\beta^0,\gamma^0,\lambda^0)$ be
such that $\gamma^0_{q_1} \neq \gamma^0_{q_2}$ for
all $q_1 \neq q_2$. For given $y_0 \in \{1,\ldots,Q\}$ and 
$x\in \mathbb{R}^{K\times T}$, let $m_{y_0}(y,x,\theta^0) \in \mathbb{R}$ be a
moment function 
that satisfies \eqref{MomentFunctionIdea} for all $\alpha \in \mathbb{R} \cup \{\pm \infty\}$. Our calculations suggest that there exist unique 
weights $w_{y_{0}}(q_1,q_2,q_3,s,y_{1},\ldots ,y_{t-1},x,\theta^0) \in \mathbb{R}$ such that for all $y \in \{1,\ldots,Q\}^T$ we have
\begin{align}
    m_{y_0}(y,x,\theta^0)  = 
    \sum_{q_1=1}^{Q-1}
     \sum_{q_2=1}^{Q}
      \sum_{q_3=1}^{Q-1}
    \sum_{t=1}^{T-2} \sum_{s=t+1}^{T-1}
  w_{y_{0}}(q_1,q_2,q_3,t,s,y_{1},\ldots ,y_{t-1},x,\theta^0)  \;
    m^{(t,s,s+1)}_{y_0,q_1,q_2,q_3}(y,x,\theta^0),
    \label{completeness}
\end{align}
where $m^{(t,s,s+1)}_{y_0,q_1,q_2,q_3}(y,x,\theta^0)$ 
are the moment functions defined in \eqref{MomentFunctionsTgeneral} and \eqref{MomentFunctionsTgeneralBC}. In other words, 
we conjecture that
every valid 
moment condition in this model is a unique linear combination of the moment conditions in
Theorem~\ref{th:momentsTgeneralA} with $r=s+1$. Notice that
the uniqueness of the linear combination implies that the moment
functions involved in this linear combination are linearly independent.

In equation \eqref{completeness}, the function 
$m^{(t,s,s+1)}_{y_0,q_1,q_2,q_3}(y,x,\theta^0)$ is
multiplied with an arbitrary function of $y_1,\ldots,y_{t-1}$.
Those functions of $y_1,\ldots,y_{t-1}$ constitute a
$Q^{t-1}$ dimensional space. Thus, \eqref{completeness} suggests that the total number of available moment conditions for each value of the covariates $x$ and initial conditions $y_0$ is equal to
\begin{align}
    \ell &= \sum_{q_1=1}^{Q-1}
     \sum_{q_2=1}^{Q}
      \sum_{q_3=1}^{Q-1}
    \sum_{t=1}^{T-2} \sum_{s=t+1}^{T-1} Q^{t-1}
    = (Q-1) \, Q\, (Q-1) \sum_{t=1}^{T-2} \, (T-t-1) \, Q^{t-1}
  \nonumber  \\
    &= Q^T - (T-1) \, Q^2 + (T-2) \, Q .
    \label{NumberMoment}
\end{align}
As explained in Section~\ref{subsec:momentApproach}, the function $m_{y_0}(\cdot,x,\theta^0) : \{1,\ldots,Q\}^T \rightarrow \mathbb{R}$ is a vector in a $Q^T$ dimensional space.
The condition \eqref{MomentFunctionIdea}, for all $\alpha$, 
then imposes  $Q^T - \ell = (T-1) \, Q^2 + (T-2) \, Q$ linear restrictions on this vector, leaving an $\ell$-dimensional linear subspace 
of valid moment functions, a basis representation
of which is given by
\eqref{completeness}. For fixed values
of $y_0$, $x$, $\theta_0$, $T$, $Q$, one can numerically verify 
the dimension of the solution space of the system of linear equations \eqref{MomentFunctionIdea}, and thereby check 
\eqref{NumberMoment} numerically. In Appendix~\ref{app:LowerBound} we furthermore show that 
the total number of linearly independent 
conditional moment conditions for our model
is at least 
the number obtained in \eqref{app:LowerBound},
but that argument in the appendix still allows for the possibility
that there could be more, although we do not believe that there are.

The condition 
$\gamma^0_{q_1} \neq \gamma^0_{q_2}$ for
all $q_1 \neq q_2$ is important for this result. For example,
if all the $\gamma^0_q$ are the same, then the parameter $\gamma^0$ can be absorbed into
the fixed effects, and we are left with a static ordered logit model
as in \cite{Muris2017}, for which
one finds an additional $(T-1) (Q-1)^2$ moment conditions to be available, bringing the total
number of linearly independent valid moment conditions (for each value of covariates and parameters)  in the static model
to
$\ell = Q^T -T (Q-1) - 1$.

We reiterate that the discussion of linear independence and completeness 
 of the moment functions presented above are  conjectures  which we
do not aim to prove in this paper. A proof for the special case $Q=2$ (dynamic binary choice logit models) is provided in \cite{kruiniger2020further}
and  \cite{dobronyi2021identification}.
 We also note that the counting of moment conditions as above  does not consider whether the resulting moment conditions actually contain
 information about (all) the parameters $\theta$. Some of the valid moment
 functions may not depend on (all of) those model parameters. Identification of the model
 parameters through the moment conditions is discussed in Section~\ref{sec:Identification}.

\subsection{More general regressors} 
\label{sec:GeneralModel}

The model probabilities in \eqref{DefProb} and the moment functions in \eqref{MomentFunctionsTgeneral}
and \eqref{MomentFunctionsTgeneralBC}
only depend on the regressors and the parameters $\beta$ and $\gamma$
through the single index $z_t=z(y_{t-1},x_{t},\theta)$.\footnote{
As written, the moment condition in \eqref{MomentFunctionsTgeneral} 
depends explicitly on the model parameter $\gamma$  for 
the case that $y_t \leq q_1$ and $y_s>q_2$. However, that is a notational
artefact, because in that line of the moment condition we could have 
written 
$ \exp\left[  z(y_{t-1},x_{t},\theta) - z(q_2,x_{s+1},\theta) + \lambda_{q_3,q_1}   \right]$
instead of
$ \exp\left(   z_{t,s+1}  + \gamma_{y_s,q_2} + \lambda_{q_3,q_1} \right)$; that is, the explicit
 dependence on $\gamma$ can be fully absorbed into the single index, but one needs to evaluate $z_{s+1}=z(y_s,x_{s+1},\theta)$ at $q_2$ instead of $y_s$. 
}
So far, we have only explicitly discussed the linear specification in \eqref{DefSingleIndex} for this single index,
but Theorem~\ref{th:momentsTgeneralA} is valid independently of the functional form 
of $z(y_{t-1},x_{t},\theta)$.\footnote{%
The parameters $\lambda$ can also be absorbed into the single index. One just needs to 
define $\widetilde z_q(y_{t-1},x_{t},\theta) := z(y_{t-1},x_{t},\theta)  - \lambda_{q} $
and rewrite \eqref{DefProb} as
\begin{equation*}
p_{y_{0}}(y_{},x_{},\theta ,\alpha _{}) =\prod_{t=1}^{T}
\left\{  \Lambda\Big[ \widetilde{z}_{y_{t}-1}(y_{t-1},x_{t},\theta) + \alpha  \Big] -  \Lambda\Big[  \widetilde{z}_{y_{t}}(y_{t-1},x_{t},\theta) + \alpha    \Big]
\right\} .
\end{equation*}
The moment functions in \eqref{MomentFunctionsTgeneral}
and \eqref{MomentFunctionsTgeneralBC} then remain valid for arbitrary functional forms 
of $\widetilde z_q(y_{t-1},x_{t},\theta)$. We just need to replace $z_t - \lambda_{q_1}$, $z_s - \lambda_{q_2}$,
and $z_r - \lambda_{q_3}$ (with $r=s+1$ in \eqref{MomentFunctionsTgeneral}) by
$\widetilde z_{q_1}(y_{t-1},x_{t},\theta)$, $\widetilde z_{q_2}(y_{s-1},x_{s},\theta)$,
and $\widetilde z_{q_3}(y_{r-1},x_{r},\theta)$, respectively. The proof of Theorem~\ref{th:momentsTgeneralA} remains 
valid under that replacement.
}
In other words, if we replace the latent variable specification 
in \eqref{ModelYstar} by
\begin{align*}
    Y^*_{it} = z \left( Y_{i,t-1},X_{it},\theta \right)
    + A_i + \varepsilon_{it} 
\end{align*}
for an arbitrary function $z(\cdot,\cdot,\cdot)$, then 
the moment functions~\eqref{MomentFunctionsT3}, \eqref{MomentFunctionsT3b},
\eqref{MomentFunctionsT3c}, and Theorem~\ref{th:momentsTgeneralA} remain 
fully valid.

We believe that the linear specification in \eqref{ModelYstar} is the most relevant in practice, but one could certainly consider other 
specifications as well. In particular, it is possible to include regressors that are interactions between the
observed regressors
and the lagged dependent variable:
\begin{align}
    z(Y_{i,t-1},X_{it},\theta) :=  X_{it}^{\prime }\,\beta  +  \sum_{q=1}^{Q} \gamma_q  
    \left[ \mathbbm{1}\left\{ Y_{i,t-1} = q \right\} 
       +   \mathbbm{1}\left\{ Y_{i,t-1} = q \right\}  X_{it}^{\prime }\,\delta_q \right],
    \label{AlternativeSpecification}   
\end{align}    
where $\delta_q \in \mathbb{R}^K$ are the additional unknown parameters to be included in $\theta$. This specification allows
the effect of the regressors $X_{it}$ on the outcome $Y_{it}$ to be arbitrarily dependent on the current state $Y_{i,t-1}$.
While a GMM estimator based on moment functions developed in this paper could be employed in applications with the more general state dependence as in \eqref{AlternativeSpecification}, we  do not consider these more general models further.

\section{Identification}
\label{sec:Identification}

This section presents identification results for the parameters $\theta = (\beta,\gamma,\lambda)$ based on the moment conditions for $T=3$
in Theorem~\ref{th:momentsT3}. 
All  results in this section impose the following model assumption.
\newtheorem{IDassumption}{Assumption}
           \renewcommand{\theIDassumption}{ID}
\begin{IDassumption}
    \label{ass:ID}
   The outcomes $Y=(Y_1,Y_2,Y_3)$ are generated
from model \eqref{model} with $z(\cdot,\cdot,\cdot)$
defined in \eqref{DefSingleIndex}, $Q \geq 2$,
$T=3$, and true parameters $\theta^0 = (\beta^0,\gamma^0,\lambda^0)$.
Furthermore, for all $y_0 \in \{1,\ldots,Q\}$,
there exists a non-empty set ${\cal X}^{\rm reg}_{y_0} \subset \mathbb{R}^{K\times 3}$ such that
for all $x \in  {\cal X}^{\rm reg}_{y_0}$, the conditional probability 
   $ {\rm Pr}(A \in \{ \pm \infty\} \mid  Y_0=y_0, \, X=x )$ is well-defined and smaller than one.
\end{IDassumption}

We  impose the assumption   $ {\rm Pr}(A \in \{ \pm \infty\} \mid  Y_0=y_0, \, X=x ) < 1$  for some $x$  in order to ensure that
the model probabilities in \eqref{DefProb} are strictly positive for all possible outcomes.  If $ {\rm Pr}(A \in \{ \pm \infty\} \mid  Y_0=y_0, \, X=x ) = 1$ 
for all $x$,
then only the outcomes $Y_t=1$ and $Y_t=Q$ would be generated by the model. A violation of this assumption on the fixed effects $A$ would 
therefore be readily observable from the data. 
All the propositions below also impose that $X \in {\cal X}^{\rm reg}_{y_0}$
occurs with non-zero probability.

The aim is to  identify the parameter vector $\theta^0$ from the distribution of $Y$ conditional on $Y_0$ and $X$ under Assumption~\ref{ass:ID}.
The model for that conditional distribution is semi-parametric:
The distribution of $Y$ conditional on  $Y_0$, $X$,
and $A$ is specified parametrically, but only weak regularity conditions are
imposed on the unknown distribution of $A$ conditional on $Y_0$ and $X$.
The main challenge in the identification problem is how 
to deal with the unspecified conditional distribution of $A$, which is
an infinite-dimensional component of the parameter space of the model. Fortunately, the moment conditions in Theorem~\ref{th:momentsT3} already partly
solve this challenge, because they give us implications of the model
that do not depend on  $A$. The remaining question is whether $\theta^0$ is point-identified from those moment conditions.

\subsubsection*{Identification of $\gamma$}

In order to identify the parameters $\gamma=(\gamma_1,\ldots,\gamma_Q)$  up to   normalization, we condition on the event $X_1=X_2=X_3$.
For $x=(x_1,x_1,x_1)$ and $q_1=q_2=q_3=1$, the moment function in \eqref{MomentFunctionsT3b} reads
\begin{align}
   m_{y_0}(y,\gamma) := \exp(\gamma_{y_0} )  \, m_{y_0,1,1,1}(y,x,\theta)  &= 
    \left\{  \begin{array}{ll}
             - \exp(\gamma_{y_0}) & \text{if }   y_1 = 1, \, y_2>1,\\
                 \exp\left(  \gamma_{1}   \right)  & \text{if }  y_1 > 1, \, y_2=1 , \,y_3= 1,\\
                \exp\left(   \gamma_{y_1}      \right)  & \text{if } y_1 > 1, \, y_2=1 , \, y_3>1,\\
             0 & \text{otherwise}.
       \end{array} \right.
    \label{MomentFunctionID}   
\end{align}
Theorem~\ref{th:momentsT3} implies that 
$\mathbb{E} \left[  m_{y_0}(Y,\gamma^0) \, \big| \, Y_0=y_0, \, X=(x_1,x_1,x_1) \right] = 0$.
The following lemma states that these moment conditions are sufficient to uniquely identify $\gamma$ up to a normalization.

\begin{proposition}
     \label{prop:IDgamma}
     Let Assumption~\ref{ass:ID} hold, and let $x_1 \in \mathbb{R}$ be such that 
     $${\rm Pr}\left( 
     Y_0=y_0
     \; \& \; 
     X \in {\cal X}^{\rm reg}_{y_0} 
     \; \& \;
     \left\| X- (x_1,x_1,x_1) \right\| \leq \epsilon    \right) >0 
     \qquad \text{for all  $y_0 \in \{1,\ldots,Q\}$ and $\epsilon>0$.}
     $$
     
     Then, if $\gamma \in \mathbb{R}^Q$ satisfies\footnote{
      Here, we implicitly assume that 
      $ \mathbb{E} \left[  m_{y_0}(Y,\gamma) \, \big| \, Y_0=y_0, \, X=(x_1,x_1,x_1) \right]$ is uniquely 
      defined. 
      This can be guaranteed, for example,
      by demanding that this conditional expectation is
      continuous in $x_1$.
     }
     \begin{align}
     \mathbb{E} \left[  m_{y_0}(Y,\gamma) \, \big| \, Y_0=y_0, \, X=(x_1,x_1,x_1) \right] = 0 
     \qquad \text{for all $y_0 \in \{1,\ldots,Q\}$},
        \label{MomentConditionIDprop1}
     \end{align}
     for $m_{y_0}(y,\gamma)$ as defined in \eqref{MomentFunctionID},
      we have $\gamma = \gamma^0 + c$ for some $c \in \mathbb{R}$.
      Thus,
     if we normalize $\gamma^0_1=0$, then $ \gamma^0$ is uniquely identified from the data.
\end{proposition}

The proof is given in the appendix. This identification result requires observed data for every initial condition  $y_0 \in \{1,\ldots,Q\}$.
If this is not available, but we observe $T=4$ time periods of data after the initial condition, then we can instead apply 
Proposition~\ref{prop:IDgamma} to the data shifted by one time period.

In addition to Assumption~\ref{ass:ID}, the proposition 
demands that covariate values $X \in {\cal X}^{\rm reg}_{y_0}$ 
in any $\epsilon$-ball around $(x_1,x_1,x_1)$ occur with positive probability. This condition, in particular, guarantees that 
the conditional expectation in \eqref{MomentConditionIDprop1} 
is well-defined, and that conditional on
$X=(x_1,x_1,x_1)$ the event $A \in \{ \pm \infty\}$ occurs
with probability less than one for every value of the initial condition
$Y_0$.

\subsubsection*{Identification of $\beta$}

Taking the identification result for $\gamma$ as given, we now turn to the problem of identifying $\beta$.
We again
consider the moment function in \eqref{MomentFunctionsT3b}
with $q_1=q_2=q_3=1$, but now for general regressor values 
\begin{align}
 m_{y_0,1,1,1}(y,x,\beta,\gamma) 
 :=
  m_{y_0,1,1,1}(y,x,\theta)  & = 
    \left\{  \begin{array}{ll}
             \exp\left(  x_{23}' \, \beta     \right)     -  1 & \text{if }  y_1 = 1, \, y_2=1 , \, y_3>1, \\
             - 1 & \text{if }   y_1 = 1, \, y_2>1,\\
                 \exp\left( x_{31}' \, \beta + \gamma_1 - \gamma_{y_0}   \right)  & \text{if }  y_1 > 1, \, y_2=1 , \,y_3 = 1,\\
                \exp\left( x_{21}' \, \beta + \gamma_{y_1} - \gamma_{y_0}    \right)  & \text{if } y_1 > 1, \, y_2=1 , \, y_3>1,\\
             0 & \text{otherwise}.
       \end{array} \right.
     \label{momIdentify}  
\end{align}
For $k\in \{1,\ldots ,K\}$ we define 
\begin{align*}
\mathcal{X}_{k,+}& :=\{x\in  {\cal X}^{\rm reg}_{y_0} \,:\,x_{k,1}\leq
x_{k,3}<x_{k,2}\;\;\text{or}\;\;x_{k,1}<x_{k,3}\leq x_{k,2}\}, \\
\mathcal{X}_{k,-}& :=\{x\in  {\cal X}^{\rm reg}_{y_0}\,:\,x_{k,1}\geq
x_{k,3}>x_{k,2}\;\;\text{or}\;\;x_{k,1}>x_{k,3}\geq x_{k,2}\}.
\end{align*}
Here, the set $\mathcal{X}_{k,+}$ is the
set of possible regressor values $x\in \mathbb{R}^{K\times 3}$ such that
$x_{k,1}\leq x_{k,3}\leq x_{k,2}$ with at least one of the inequalities being strict.
For the set $\mathcal{X}_{k,-}$ those inequalities are reversed.
Therefore, if $x \in \mathcal{X}_{k,+}$, then $m_{y_0,1,1,1}(y,x,\beta,\gamma)$ is strictly 
increasing in $\beta_k$,
and if $x \in \mathcal{X}_{k,-}$, then $m_{y_0,1,1,1}(y,x,\beta,\gamma)$ is strictly 
decreasing in $\beta_k$.

For any vector $s\in \{-,+\}^{K}$, we furthermore define the set $\mathcal{X}_{s}=\bigcap_{k\in
\{1,\ldots ,K\}}\mathcal{X}_{k,s_{k}}$. If $x \in \mathcal{X}_{s}$, then for all $k\in
\{1,\ldots ,K\}$ we have that $\beta_k$ is strictly increasing (or strictly decreasing) in $m_{y_0,1,1,1}(y,x,\beta,\gamma)$
if $s_k=+$ (or $s_k=-$). These monotonicity properties allow us to uniquely identify $\beta$
from the moment conditions $\mathbb{E}\left[ m_{y_0,1,1,1}(Y,X,\beta^0,\gamma^0)\,\Big|\,Y_{0}=y_{0},\;X\in \mathcal{X}_{s}\right] = 0$,
which are valid moment conditions according to Theorem~\ref{th:momentsT3}. The following proposition formalizes this.

\begin{proposition}
     \label{prop:IDbeta}
     Let Assumption~\ref{ass:ID} hold and let $y_0 \in  \{1,\ldots,Q\}$ be such that
     $${\rm Pr}\left( 
     Y_0=y_0
     \; \& \;
     X\in \mathcal{X}_{s}  \right) >0 
     \qquad \text{for all $s \in  \{-,+\}^{K}$ with $s_K=+$.}
     $$
     Then, if $\beta \in \mathbb{R}^K$ satisfies
     \begin{align}
    \mathbb{E}\left[ m_{y_0,1,1,1}(Y,X,\beta,\gamma^0)\,\Big|\,Y_{0}=y_{0},\;X\in \mathcal{X}_{s}\right] = 0 
     \qquad \text{for all $s \in  \{-,+\}^{K}$ with $s_K=+$,}
     \label{Condition:prop:IDbeta}
     \end{align}
     we have $\beta= \beta^0$.
    Thus, since $\gamma^0$ is already identified from Proposition~\ref{prop:IDgamma},
    we find that $\beta^0$ is also  uniquely identified from the data.
\end{proposition}

The proof is given in the appendix. 
Again, in addition
to Assumption~\ref{ass:ID}, the additional
condition in Proposition~\ref{prop:IDbeta}
simply guarantees that
the conditional expectation 
in \eqref{Condition:prop:IDbeta} is well-defined.

\subsubsection*{Identification of $\lambda$}

Having identified $\gamma$ and $\beta$, we now turn to the problem of identifying $\lambda$  up to a normalization.
The moment function in \eqref{MomentFunctionsT3b}
with $q_2=q_3=1$ and   $q_1 \in \{2,\ldots,Q-1\}$ can be written as
\begin{align}
m_{y_0,q_1,1,1}(y,x,\beta,\gamma,\lambda) =
    \left\{  \begin{array}{ll}
             \exp\left(  x_{23}' \, \beta + \gamma_{y_1,1}    \right)     -  1 & \text{if }  y_1 \leq q_1, \, y_2=1 , \, y_3>1, \\
             - 1 & \text{if }   y_1 \leq q_1, \, y_2>1,\\
                 \exp\left( x_{31}' \, \beta + \gamma_{1,y_0} + \lambda_{q_1} - \lambda_1 \right)  & \text{if }  y_1 > q_1, \, y_2=1 , \,y_3 =1,\\
                \exp\left( x_{21}' \, \beta + \gamma_{y_1,y_0} + \lambda_{q_1} - \lambda_1  \right)  & \text{if } y_1 > q_1, \, y_2=1 , \, y_3>1,\\
             0 & \text{otherwise}.
       \end{array} \right.
       \label{Moments:prop:IDlambda}
\end{align} 
The expected value of this moment function 
only depends on $\lambda$ through $\lambda_{q_1} - \lambda_1$,
and is strictly increasing in $\lambda_{q_1} - \lambda_1$.
This implies that this moment function identifies $\lambda_{q_1} - \lambda_1$ uniquely. 
By applying this argument to all $q_1  \in \{2,\ldots,Q-1\}$,
we can therefore identify $\lambda$ up to an additive constant.
This is summarized in the following proposition.
 
\begin{proposition}
     \label{prop:IDlambda}
     Let Assumption~\ref{ass:ID} hold. Let $y_0 \in  \{1,\ldots,Q\}$ be such that
     ${\rm Pr}\big(
     Y_0=y_0 \; \& \;
     X\in {\cal X}^{\rm reg}_{y_0} \big) > 0$.
     Then, if $\lambda$ satisfies
     $$
    \mathbb{E}\left[ m_{y_0,q_1,1,1}(Y,X,\beta^0,\gamma^0,\lambda) \,\Big|\,Y_{0}=y_{0} \right] = 0 
     \qquad \text{for all $q_1  \in \{2,\ldots,Q-1\}$,}
     $$
     we have $\lambda = \lambda^0 + c$ for some $c \in \mathbb{R}$.
    Thus, if we normalize $\lambda^0_1=0$,
    and    
    since $\gamma^0$ and $\beta^0$ are already identified from Proposition~\ref{prop:IDgamma} and \ref{prop:IDbeta},
    we find that $\lambda^0$ is also  uniquely identified from the data.
\end{proposition}

The proof is given in the appendix.

Combining Proposition~\ref{prop:IDgamma}, \ref{prop:IDbeta},
and \ref{prop:IDlambda}, we find that $\theta^0$ is uniquely identified 
from the data. Under the regularity conditions
of those propositions, we can recover $\theta^0 = (\beta^0,\gamma^0,\lambda^0)$ uniquely from the distribution of $Y$ conditional on $Y_0$ and $X$.

Our identification arguments in this section are constructive.
However, they condition on special values of the regressors. In particular,
Proposition~\ref{prop:IDgamma} conditions on the event $X_1=X_2=X_3$, 
which is a zero-probability event if $X$ is continuously distributed 
(and may happen rarely even for discrete $X$). An estimator based on
the identification strategy in this section would therefore in general
be quite inefficient. Hence, in our Monte Carlo simulations and empirical
application, we construct more general GMM estimators based
on our moment conditions.

\section{Implication for estimation and specification testing\label{Section:
Implication for estimation}}

The moment conditions in Section\ \ref{sec:model} are conditional on the
initial condition $Y_{i0}$ and the strictly exogenous explanatory variables $%
X_{i}$. It is tempting to try to mimic the identification argument in
Section \ref{sec:Identification} in order to turn these moment conditions
into an estimator. The problem with such an approach is that the
conditioning set in Proposition \ref{prop:IDgamma} will often have
probability 0. Alternatively, one can form a set of unconditional moment
functions by constructing
\begin{equation*}
M(Y_{i0},Y_{i},X_{i},\beta ,\gamma ,\lambda )=g\left( Y_{i0},X_{i}\right)
\otimes m_{Y_{i0}}\left( Y_{i},X_{i},\beta ,\gamma ,\lambda \right)
\end{equation*}%
where the vector-valued function, $m_{Y_{i0}}$, is composed of linear
combinations of the moment functions in \eqref{MomentFunctionsT3}, %
\eqref{MomentFunctionsT3b}, and \eqref{MomentFunctionsT3c}, and $g$ is a
vector-valued function of the initial condition $Y_{i0}$ and the strictly
exogenous $X_{i}$. Let $\theta =\left( \beta ^{\prime },\gamma ^{\prime
},\lambda ^{\prime }\right) ^{\prime }$. A generalized method of moments
(GMM) estimator can then be defined by\footnote{%
As mentioned in Section \ref{sec:model}, it is necessary to normalize one of
the $Q$ elements of $\gamma $ and one of the $Q-1$ elements of $\lambda $.}
\begin{multline*}
\widehat{\theta }=%
\begin{pmatrix}
\widehat{\beta } \\
\widehat{\gamma } \\
\widehat{\lambda }%
\end{pmatrix}%
=\limfunc{argmin}_{\beta \in \mathbb{R}^{K},\,\gamma \in \mathbb{R}%
^{Q-1},\,\lambda \in \mathbb{R}^{Q-2}}\left(
\sum_{i=1}^{n}M(Y_{i0},Y_{i},X_{i},\beta ,\gamma ,\lambda )\right) ^{\prime }
\\
\widehat{W}_{n}\left( \sum_{i=1}^{n}M(Y_{i0},Y_{i},X_{i},\beta ,\gamma
,\lambda )\right) ,
\end{multline*}%
where the weighting matrix $\widehat{W}_{n}$ converges to a positive
definite matrix, $W_{0}$. Assuming that $\mathbb{E}\left[
M(Y_{i0},Y_{i},X_{i},\theta )\right] =0$ is \textsl{uniquely} satisfied at $%
\theta =\theta ^{0}$, and that mild regularity conditions (see %
\citealt{Hansen1982}) are satisfied, $\widehat{\theta }$ will be consistent
and asymptotically normally distributed.

One limitation of the GMM\ approach is that it is often difficult to know
whether the moment condition $\mathbb{E}\left[ M(Y_{i0},Y_{i},X_{i},\theta )%
\right] =0$ is uniquely satisfied at the true parameter value. When the
strictly exogenous explanatory variables, $X_{i}$, are discrete, sufficient
conditions for this can be obtained from the identification results in
Section~\ref{sec:Identification} by defining $g\left( Y_{i0},X_{i}\right) $
to be a vector of indicator functions for values in the support of $\left(
Y_{i0},X_{i}\right) $. If $X_{i}$ is not discrete, it may be possible to
define a root-$n$ consistent estimator by combining nonparametrically
estimated conditional moment conditions with the unconditional moment
conditions. See, for example, \cite{HonoreHu2004a} for such an approach.
Whether or not $\mathbb{E}\left[ M(Y_{i0},Y_{i},X_{i},\theta )\right] =0$ is
uniquely satisfied at the true parameter value, one can calculate valid
confidence sets for $\theta _{0}$ based on moment conditions like $\mathbb{E}%
\left[ M(Y_{i0},Y_{i},X_{i},\theta )\right] =0$. See, for example, \cite%
{ChenChristensenTamer2018}.

A second limitation of the GMM\ approach is that even if one ignores the
issue of identification, there are many ways to form a finite set of
unconditional moment conditions from the expressions in %
\eqref{MomentFunctionsT3}, \eqref{MomentFunctionsT3b}, and %
\eqref{MomentFunctionsT3c}. Moreover, the most natural ad hoc ways to do this, such as 
considering all interactions between the conditional moment and the explanatory 
variables and the initial condition, can lead to a very large number of moment conditions, 
which in turn can result in poor small sample performance. 
It is in principle known how to most efficiently
turn a set of conditional moment conditions into a set of moment conditions
of the same dimensionality as the parameter to be estimated. See, for
example, the discussion in \cite{NeweyMcFadden94:HoE}. Specifically, with a
conditional moment condition $\mathbb{E}\left[ \left. m_{Y_{0}}\left(
Y,X,\theta \right) \right\vert X,Y_{0}\right] =0$ when $\theta $ takes its
true value, $\theta _{0}$, the optimal unconditional moment function is $%
A\left( X,Y_{0}\right) m_{Y_{0}}\left( Y,X,\theta \right) $, where $A\left(
X,Y_{0}\right) =\mathbb{E}\left[ \left. \nabla _{\theta }m_{Y_{0}}\left(
Y,X,\theta _{0}\right) \right\vert X,Y_{0}\right] ^{\prime }V\left[ \left.
m_{Y_{0}}\left( Y,X,\theta _{0}\right) \right\vert X,Y_{0}\right] ^{-1}$.
Unfortunately, the construction of estimators of these efficient moments
depends heavily on the distribution of $Y$ given $\left( X,Y_{0}\right) $.
On the other hand, the moment conditions are still valid if $A\left(
X,Y_{0}\right) $ is misspecified. One approach therefore is to estimate a
flexible reduced form model for the distribution of $Y$ given $\left(
X,Y_{0}\right) $, and then use this reduced form for the distribution of $Y$
given $X$ to construct an estimate of $A\left( X,Y_{0}\right) $. In the
simulations and the empirical illustration below, we take this approach
using a correlated random effects approach to obtain the reduced form for
the distribution of $Y$ given $\left( X,Y_{0}\right) $. The dimensionality
of the moment function $A\left( X,Y_{0}\right) m_{Y_{0}}\left( Y,X,\theta
\right) $ is the same as that of the parameter vector, and the asymptotic
distribution of the estimator therefore follows from the theory of nonlinear
method of moments estimators.

The moment conditions derived in this paper can also be used for
specification testing. Suppose that a researcher has estimated the
parameters of interest, $\theta _{0}=\left( \beta _{0},\gamma _{0},\lambda
_{0}\right) $, by an estimator, $\widehat{\theta }$, that solves a moment
condition of the type $\frac{1}{n}\sum_{i=1}^{n}\psi \left( Y_{i},X_{i},%
\widehat{\theta }\right) =0$. For example, she might have estimated a model
without individual-specific heterogeneity or a model in which the
heterogeneity is captured parametrically by a random effects approach, and
she might be interested in testing her parametric assumptions against the
less parametric fixed effects model. Let $\widehat{M}=\frac{1}{n}%
\sum_{i=1}^{n}M\left( Y_{i},X_{i},\widehat{\theta }\right) $ where $M$ is
defined as above. $\widehat{M}$ is then a standard two-step estimator, and
it is straightforward to test whether $\widehat{M}$ is statistically
different from 0.

\section{Practical performance of a method of moments estimator}
\label{PracticalPerformance}

In the next subsection, we present the results from a small Monte Carlo
experiment designed to illustrate the performance of the method of moments
estimator based on the discussion in Section \ref{Section: Implication for
estimation}, and we compare the performance of the estimator to its
asymptotic distribution as well as to a correlated random effects estimator.
We then illustrate the use of the method of moments estimator in an
empirical example.  

\subsection{Monte Carlo illustration\label{Section: Monte Carlo Results}}

We illustrate the performance of the GMM\ estimator described above through
a Monte Carlo study that considers three data generating processes,
two with a fixed effect and one without a fixed effect. The two data
generating processes that include a fixed effect are chosen such that one
satisfies the assumptions underpinning the correlated random effects
estimator proposed by \cite{wooldridge2005}, while the other does not. We consider
sample sizes of $N=500$, $1000$, and $2000$ with five time periods for each
individual. This includes the initial observations, so $T=4$ using the
notation above. There are $k=3$ explanatory variables and the dependent
variable can take $Q=4$ values. The true parameters are $\beta =\left(
1,0,0\right) ^{\prime },$ $\gamma =\left( -1,0,0,1\right) ^{\prime }$ and $%
\lambda =\left( -2,0,2\right) ^{\prime }$ and we normalize $\gamma _{2}=\lambda
_{2}=0$.

The explanatory variables are drawn as follows. First, let $\tilde{A}_{i}$
be a discrete random variable with $E\left[ \tilde{A}_{i}\right] =0$ and $V%
\left[ \tilde{A}_{i}\right] =3$. The exact distribution of $\tilde{A}_{i}$
differs across specifications.\ Secondly, let $Z_{ijt}$ ($j=1,...,k$, $%
t=0,...,4$) be independent normal random variables with mean 0 and variance
3, and let the first explanatory variable be $X_{i1t}=\left( Z_{i1t}+\tilde{A%
}_{i}\right) /\sqrt{2}$. The second through $k$'th explanatory variables are
given by $X_{ijt}=\left. \left( Z_{ijt}+X_{i1t}\right) \right/ \sqrt{2}$.
This implies that all the explanatory variables and $%
X_{it}^{\prime }\beta $ have mean 0 and variance 3. This is
comparable to the magnitude of the logistic distribution, which has mean 0
and variance $\left. \pi ^{2}\right/ 3$. \ 

For the two data generating
processes with a fixed effect, one (Design B) has $\tilde{A}_{i}$ normally
distributed while the other\ (Design C) has $P\left( \tilde{A}_{i}=\sqrt{6}%
\right) =\frac{1}{3}$ and $P\left( \tilde{A}_{i}=-\left. \sqrt{6}\right/
2\right) =\frac{2}{3}$. For both of these specifications, the fixed effect, $%
A_{i}$, equals $\tilde{A}_{i}$. The data generating process without fixed
effects (Design A) has the same distribution of $\tilde{A}_{i}$ as Design C,
but here $A_{i}=0$. 
For Design C, the initial dependent variable, $Y_{i0}$
is generated from the ordered logit model where the only explanatory
variable is $A_{i}$. For Designs A and B, $Y_{i0}$ is generated from the
ordered logit model without explanatory variables. From a fixed effects
perspective, this makes Design B\ a little special, but it makes it fit the
assumptions for the correlated random effects approach. Design A\ is without
fixed effects, so it also satisfies the assumptions for the correlated
random effects approach, but in this case the true parameter value of one of
the parameters (the variance of the error in the specification of the fixed
effect) is at the boundary of the parameter space, which could render standard inference problematic.

We perform 400 Monte Carlo replications. The results are presented in Tables %
\ref{Design A}, \ref{Design B}\ and \ref{Design C}. For comparison, we also
include the results for the correlated random effects estimator (cf. %
\citealt{wooldridge2005}) that specifies the distribution of the unobserved
heterogeneity as $A=\sum_{t=0}^{4}X_{t}^{\prime }\theta \left( t\right)
+\sum_{q=1}^{4}1\left\{ Y_{0}=q\right\} \theta \left( q\right) +\sigma Z$,
where $Z\sim N\left( 0,1\right) $. Design C violates the implicit assumption behind the correlated random effects
approach. Most importantly, the distribution of $A_{i}$ is discrete. On the
other hand, the relationship between the explanatory variables and $A_{i}$
is linear, so the violation is not extreme.

For each design and for each sample size, Tables \ref{Design A}, \ref{Design
B}\ and \ref{Design C} report the true values of the parameters, the median
bias of the method of moments estimator and the correlated random effects
estimator, the interquartile range of the estimators, and the median
absolute errors of the estimators. For each parameter, we also report the
ratio of the median absolute error of the correlated random effects
estimator relative to the method of moment estimator. Values of this ratio
greater than one suggest that the method of moments estimator is more
precise than the correlated random effects estimator. In Designs A and B,
the correlated random effects estimator is the correctly specified maximum
likelihood estimator. It is therefore not surprising that the median
absolute error ratio is less than one for all parameters and all sample
sized in this case. On the other hand, the ratio is above 0.65 in all cases,
suggesting that the loss of efficiency from the method of moments estimator
is not too large for this design. For these designs, both
estimators appear to be close to median unbiased, and the relative
performance of the estimators is driven by the difference in their
variability. For Design C with non-normal heterogeneity, the relative
performance of the two estimators is different for the different parameters.
The correlated random effects estimator is always less variable in terms of
interquartile range, but the biases in the estimates of the $\gamma $'s and $%
\delta $'s are large enough that the method of moments estimator tends to be
more precise when the sample size is large. For the coefficients on the
explanatory variables, $\beta $, the correlated random effects estimator is almost unbiased in Design C despite the
misspecification of the model for the fixed effect. This makes sense, because the specified model for the unobserved
heterogeneity will tend to control for any linear dependence between
explanatory variables and the level of the fixed effect. The specific results
for each of the $\gamma $'s and for each of the $\lambda $'s should be
considered with some care since the calculations are done under the specific
normalization that $\gamma _{2}=0$ and $\lambda _{2}=0$. With different
normalizations, the pattern of the results would have been different.
However, it is clear from Table \ref{Design C} that in the design with
heterogeneity, the misspecification embedded in the correlated random
effects approach generally speaking leads to biased estimates of the $\gamma
$'s and the $\lambda $, and that these biases will make the method of
moments estimator more precise for large sample sizes.

\begin{table}[tbp]
\caption{Design A}
\label{Design A}
\begin{center}
{\footnotesize \begin{tabular}{lrrrrrrrrrr} 
\multicolumn{11}{c}{Correlated Random Effects.  Sample Size: 500. (nrep = 400)} \\ 
& \multicolumn{1}{c}{$\beta _{1}$} & \multicolumn{1}{c}{$\beta _{2}$} &
  \multicolumn{1}{c}{$\beta _{3}$} & \multicolumn{1}{c}{$\gamma _{1}$} &
  \multicolumn{1}{c}{$\gamma _{2}$} & \multicolumn{1}{c}{$\gamma _{3}$} &
  \multicolumn{1}{c}{$\gamma _{4}$} & \multicolumn{1}{c}{$\lambda _{1}$} &
  \multicolumn{1}{c}{$\lambda _{2}$} & \multicolumn{1}{c}{$\lambda _{3}$} \\ 
True   &   1.000 &   0.000 &   0.000 &  -1.000 &   0.000 &   0.000 &   1.000 &  -2.000 &   0.000 &   2.000  \\
Bias   &   0.008 &   0.004 &   0.001 &   0.014 &   0.000 &  -0.036 &  -0.027 &  -0.027 &   0.000 &   0.024  \\
IRQ    &   0.096 &   0.057 &   0.055 &   0.181 &   0.000 &   0.174 &   0.201 &   0.097 &   0.000 &   0.108  \\
MAE    &   0.046 &   0.028 &   0.028 &   0.097 &   0.000 &   0.092 &   0.104 &   0.051 &   0.000 &   0.060  \\
 \\\multicolumn{11}{c}{Method of Moments.  Sample Size: 500. (nrep = 400)} \\ 
& \multicolumn{1}{c}{$\beta _{1}$} & \multicolumn{1}{c}{$\beta _{2}$} &
  \multicolumn{1}{c}{$\beta _{3}$} & \multicolumn{1}{c}{$\gamma _{1}$} &
  \multicolumn{1}{c}{$\gamma _{2}$} & \multicolumn{1}{c}{$\gamma _{3}$} &
  \multicolumn{1}{c}{$\gamma _{4}$} & \multicolumn{1}{c}{$\lambda _{1}$} &
  \multicolumn{1}{c}{$\lambda _{2}$} & \multicolumn{1}{c}{$\lambda _{3}$} \\ 
True   &   1.000 &   0.000 &   0.000 &  -1.000 &   0.000 &   0.000 &   1.000 &  -2.000 &   0.000 &   2.000  \\
Bias    &  -0.001 &   0.004 &  -0.002 &   0.019 &   0.000 &  -0.018 &  -0.010 &  -0.003 &   0.000 &   0.009  \\
IQR     &   0.112 &   0.068 &   0.066 &   0.249 &   0.000 &   0.232 &   0.252 &   0.135 &   0.000 &   0.144  \\
MAE     &   0.056 &   0.034 &   0.033 &   0.124 &   0.000 &   0.117 &   0.123 &   0.067 &   0.000 &   0.072  \\
 \\
MAE ratio &   0.832 &   0.833 &   0.830 &   0.780 & \multicolumn{1}{c}{---} &   0.787 &   0.845 &   0.760 & \multicolumn{1}{c}{---} &   0.835  \\
 \\\multicolumn{11}{c}{Correlated Random Effects.  Sample Size: 1000. (nrep = 400)} \\ 
& \multicolumn{1}{c}{$\beta _{1}$} & \multicolumn{1}{c}{$\beta _{2}$} &
  \multicolumn{1}{c}{$\beta _{3}$} & \multicolumn{1}{c}{$\gamma _{1}$} &
  \multicolumn{1}{c}{$\gamma _{2}$} & \multicolumn{1}{c}{$\gamma _{3}$} &
  \multicolumn{1}{c}{$\gamma _{4}$} & \multicolumn{1}{c}{$\lambda _{1}$} &
  \multicolumn{1}{c}{$\lambda _{2}$} & \multicolumn{1}{c}{$\lambda _{3}$} \\ 
True   &   1.000 &   0.000 &   0.000 &  -1.000 &   0.000 &   0.000 &   1.000 &  -2.000 &   0.000 &   2.000  \\
Bias   &   0.006 &   0.000 &   0.002 &   0.010 &   0.000 &  -0.014 &  -0.023 &  -0.013 &   0.000 &   0.013  \\
IRQ    &   0.066 &   0.038 &   0.039 &   0.118 &   0.000 &   0.112 &   0.139 &   0.082 &   0.000 &   0.082  \\
MAE    &   0.032 &   0.019 &   0.020 &   0.057 &   0.000 &   0.060 &   0.072 &   0.038 &   0.000 &   0.041  \\
 \\\multicolumn{11}{c}{Method of Moments.  Sample Size: 1000. (nrep = 400)} \\ 
& \multicolumn{1}{c}{$\beta _{1}$} & \multicolumn{1}{c}{$\beta _{2}$} &
  \multicolumn{1}{c}{$\beta _{3}$} & \multicolumn{1}{c}{$\gamma _{1}$} &
  \multicolumn{1}{c}{$\gamma _{2}$} & \multicolumn{1}{c}{$\gamma _{3}$} &
  \multicolumn{1}{c}{$\gamma _{4}$} & \multicolumn{1}{c}{$\lambda _{1}$} &
  \multicolumn{1}{c}{$\lambda _{2}$} & \multicolumn{1}{c}{$\lambda _{3}$} \\ 
True   &   1.000 &   0.000 &   0.000 &  -1.000 &   0.000 &   0.000 &   1.000 &  -2.000 &   0.000 &   2.000  \\
Bias    &  -0.001 &  -0.001 &   0.002 &   0.005 &   0.000 &  -0.009 &  -0.005 &   0.008 &   0.000 &   0.005  \\
IQR     &   0.076 &   0.046 &   0.049 &   0.155 &   0.000 &   0.159 &   0.195 &   0.107 &   0.000 &   0.107  \\
MAE     &   0.037 &   0.023 &   0.023 &   0.077 &   0.000 &   0.081 &   0.097 &   0.055 &   0.000 &   0.053  \\
 \\
MAE ratio &   0.876 &   0.814 &   0.859 &   0.737 & \multicolumn{1}{c}{---} &   0.736 &   0.741 &   0.701 & \multicolumn{1}{c}{---} &   0.775  \\
 \\\multicolumn{11}{c}{Correlated Random Effects.  Sample Size: 2000. (nrep = 400)} \\ 
& \multicolumn{1}{c}{$\beta _{1}$} & \multicolumn{1}{c}{$\beta _{2}$} &
  \multicolumn{1}{c}{$\beta _{3}$} & \multicolumn{1}{c}{$\gamma _{1}$} &
  \multicolumn{1}{c}{$\gamma _{2}$} & \multicolumn{1}{c}{$\gamma _{3}$} &
  \multicolumn{1}{c}{$\gamma _{4}$} & \multicolumn{1}{c}{$\lambda _{1}$} &
  \multicolumn{1}{c}{$\lambda _{2}$} & \multicolumn{1}{c}{$\lambda _{3}$} \\ 
True   &   1.000 &   0.000 &   0.000 &  -1.000 &   0.000 &   0.000 &   1.000 &  -2.000 &   0.000 &   2.000  \\
Bias   &   0.002 &  -0.000 &  -0.001 &   0.013 &   0.000 &  -0.010 &  -0.011 &  -0.010 &   0.000 &   0.012  \\
IRQ    &   0.046 &   0.032 &   0.025 &   0.100 &   0.000 &   0.084 &   0.098 &   0.049 &   0.000 &   0.059  \\
MAE    &   0.022 &   0.016 &   0.013 &   0.052 &   0.000 &   0.041 &   0.051 &   0.025 &   0.000 &   0.032  \\
 \\\multicolumn{11}{c}{Method of Moments.  Sample Size: 2000. (nrep = 400)} \\ 
& \multicolumn{1}{c}{$\beta _{1}$} & \multicolumn{1}{c}{$\beta _{2}$} &
  \multicolumn{1}{c}{$\beta _{3}$} & \multicolumn{1}{c}{$\gamma _{1}$} &
  \multicolumn{1}{c}{$\gamma _{2}$} & \multicolumn{1}{c}{$\gamma _{3}$} &
  \multicolumn{1}{c}{$\gamma _{4}$} & \multicolumn{1}{c}{$\lambda _{1}$} &
  \multicolumn{1}{c}{$\lambda _{2}$} & \multicolumn{1}{c}{$\lambda _{3}$} \\ 
True   &   1.000 &   0.000 &   0.000 &  -1.000 &   0.000 &   0.000 &   1.000 &  -2.000 &   0.000 &   2.000  \\
Bias    &  -0.001 &   0.001 &  -0.002 &   0.005 &   0.000 &  -0.001 &   0.008 &  -0.004 &   0.000 &   0.001  \\
IQR     &   0.046 &   0.033 &   0.030 &   0.123 &   0.000 &   0.101 &   0.129 &   0.067 &   0.000 &   0.072  \\
MAE     &   0.023 &   0.016 &   0.015 &   0.061 &   0.000 &   0.050 &   0.066 &   0.033 &   0.000 &   0.036  \\
 \\
MAE ratio &   0.955 &   0.979 &   0.842 &   0.839 & \multicolumn{1}{c}{---} &   0.810 &   0.775 &   0.768 & \multicolumn{1}{c}{---} &   0.898  \\
 \\\end{tabular}
 }
\par
{\footnotesize MAE is the median absolute error, IQR is the interquartile
range. $\gamma_2$ and $\lambda_2$ are normalized.}
\end{center}
\end{table}

\begin{table}[tbp]
\caption{Design B}
\label{Design B}
\begin{center}
{\footnotesize \begin{tabular}{lrrrrrrrrrr} 
\multicolumn{11}{c}{Correlated Random Effects.  Sample Size: 500. (nrep = 400)} \\ 
& \multicolumn{1}{c}{$\beta _{1}$} & \multicolumn{1}{c}{$\beta _{2}$} &
  \multicolumn{1}{c}{$\beta _{3}$} & \multicolumn{1}{c}{$\gamma _{1}$} &
  \multicolumn{1}{c}{$\gamma _{2}$} & \multicolumn{1}{c}{$\gamma _{3}$} &
  \multicolumn{1}{c}{$\gamma _{4}$} & \multicolumn{1}{c}{$\lambda _{1}$} &
  \multicolumn{1}{c}{$\lambda _{2}$} & \multicolumn{1}{c}{$\lambda _{3}$} \\ 
True   &   1.000 &   0.000 &   0.000 &  -1.000 &   0.000 &   0.000 &   1.000 &  -2.000 &   0.000 &   2.000  \\
Bias   &   0.003 &  -0.003 &  -0.003 &  -0.007 &   0.000 &  -0.010 &   0.007 &  -0.007 &   0.000 &   0.009  \\
IRQ    &   0.102 &   0.057 &   0.064 &   0.224 &   0.000 &   0.203 &   0.262 &   0.144 &   0.000 &   0.131  \\
MAE    &   0.051 &   0.028 &   0.032 &   0.112 &   0.000 &   0.100 &   0.130 &   0.073 &   0.000 &   0.062  \\
 \\\multicolumn{11}{c}{Method of Moments.  Sample Size: 500. (nrep = 400)} \\ 
& \multicolumn{1}{c}{$\beta _{1}$} & \multicolumn{1}{c}{$\beta _{2}$} &
  \multicolumn{1}{c}{$\beta _{3}$} & \multicolumn{1}{c}{$\gamma _{1}$} &
  \multicolumn{1}{c}{$\gamma _{2}$} & \multicolumn{1}{c}{$\gamma _{3}$} &
  \multicolumn{1}{c}{$\gamma _{4}$} & \multicolumn{1}{c}{$\lambda _{1}$} &
  \multicolumn{1}{c}{$\lambda _{2}$} & \multicolumn{1}{c}{$\lambda _{3}$} \\ 
True   &   1.000 &   0.000 &   0.000 &  -1.000 &   0.000 &   0.000 &   1.000 &  -2.000 &   0.000 &   2.000  \\
Bias    &  -0.001 &  -0.003 &  -0.004 &   0.010 &   0.000 &  -0.007 &  -0.006 &  -0.000 &   0.000 &   0.002  \\
IQR     &   0.118 &   0.078 &   0.084 &   0.323 &   0.000 &   0.300 &   0.323 &   0.187 &   0.000 &   0.172  \\
MAE     &   0.059 &   0.039 &   0.043 &   0.163 &   0.000 &   0.154 &   0.162 &   0.096 &   0.000 &   0.087  \\
 \\
MAE ratio &   0.867 &   0.718 &   0.742 &   0.688 & \multicolumn{1}{c}{---} &   0.650 &   0.802 &   0.754 & \multicolumn{1}{c}{---} &   0.715  \\
 \\\multicolumn{11}{c}{Correlated Random Effects.  Sample Size: 1000. (nrep = 400)} \\ 
& \multicolumn{1}{c}{$\beta _{1}$} & \multicolumn{1}{c}{$\beta _{2}$} &
  \multicolumn{1}{c}{$\beta _{3}$} & \multicolumn{1}{c}{$\gamma _{1}$} &
  \multicolumn{1}{c}{$\gamma _{2}$} & \multicolumn{1}{c}{$\gamma _{3}$} &
  \multicolumn{1}{c}{$\gamma _{4}$} & \multicolumn{1}{c}{$\lambda _{1}$} &
  \multicolumn{1}{c}{$\lambda _{2}$} & \multicolumn{1}{c}{$\lambda _{3}$} \\ 
True   &   1.000 &   0.000 &   0.000 &  -1.000 &   0.000 &   0.000 &   1.000 &  -2.000 &   0.000 &   2.000  \\
Bias   &  -0.001 &  -0.000 &  -0.000 &  -0.008 &   0.000 &   0.006 &   0.006 &  -0.005 &   0.000 &   0.000  \\
IRQ    &   0.073 &   0.047 &   0.048 &   0.173 &   0.000 &   0.141 &   0.183 &   0.106 &   0.000 &   0.098  \\
MAE    &   0.037 &   0.023 &   0.023 &   0.086 &   0.000 &   0.070 &   0.089 &   0.054 &   0.000 &   0.049  \\
 \\\multicolumn{11}{c}{Method of Moments.  Sample Size: 1000. (nrep = 400)} \\ 
& \multicolumn{1}{c}{$\beta _{1}$} & \multicolumn{1}{c}{$\beta _{2}$} &
  \multicolumn{1}{c}{$\beta _{3}$} & \multicolumn{1}{c}{$\gamma _{1}$} &
  \multicolumn{1}{c}{$\gamma _{2}$} & \multicolumn{1}{c}{$\gamma _{3}$} &
  \multicolumn{1}{c}{$\gamma _{4}$} & \multicolumn{1}{c}{$\lambda _{1}$} &
  \multicolumn{1}{c}{$\lambda _{2}$} & \multicolumn{1}{c}{$\lambda _{3}$} \\ 
True   &   1.000 &   0.000 &   0.000 &  -1.000 &   0.000 &   0.000 &   1.000 &  -2.000 &   0.000 &   2.000  \\
Bias    &   0.001 &  -0.002 &   0.001 &   0.004 &   0.000 &  -0.007 &   0.006 &  -0.002 &   0.000 &   0.004  \\
IQR     &   0.086 &   0.060 &   0.055 &   0.205 &   0.000 &   0.196 &   0.231 &   0.126 &   0.000 &   0.128  \\
MAE     &   0.043 &   0.030 &   0.028 &   0.104 &   0.000 &   0.098 &   0.119 &   0.062 &   0.000 &   0.063  \\
 \\
MAE ratio &   0.867 &   0.782 &   0.852 &   0.829 & \multicolumn{1}{c}{---} &   0.720 &   0.751 &   0.883 & \multicolumn{1}{c}{---} &   0.775  \\
 \\\multicolumn{11}{c}{Correlated Random Effects.  Sample Size: 2000. (nrep = 400)} \\ 
& \multicolumn{1}{c}{$\beta _{1}$} & \multicolumn{1}{c}{$\beta _{2}$} &
  \multicolumn{1}{c}{$\beta _{3}$} & \multicolumn{1}{c}{$\gamma _{1}$} &
  \multicolumn{1}{c}{$\gamma _{2}$} & \multicolumn{1}{c}{$\gamma _{3}$} &
  \multicolumn{1}{c}{$\gamma _{4}$} & \multicolumn{1}{c}{$\lambda _{1}$} &
  \multicolumn{1}{c}{$\lambda _{2}$} & \multicolumn{1}{c}{$\lambda _{3}$} \\ 
True   &   1.000 &   0.000 &   0.000 &  -1.000 &   0.000 &   0.000 &   1.000 &  -2.000 &   0.000 &   2.000  \\
Bias   &   0.004 &  -0.000 &   0.000 &   0.010 &   0.000 &   0.011 &   0.005 &  -0.004 &   0.000 &  -0.001  \\
IRQ    &   0.049 &   0.034 &   0.032 &   0.121 &   0.000 &   0.096 &   0.119 &   0.066 &   0.000 &   0.065  \\
MAE    &   0.024 &   0.017 &   0.016 &   0.059 &   0.000 &   0.049 &   0.062 &   0.034 &   0.000 &   0.033  \\
 \\\multicolumn{11}{c}{Method of Moments.  Sample Size: 2000. (nrep = 400)} \\ 
& \multicolumn{1}{c}{$\beta _{1}$} & \multicolumn{1}{c}{$\beta _{2}$} &
  \multicolumn{1}{c}{$\beta _{3}$} & \multicolumn{1}{c}{$\gamma _{1}$} &
  \multicolumn{1}{c}{$\gamma _{2}$} & \multicolumn{1}{c}{$\gamma _{3}$} &
  \multicolumn{1}{c}{$\gamma _{4}$} & \multicolumn{1}{c}{$\lambda _{1}$} &
  \multicolumn{1}{c}{$\lambda _{2}$} & \multicolumn{1}{c}{$\lambda _{3}$} \\ 
True   &   1.000 &   0.000 &   0.000 &  -1.000 &   0.000 &   0.000 &   1.000 &  -2.000 &   0.000 &   2.000  \\
Bias    &   0.003 &  -0.002 &   0.001 &   0.016 &   0.000 &   0.008 &   0.004 &  -0.002 &   0.000 &   0.001  \\
IQR     &   0.058 &   0.040 &   0.042 &   0.145 &   0.000 &   0.122 &   0.156 &   0.084 &   0.000 &   0.077  \\
MAE     &   0.029 &   0.020 &   0.021 &   0.073 &   0.000 &   0.061 &   0.078 &   0.042 &   0.000 &   0.039  \\
 \\
MAE ratio &   0.853 &   0.847 &   0.753 &   0.807 & \multicolumn{1}{c}{---} &   0.799 &   0.793 &   0.802 & \multicolumn{1}{c}{---} &   0.854  \\
 \\\end{tabular}
 }
\par
{\footnotesize MAE is the median absolute error, IQR is the interquartile
range. $\gamma_2$ and $\lambda_2$ are normalized.}
\end{center}
\end{table}

\begin{table}[tbp]
\caption{Design C}
\label{Design C}
\begin{center}
{\footnotesize \begin{tabular}{lrrrrrrrrrr} 
\multicolumn{11}{c}{Correlated Random Effects.  Sample Size: 500. (nrep = 400)} \\ 
& \multicolumn{1}{c}{$\beta _{1}$} & \multicolumn{1}{c}{$\beta _{2}$} &
  \multicolumn{1}{c}{$\beta _{3}$} & \multicolumn{1}{c}{$\gamma _{1}$} &
  \multicolumn{1}{c}{$\gamma _{2}$} & \multicolumn{1}{c}{$\gamma _{3}$} &
  \multicolumn{1}{c}{$\gamma _{4}$} & \multicolumn{1}{c}{$\lambda _{1}$} &
  \multicolumn{1}{c}{$\lambda _{2}$} & \multicolumn{1}{c}{$\lambda _{3}$} \\ 
True   &   1.000 &   0.000 &   0.000 &  -1.000 &   0.000 &   0.000 &   1.000 &  -2.000 &   0.000 &   2.000  \\
Bias   &   0.006 &   0.007 &   0.003 &   0.168 &   0.000 &   0.127 &   0.674 &  -0.049 &   0.000 &  -0.308  \\
IRQ    &   0.108 &   0.073 &   0.070 &   0.216 &   0.000 &   0.260 &   0.365 &   0.156 &   0.000 &   0.173  \\
MAE    &   0.054 &   0.037 &   0.037 &   0.170 &   0.000 &   0.156 &   0.674 &   0.081 &   0.000 &   0.308  \\
 \\\multicolumn{11}{c}{Method of Moments.  Sample Size: 500. (nrep = 400)} \\ 
& \multicolumn{1}{c}{$\beta _{1}$} & \multicolumn{1}{c}{$\beta _{2}$} &
  \multicolumn{1}{c}{$\beta _{3}$} & \multicolumn{1}{c}{$\gamma _{1}$} &
  \multicolumn{1}{c}{$\gamma _{2}$} & \multicolumn{1}{c}{$\gamma _{3}$} &
  \multicolumn{1}{c}{$\gamma _{4}$} & \multicolumn{1}{c}{$\lambda _{1}$} &
  \multicolumn{1}{c}{$\lambda _{2}$} & \multicolumn{1}{c}{$\lambda _{3}$} \\ 
True   &   1.000 &   0.000 &   0.000 &  -1.000 &   0.000 &   0.000 &   1.000 &  -2.000 &   0.000 &   2.000  \\
Bias    &  -0.018 &   0.005 &   0.009 &   0.074 &   0.000 &   0.007 &   0.303 &   0.026 &   0.000 &  -0.146  \\
IQR     &   0.185 &   0.112 &   0.106 &   0.327 &   0.000 &   0.479 &   0.702 &   0.271 &   0.000 &   0.376  \\
MAE     &   0.097 &   0.058 &   0.054 &   0.173 &   0.000 &   0.232 &   0.445 &   0.139 &   0.000 &   0.227  \\
 \\
MAE ratio &   0.556 &   0.638 &   0.677 &   0.983 & \multicolumn{1}{c}{---} &   0.671 &   1.515 &   0.586 & \multicolumn{1}{c}{---} &   1.355  \\
 \\\multicolumn{11}{c}{Correlated Random Effects.  Sample Size: 1000. (nrep = 400)} \\ 
& \multicolumn{1}{c}{$\beta _{1}$} & \multicolumn{1}{c}{$\beta _{2}$} &
  \multicolumn{1}{c}{$\beta _{3}$} & \multicolumn{1}{c}{$\gamma _{1}$} &
  \multicolumn{1}{c}{$\gamma _{2}$} & \multicolumn{1}{c}{$\gamma _{3}$} &
  \multicolumn{1}{c}{$\gamma _{4}$} & \multicolumn{1}{c}{$\lambda _{1}$} &
  \multicolumn{1}{c}{$\lambda _{2}$} & \multicolumn{1}{c}{$\lambda _{3}$} \\ 
True   &   1.000 &   0.000 &   0.000 &  -1.000 &   0.000 &   0.000 &   1.000 &  -2.000 &   0.000 &   2.000  \\
Bias   &   0.004 &  -0.004 &  -0.000 &   0.189 &   0.000 &   0.137 &   0.691 &  -0.045 &   0.000 &  -0.308  \\
IRQ    &   0.081 &   0.052 &   0.055 &   0.151 &   0.000 &   0.192 &   0.272 &   0.098 &   0.000 &   0.110  \\
MAE    &   0.040 &   0.026 &   0.028 &   0.189 &   0.000 &   0.156 &   0.691 &   0.058 &   0.000 &   0.308  \\
 \\\multicolumn{11}{c}{Method of Moments.  Sample Size: 1000. (nrep = 400)} \\ 
& \multicolumn{1}{c}{$\beta _{1}$} & \multicolumn{1}{c}{$\beta _{2}$} &
  \multicolumn{1}{c}{$\beta _{3}$} & \multicolumn{1}{c}{$\gamma _{1}$} &
  \multicolumn{1}{c}{$\gamma _{2}$} & \multicolumn{1}{c}{$\gamma _{3}$} &
  \multicolumn{1}{c}{$\gamma _{4}$} & \multicolumn{1}{c}{$\lambda _{1}$} &
  \multicolumn{1}{c}{$\lambda _{2}$} & \multicolumn{1}{c}{$\lambda _{3}$} \\ 
True   &   1.000 &   0.000 &   0.000 &  -1.000 &   0.000 &   0.000 &   1.000 &  -2.000 &   0.000 &   2.000  \\
Bias    &  -0.013 &  -0.006 &  -0.000 &   0.017 &   0.000 &   0.035 &   0.202 &   0.025 &   0.000 &  -0.075  \\
IQR     &   0.144 &   0.084 &   0.092 &   0.290 &   0.000 &   0.457 &   0.638 &   0.176 &   0.000 &   0.322  \\
MAE     &   0.070 &   0.044 &   0.045 &   0.148 &   0.000 &   0.227 &   0.354 &   0.098 &   0.000 &   0.169  \\
 \\
MAE ratio &   0.574 &   0.594 &   0.626 &   1.277 & \multicolumn{1}{c}{---} &   0.687 &   1.950 &   0.597 & \multicolumn{1}{c}{---} &   1.822  \\
 \\\multicolumn{11}{c}{Correlated Random Effects.  Sample Size: 2000. (nrep = 400)} \\ 
& \multicolumn{1}{c}{$\beta _{1}$} & \multicolumn{1}{c}{$\beta _{2}$} &
  \multicolumn{1}{c}{$\beta _{3}$} & \multicolumn{1}{c}{$\gamma _{1}$} &
  \multicolumn{1}{c}{$\gamma _{2}$} & \multicolumn{1}{c}{$\gamma _{3}$} &
  \multicolumn{1}{c}{$\gamma _{4}$} & \multicolumn{1}{c}{$\lambda _{1}$} &
  \multicolumn{1}{c}{$\lambda _{2}$} & \multicolumn{1}{c}{$\lambda _{3}$} \\ 
True   &   1.000 &   0.000 &   0.000 &  -1.000 &   0.000 &   0.000 &   1.000 &  -2.000 &   0.000 &   2.000  \\
Bias   &   0.007 &  -0.003 &  -0.002 &   0.179 &   0.000 &   0.133 &   0.712 &  -0.043 &   0.000 &  -0.315  \\
IRQ    &   0.052 &   0.037 &   0.034 &   0.099 &   0.000 &   0.125 &   0.171 &   0.080 &   0.000 &   0.079  \\
MAE    &   0.025 &   0.019 &   0.017 &   0.179 &   0.000 &   0.133 &   0.712 &   0.051 &   0.000 &   0.315  \\
 \\\multicolumn{11}{c}{Method of Moments.  Sample Size: 2000. (nrep = 400)} \\ 
& \multicolumn{1}{c}{$\beta _{1}$} & \multicolumn{1}{c}{$\beta _{2}$} &
  \multicolumn{1}{c}{$\beta _{3}$} & \multicolumn{1}{c}{$\gamma _{1}$} &
  \multicolumn{1}{c}{$\gamma _{2}$} & \multicolumn{1}{c}{$\gamma _{3}$} &
  \multicolumn{1}{c}{$\gamma _{4}$} & \multicolumn{1}{c}{$\lambda _{1}$} &
  \multicolumn{1}{c}{$\lambda _{2}$} & \multicolumn{1}{c}{$\lambda _{3}$} \\ 
True   &   1.000 &   0.000 &   0.000 &  -1.000 &   0.000 &   0.000 &   1.000 &  -2.000 &   0.000 &   2.000  \\
Bias    &  -0.003 &  -0.003 &  -0.008 &   0.006 &   0.000 &   0.026 &   0.090 &   0.024 &   0.000 &  -0.028  \\
IQR     &   0.107 &   0.066 &   0.058 &   0.200 &   0.000 &   0.336 &   0.559 &   0.143 &   0.000 &   0.226  \\
MAE     &   0.052 &   0.033 &   0.030 &   0.101 &   0.000 &   0.163 &   0.276 &   0.075 &   0.000 &   0.122  \\
 \\
MAE ratio &   0.483 &   0.561 &   0.560 &   1.781 & \multicolumn{1}{c}{---} &   0.819 &   2.577 &   0.686 & \multicolumn{1}{c}{---} &   2.589  \\
 \\\end{tabular}
 }
\par
{\footnotesize MAE is the median absolute error, IQR is the interquartile
range. $\gamma_2$ and $\lambda_2$ are normalized.}
\end{center}
\end{table}

Tables \ref{Design A CI}-\ref{Design C CI} illustrate that the smaller bias of the method of
moments estimator can make it more reliable for inference than the
correlated random effects estimator when the assumptions underpinning the
correlated random effects estimator are violated. Specifically, the table
presents the fraction of times that 80, 90 and 95 percent confidence
intervals based on the correlated random effects estimator and on the method
of moments estimator cover the true unknown parameter. As one would expect
from the results in Table \ref{Design C}, the bias in the correlated random
effects estimator combined with its low variability can make it unlikely
that a confidence interval covers the true parameter value. Tables \ref%
{Design A CI} and \ref{Design B CI} show the same results for Design A and
Design B. In these case, the confidence intervals based on the 
correlated random effects estimator do very well. For Design B, this is not
surprising as this is a correctly specified maximum likelihood setting. It is
interesting that the correlated random effects estimator also does well in
Design A. Although it is the maximum likelihood estimator of a correctly
specified model, the estimation is made non-standard by the fact that the
true value of one of the parameters (the variance of the error in the
specification for the random effect) is on the boundary of the parameter
space and asymptotic normality would therefore not follow from textbook
asymptotic theory. As a general statement, Tables \ref{Design A CI}, \ref%
{Design B CI} and \ref{Design C CI} also illustrate that the confidence
interval based on the method of moments estimator can be somewhat erratic even
with relatively large sample sizes.\footnote{%
On the other hand, in simulations not reported here, we have found that the
natural estimator of the variance of the method of moments estimator can
perform poorly. The results reported here therefore use bootstrap
standard errors based on the interquartile range of 1000 bootstrap
replications. The standard errors reported for the correlated random effects
are based on the \textquotedblleft robust\textquotedblright\ expression for
the asymptotic variance of extremum estimators.}

\begin{table}[tbp]
\caption{Design A}
\label{Design A CI}
\begin{center}
{\footnotesize \begin{tabular}{lrrrrrrrrrr} 
\multicolumn{11}{c}{Correlated Random Effects.  Sample Size: 500. (nrep = 400)} \\ 
& \multicolumn{1}{c}{$\beta _{1}$} & \multicolumn{1}{c}{$\beta _{2}$} &
  \multicolumn{1}{c}{$\beta _{3}$} & \multicolumn{1}{c}{$\gamma _{1}$} &
  \multicolumn{1}{c}{$\gamma _{2}$} & \multicolumn{1}{c}{$\gamma _{3}$} &
  \multicolumn{1}{c}{$\gamma _{4}$} & \multicolumn{1}{c}{$\lambda _{1}$} &
  \multicolumn{1}{c}{$\lambda _{2}$} & \multicolumn{1}{c}{$\lambda _{3}$} \\ 
80\% CI &   0.797 &   0.812 &   0.807 &   0.790 & \multicolumn{1}{c}{---} &   0.770 &   0.780 &   0.807 & \multicolumn{1}{c}{---} &   0.787  \\
90\% CI &   0.912 &   0.897 &   0.902 &   0.892 & \multicolumn{1}{c}{---} &   0.892 &   0.875 &   0.912 & \multicolumn{1}{c}{---} &   0.900  \\
95\% CI &   0.948 &   0.953 &   0.950 &   0.938 & \multicolumn{1}{c}{---} &   0.940 &   0.938 &   0.960 & \multicolumn{1}{c}{---} &   0.968  \\
 \\ 
\multicolumn{11}{c}{Method of Moments.  Sample Size: 500. (nrep = 400)} \\ 
& \multicolumn{1}{c}{$\beta _{1}$} & \multicolumn{1}{c}{$\beta _{2}$} &
  \multicolumn{1}{c}{$\beta _{3}$} & \multicolumn{1}{c}{$\gamma _{1}$} &
  \multicolumn{1}{c}{$\gamma _{2}$} & \multicolumn{1}{c}{$\gamma _{3}$} &
  \multicolumn{1}{c}{$\gamma _{4}$} & \multicolumn{1}{c}{$\lambda _{1}$} &
  \multicolumn{1}{c}{$\lambda _{2}$} & \multicolumn{1}{c}{$\lambda _{3}$} \\ 
80\% CI &   0.870 &   0.850 &   0.850 &   0.902 & \multicolumn{1}{c}{---} &   0.870 &   0.875 &   0.930 & \multicolumn{1}{c}{---} &   0.953  \\
90\% CI &   0.935 &   0.932 &   0.950 &   0.950 & \multicolumn{1}{c}{---} &   0.945 &   0.930 &   0.973 & \multicolumn{1}{c}{---} &   0.988  \\
95\% CI &   0.965 &   0.975 &   0.983 &   0.970 & \multicolumn{1}{c}{---} &   0.978 &   0.968 &   0.995 & \multicolumn{1}{c}{---} &   0.995  \\
 \\ \multicolumn{11}{c}{Correlated Random Effects.  Sample Size: 1000. (nrep = 400)} \\ 
& \multicolumn{1}{c}{$\beta _{1}$} & \multicolumn{1}{c}{$\beta _{2}$} &
  \multicolumn{1}{c}{$\beta _{3}$} & \multicolumn{1}{c}{$\gamma _{1}$} &
  \multicolumn{1}{c}{$\gamma _{2}$} & \multicolumn{1}{c}{$\gamma _{3}$} &
  \multicolumn{1}{c}{$\gamma _{4}$} & \multicolumn{1}{c}{$\lambda _{1}$} &
  \multicolumn{1}{c}{$\lambda _{2}$} & \multicolumn{1}{c}{$\lambda _{3}$} \\ 
80\% CI &   0.785 &   0.815 &   0.800 &   0.818 & \multicolumn{1}{c}{---} &   0.775 &   0.795 &   0.772 & \multicolumn{1}{c}{---} &   0.782  \\
90\% CI &   0.912 &   0.915 &   0.887 &   0.895 & \multicolumn{1}{c}{---} &   0.883 &   0.875 &   0.873 & \multicolumn{1}{c}{---} &   0.900  \\
95\% CI &   0.953 &   0.950 &   0.940 &   0.938 & \multicolumn{1}{c}{---} &   0.935 &   0.932 &   0.935 & \multicolumn{1}{c}{---} &   0.950  \\
 \\ 
\multicolumn{11}{c}{Method of Moments.  Sample Size: 1000. (nrep = 400)} \\ 
& \multicolumn{1}{c}{$\beta _{1}$} & \multicolumn{1}{c}{$\beta _{2}$} &
  \multicolumn{1}{c}{$\beta _{3}$} & \multicolumn{1}{c}{$\gamma _{1}$} &
  \multicolumn{1}{c}{$\gamma _{2}$} & \multicolumn{1}{c}{$\gamma _{3}$} &
  \multicolumn{1}{c}{$\gamma _{4}$} & \multicolumn{1}{c}{$\lambda _{1}$} &
  \multicolumn{1}{c}{$\lambda _{2}$} & \multicolumn{1}{c}{$\lambda _{3}$} \\ 
80\% CI &   0.890 &   0.897 &   0.883 &   0.900 & \multicolumn{1}{c}{---} &   0.885 &   0.905 &   0.910 & \multicolumn{1}{c}{---} &   0.940  \\
90\% CI &   0.958 &   0.953 &   0.955 &   0.960 & \multicolumn{1}{c}{---} &   0.960 &   0.963 &   0.988 & \multicolumn{1}{c}{---} &   0.978  \\
95\% CI &   0.980 &   0.973 &   0.978 &   0.978 & \multicolumn{1}{c}{---} &   0.975 &   0.990 &   0.998 & \multicolumn{1}{c}{---} &   0.988  \\
 \\ \multicolumn{11}{c}{Correlated Random Effects.  Sample Size: 2000. (nrep = 400)} \\ 
& \multicolumn{1}{c}{$\beta _{1}$} & \multicolumn{1}{c}{$\beta _{2}$} &
  \multicolumn{1}{c}{$\beta _{3}$} & \multicolumn{1}{c}{$\gamma _{1}$} &
  \multicolumn{1}{c}{$\gamma _{2}$} & \multicolumn{1}{c}{$\gamma _{3}$} &
  \multicolumn{1}{c}{$\gamma _{4}$} & \multicolumn{1}{c}{$\lambda _{1}$} &
  \multicolumn{1}{c}{$\lambda _{2}$} & \multicolumn{1}{c}{$\lambda _{3}$} \\ 
80\% CI &   0.765 &   0.772 &   0.850 &   0.777 & \multicolumn{1}{c}{---} &   0.815 &   0.805 &   0.782 & \multicolumn{1}{c}{---} &   0.807  \\
90\% CI &   0.863 &   0.892 &   0.917 &   0.890 & \multicolumn{1}{c}{---} &   0.927 &   0.897 &   0.883 & \multicolumn{1}{c}{---} &   0.895  \\
95\% CI &   0.927 &   0.940 &   0.945 &   0.940 & \multicolumn{1}{c}{---} &   0.963 &   0.927 &   0.932 & \multicolumn{1}{c}{---} &   0.953  \\
 \\ 
\multicolumn{11}{c}{Method of Moments.  Sample Size: 2000. (nrep = 400)} \\ 
& \multicolumn{1}{c}{$\beta _{1}$} & \multicolumn{1}{c}{$\beta _{2}$} &
  \multicolumn{1}{c}{$\beta _{3}$} & \multicolumn{1}{c}{$\gamma _{1}$} &
  \multicolumn{1}{c}{$\gamma _{2}$} & \multicolumn{1}{c}{$\gamma _{3}$} &
  \multicolumn{1}{c}{$\gamma _{4}$} & \multicolumn{1}{c}{$\lambda _{1}$} &
  \multicolumn{1}{c}{$\lambda _{2}$} & \multicolumn{1}{c}{$\lambda _{3}$} \\ 
80\% CI &   0.863 &   0.865 &   0.892 &   0.892 & \multicolumn{1}{c}{---} &   0.920 &   0.885 &   0.905 & \multicolumn{1}{c}{---} &   0.890  \\
90\% CI &   0.915 &   0.935 &   0.958 &   0.950 & \multicolumn{1}{c}{---} &   0.973 &   0.955 &   0.965 & \multicolumn{1}{c}{---} &   0.958  \\
95\% CI &   0.958 &   0.968 &   0.985 &   0.980 & \multicolumn{1}{c}{---} &   0.990 &   0.985 &   0.975 & \multicolumn{1}{c}{---} &   0.988  \\
 \\ \end{tabular}
 }
\par
{\footnotesize Based on 1000 bootstrap replications. $\gamma_2$ and $%
\lambda_2$ are normalized.}
\end{center}
\end{table}

\begin{table}[tbp]
\caption{Design B}
\label{Design B CI}
\begin{center}
{\footnotesize \begin{tabular}{lrrrrrrrrrr} 
\multicolumn{11}{c}{Correlated Random Effects.  Sample Size: 500. (nrep = 400)} \\ 
& \multicolumn{1}{c}{$\beta _{1}$} & \multicolumn{1}{c}{$\beta _{2}$} &
  \multicolumn{1}{c}{$\beta _{3}$} & \multicolumn{1}{c}{$\gamma _{1}$} &
  \multicolumn{1}{c}{$\gamma _{2}$} & \multicolumn{1}{c}{$\gamma _{3}$} &
  \multicolumn{1}{c}{$\gamma _{4}$} & \multicolumn{1}{c}{$\lambda _{1}$} &
  \multicolumn{1}{c}{$\lambda _{2}$} & \multicolumn{1}{c}{$\lambda _{3}$} \\ 
80\% CI &   0.828 &   0.843 &   0.815 &   0.777 & \multicolumn{1}{c}{---} &   0.792 &   0.795 &   0.785 & \multicolumn{1}{c}{---} &   0.820  \\
90\% CI &   0.900 &   0.925 &   0.915 &   0.900 & \multicolumn{1}{c}{---} &   0.892 &   0.895 &   0.880 & \multicolumn{1}{c}{---} &   0.925  \\
95\% CI &   0.943 &   0.955 &   0.955 &   0.945 & \multicolumn{1}{c}{---} &   0.938 &   0.945 &   0.935 & \multicolumn{1}{c}{---} &   0.968  \\
 \\ 
\multicolumn{11}{c}{Method of Moments.  Sample Size: 500. (nrep = 400)} \\ 
& \multicolumn{1}{c}{$\beta _{1}$} & \multicolumn{1}{c}{$\beta _{2}$} &
  \multicolumn{1}{c}{$\beta _{3}$} & \multicolumn{1}{c}{$\gamma _{1}$} &
  \multicolumn{1}{c}{$\gamma _{2}$} & \multicolumn{1}{c}{$\gamma _{3}$} &
  \multicolumn{1}{c}{$\gamma _{4}$} & \multicolumn{1}{c}{$\lambda _{1}$} &
  \multicolumn{1}{c}{$\lambda _{2}$} & \multicolumn{1}{c}{$\lambda _{3}$} \\ 
80\% CI &   0.887 &   0.880 &   0.860 &   0.875 & \multicolumn{1}{c}{---} &   0.870 &   0.873 &   0.883 & \multicolumn{1}{c}{---} &   0.943  \\
90\% CI &   0.958 &   0.945 &   0.943 &   0.935 & \multicolumn{1}{c}{---} &   0.938 &   0.943 &   0.940 & \multicolumn{1}{c}{---} &   0.983  \\
95\% CI &   0.985 &   0.965 &   0.983 &   0.965 & \multicolumn{1}{c}{---} &   0.965 &   0.970 &   0.973 & \multicolumn{1}{c}{---} &   0.995  \\
 \\ \multicolumn{11}{c}{Correlated Random Effects.  Sample Size: 1000. (nrep = 400)} \\ 
& \multicolumn{1}{c}{$\beta _{1}$} & \multicolumn{1}{c}{$\beta _{2}$} &
  \multicolumn{1}{c}{$\beta _{3}$} & \multicolumn{1}{c}{$\gamma _{1}$} &
  \multicolumn{1}{c}{$\gamma _{2}$} & \multicolumn{1}{c}{$\gamma _{3}$} &
  \multicolumn{1}{c}{$\gamma _{4}$} & \multicolumn{1}{c}{$\lambda _{1}$} &
  \multicolumn{1}{c}{$\lambda _{2}$} & \multicolumn{1}{c}{$\lambda _{3}$} \\ 
80\% CI &   0.785 &   0.780 &   0.800 &   0.792 & \multicolumn{1}{c}{---} &   0.795 &   0.782 &   0.772 & \multicolumn{1}{c}{---} &   0.795  \\
90\% CI &   0.907 &   0.897 &   0.895 &   0.885 & \multicolumn{1}{c}{---} &   0.905 &   0.897 &   0.902 & \multicolumn{1}{c}{---} &   0.885  \\
95\% CI &   0.953 &   0.948 &   0.955 &   0.943 & \multicolumn{1}{c}{---} &   0.955 &   0.935 &   0.965 & \multicolumn{1}{c}{---} &   0.943  \\
 \\ 
\multicolumn{11}{c}{Method of Moments.  Sample Size: 1000. (nrep = 400)} \\ 
& \multicolumn{1}{c}{$\beta _{1}$} & \multicolumn{1}{c}{$\beta _{2}$} &
  \multicolumn{1}{c}{$\beta _{3}$} & \multicolumn{1}{c}{$\gamma _{1}$} &
  \multicolumn{1}{c}{$\gamma _{2}$} & \multicolumn{1}{c}{$\gamma _{3}$} &
  \multicolumn{1}{c}{$\gamma _{4}$} & \multicolumn{1}{c}{$\lambda _{1}$} &
  \multicolumn{1}{c}{$\lambda _{2}$} & \multicolumn{1}{c}{$\lambda _{3}$} \\ 
80\% CI &   0.892 &   0.890 &   0.895 &   0.900 & \multicolumn{1}{c}{---} &   0.875 &   0.897 &   0.943 & \multicolumn{1}{c}{---} &   0.920  \\
90\% CI &   0.965 &   0.950 &   0.955 &   0.948 & \multicolumn{1}{c}{---} &   0.968 &   0.968 &   0.973 & \multicolumn{1}{c}{---} &   0.960  \\
95\% CI &   0.985 &   0.960 &   0.983 &   0.975 & \multicolumn{1}{c}{---} &   0.983 &   0.985 &   0.990 & \multicolumn{1}{c}{---} &   0.985  \\
 \\ \multicolumn{11}{c}{Correlated Random Effects.  Sample Size: 2000. (nrep = 400)} \\ 
& \multicolumn{1}{c}{$\beta _{1}$} & \multicolumn{1}{c}{$\beta _{2}$} &
  \multicolumn{1}{c}{$\beta _{3}$} & \multicolumn{1}{c}{$\gamma _{1}$} &
  \multicolumn{1}{c}{$\gamma _{2}$} & \multicolumn{1}{c}{$\gamma _{3}$} &
  \multicolumn{1}{c}{$\gamma _{4}$} & \multicolumn{1}{c}{$\lambda _{1}$} &
  \multicolumn{1}{c}{$\lambda _{2}$} & \multicolumn{1}{c}{$\lambda _{3}$} \\ 
80\% CI &   0.800 &   0.810 &   0.802 &   0.805 & \multicolumn{1}{c}{---} &   0.820 &   0.785 &   0.787 & \multicolumn{1}{c}{---} &   0.802  \\
90\% CI &   0.880 &   0.897 &   0.905 &   0.920 & \multicolumn{1}{c}{---} &   0.922 &   0.887 &   0.887 & \multicolumn{1}{c}{---} &   0.907  \\
95\% CI &   0.945 &   0.948 &   0.950 &   0.973 & \multicolumn{1}{c}{---} &   0.960 &   0.935 &   0.932 & \multicolumn{1}{c}{---} &   0.943  \\
 \\ 
\multicolumn{11}{c}{Method of Moments.  Sample Size: 2000. (nrep = 400)} \\ 
& \multicolumn{1}{c}{$\beta _{1}$} & \multicolumn{1}{c}{$\beta _{2}$} &
  \multicolumn{1}{c}{$\beta _{3}$} & \multicolumn{1}{c}{$\gamma _{1}$} &
  \multicolumn{1}{c}{$\gamma _{2}$} & \multicolumn{1}{c}{$\gamma _{3}$} &
  \multicolumn{1}{c}{$\gamma _{4}$} & \multicolumn{1}{c}{$\lambda _{1}$} &
  \multicolumn{1}{c}{$\lambda _{2}$} & \multicolumn{1}{c}{$\lambda _{3}$} \\ 
80\% CI &   0.860 &   0.902 &   0.880 &   0.895 & \multicolumn{1}{c}{---} &   0.883 &   0.883 &   0.910 & \multicolumn{1}{c}{---} &   0.932  \\
90\% CI &   0.948 &   0.965 &   0.960 &   0.950 & \multicolumn{1}{c}{---} &   0.953 &   0.948 &   0.963 & \multicolumn{1}{c}{---} &   0.983  \\
95\% CI &   0.968 &   0.988 &   0.978 &   0.978 & \multicolumn{1}{c}{---} &   0.983 &   0.985 &   0.988 & \multicolumn{1}{c}{---} &   0.995  \\
 \\ \end{tabular}
 }
\par
{\footnotesize Based on 1000 bootstrap replications. $\gamma_2$ and $%
\lambda_2$ are normalized.}
\end{center}
\end{table}

\begin{table}[tbp]
\caption{Design C}
\label{Design C CI}
\begin{center}
{\footnotesize \begin{tabular}{lrrrrrrrrrr} 
\multicolumn{11}{c}{Correlated Random Effects.  Sample Size: 500. (nrep = 400)} \\ 
& \multicolumn{1}{c}{$\beta _{1}$} & \multicolumn{1}{c}{$\beta _{2}$} &
  \multicolumn{1}{c}{$\beta _{3}$} & \multicolumn{1}{c}{$\gamma _{1}$} &
  \multicolumn{1}{c}{$\gamma _{2}$} & \multicolumn{1}{c}{$\gamma _{3}$} &
  \multicolumn{1}{c}{$\gamma _{4}$} & \multicolumn{1}{c}{$\lambda _{1}$} &
  \multicolumn{1}{c}{$\lambda _{2}$} & \multicolumn{1}{c}{$\lambda _{3}$} \\ 
80\% CI &   0.780 &   0.825 &   0.787 &   0.593 & \multicolumn{1}{c}{---} &   0.728 &   0.103 &   0.742 & \multicolumn{1}{c}{---} &   0.113  \\
90\% CI &   0.892 &   0.922 &   0.897 &   0.695 & \multicolumn{1}{c}{---} &   0.835 &   0.200 &   0.850 & \multicolumn{1}{c}{---} &   0.218  \\
95\% CI &   0.960 &   0.963 &   0.953 &   0.787 & \multicolumn{1}{c}{---} &   0.900 &   0.328 &   0.940 & \multicolumn{1}{c}{---} &   0.295  \\
 \\ 
\multicolumn{11}{c}{Method of Moments.  Sample Size: 500. (nrep = 400)} \\ 
& \multicolumn{1}{c}{$\beta _{1}$} & \multicolumn{1}{c}{$\beta _{2}$} &
  \multicolumn{1}{c}{$\beta _{3}$} & \multicolumn{1}{c}{$\gamma _{1}$} &
  \multicolumn{1}{c}{$\gamma _{2}$} & \multicolumn{1}{c}{$\gamma _{3}$} &
  \multicolumn{1}{c}{$\gamma _{4}$} & \multicolumn{1}{c}{$\lambda _{1}$} &
  \multicolumn{1}{c}{$\lambda _{2}$} & \multicolumn{1}{c}{$\lambda _{3}$} \\ 
80\% CI &   0.792 &   0.820 &   0.792 &   0.797 & \multicolumn{1}{c}{---} &   0.762 &   0.682 &   0.828 & \multicolumn{1}{c}{---} &   0.730  \\
90\% CI &   0.915 &   0.935 &   0.915 &   0.895 & \multicolumn{1}{c}{---} &   0.867 &   0.792 &   0.922 & \multicolumn{1}{c}{---} &   0.830  \\
95\% CI &   0.955 &   0.958 &   0.958 &   0.940 & \multicolumn{1}{c}{---} &   0.915 &   0.883 &   0.963 & \multicolumn{1}{c}{---} &   0.877  \\
 \\ \multicolumn{11}{c}{Correlated Random Effects.  Sample Size: 1000. (nrep = 400)} \\ 
& \multicolumn{1}{c}{$\beta _{1}$} & \multicolumn{1}{c}{$\beta _{2}$} &
  \multicolumn{1}{c}{$\beta _{3}$} & \multicolumn{1}{c}{$\gamma _{1}$} &
  \multicolumn{1}{c}{$\gamma _{2}$} & \multicolumn{1}{c}{$\gamma _{3}$} &
  \multicolumn{1}{c}{$\gamma _{4}$} & \multicolumn{1}{c}{$\lambda _{1}$} &
  \multicolumn{1}{c}{$\lambda _{2}$} & \multicolumn{1}{c}{$\lambda _{3}$} \\ 
80\% CI &   0.792 &   0.802 &   0.775 &   0.345 & \multicolumn{1}{c}{---} &   0.560 &   0.005 &   0.720 & \multicolumn{1}{c}{---} &   0.007  \\
90\% CI &   0.910 &   0.895 &   0.887 &   0.478 & \multicolumn{1}{c}{---} &   0.718 &   0.015 &   0.863 & \multicolumn{1}{c}{---} &   0.030  \\
95\% CI &   0.963 &   0.963 &   0.943 &   0.590 & \multicolumn{1}{c}{---} &   0.823 &   0.045 &   0.925 & \multicolumn{1}{c}{---} &   0.060  \\
 \\ 
\multicolumn{11}{c}{Method of Moments.  Sample Size: 1000. (nrep = 400)} \\ 
& \multicolumn{1}{c}{$\beta _{1}$} & \multicolumn{1}{c}{$\beta _{2}$} &
  \multicolumn{1}{c}{$\beta _{3}$} & \multicolumn{1}{c}{$\gamma _{1}$} &
  \multicolumn{1}{c}{$\gamma _{2}$} & \multicolumn{1}{c}{$\gamma _{3}$} &
  \multicolumn{1}{c}{$\gamma _{4}$} & \multicolumn{1}{c}{$\lambda _{1}$} &
  \multicolumn{1}{c}{$\lambda _{2}$} & \multicolumn{1}{c}{$\lambda _{3}$} \\ 
80\% CI &   0.845 &   0.815 &   0.792 &   0.810 & \multicolumn{1}{c}{---} &   0.785 &   0.720 &   0.833 & \multicolumn{1}{c}{---} &   0.758  \\
90\% CI &   0.915 &   0.920 &   0.877 &   0.907 & \multicolumn{1}{c}{---} &   0.873 &   0.833 &   0.917 & \multicolumn{1}{c}{---} &   0.833  \\
95\% CI &   0.950 &   0.950 &   0.932 &   0.950 & \multicolumn{1}{c}{---} &   0.900 &   0.900 &   0.963 & \multicolumn{1}{c}{---} &   0.917  \\
 \\ \multicolumn{11}{c}{Correlated Random Effects.  Sample Size: 2000. (nrep = 400)} \\ 
& \multicolumn{1}{c}{$\beta _{1}$} & \multicolumn{1}{c}{$\beta _{2}$} &
  \multicolumn{1}{c}{$\beta _{3}$} & \multicolumn{1}{c}{$\gamma _{1}$} &
  \multicolumn{1}{c}{$\gamma _{2}$} & \multicolumn{1}{c}{$\gamma _{3}$} &
  \multicolumn{1}{c}{$\gamma _{4}$} & \multicolumn{1}{c}{$\lambda _{1}$} &
  \multicolumn{1}{c}{$\lambda _{2}$} & \multicolumn{1}{c}{$\lambda _{3}$} \\ 
80\% CI &   0.780 &   0.802 &   0.795 &   0.135 & \multicolumn{1}{c}{---} &   0.480 &   0.000 &   0.645 & \multicolumn{1}{c}{---} &   0.000  \\
90\% CI &   0.883 &   0.885 &   0.900 &   0.230 & \multicolumn{1}{c}{---} &   0.637 &   0.000 &   0.765 & \multicolumn{1}{c}{---} &   0.000  \\
95\% CI &   0.940 &   0.948 &   0.945 &   0.353 & \multicolumn{1}{c}{---} &   0.735 &   0.000 &   0.875 & \multicolumn{1}{c}{---} &   0.002  \\
 \\ 
\multicolumn{11}{c}{Method of Moments.  Sample Size: 2000. (nrep = 400)} \\ 
& \multicolumn{1}{c}{$\beta _{1}$} & \multicolumn{1}{c}{$\beta _{2}$} &
  \multicolumn{1}{c}{$\beta _{3}$} & \multicolumn{1}{c}{$\gamma _{1}$} &
  \multicolumn{1}{c}{$\gamma _{2}$} & \multicolumn{1}{c}{$\gamma _{3}$} &
  \multicolumn{1}{c}{$\gamma _{4}$} & \multicolumn{1}{c}{$\lambda _{1}$} &
  \multicolumn{1}{c}{$\lambda _{2}$} & \multicolumn{1}{c}{$\lambda _{3}$} \\ 
80\% CI &   0.805 &   0.807 &   0.843 &   0.843 & \multicolumn{1}{c}{---} &   0.815 &   0.745 &   0.860 & \multicolumn{1}{c}{---} &   0.762  \\
90\% CI &   0.915 &   0.927 &   0.945 &   0.930 & \multicolumn{1}{c}{---} &   0.920 &   0.830 &   0.932 & \multicolumn{1}{c}{---} &   0.885  \\
95\% CI &   0.965 &   0.973 &   0.965 &   0.958 & \multicolumn{1}{c}{---} &   0.960 &   0.877 &   0.960 & \multicolumn{1}{c}{---} &   0.922  \\
 \\ \end{tabular}
 }
\par
{\footnotesize Based on 1000 bootstrap replications. $\gamma_2$ and $%
\lambda_2$ are normalized.}
\end{center}
\end{table}

\subsection{Empirical illustration\label{Section: Empirical Illustration}}

In this section, we illustrate the value of the moment conditions derived in
this paper in an empirical illustration inspired by \cite%
{contoyannis_dynamics_2004}. The dependent variable is self-reported health
status. We use data from the first five waves of the British Household
Panel Survey, and we restrict the sample to individuals
who are between 26 and 70 years old in the first wave. This yields a data set with $5093$ individuals observed in 5
time periods, including the initial observation (so $T=4$). In the original
data set, the dependent variable can take five values. We aggregate these
into \textquotedblleft Poor or Very Poor\textquotedblright\ ($8.1\%$ of the
observations), \textquotedblleft Fair\textquotedblright\ ($18.6\%$),
\textquotedblleft Good\textquotedblright\ ($47.6\%$), and \textquotedblleft
Excellent\textquotedblright\ ($25.7\%$). We also consider specifications
where the first two are merged into one outcome.

We use two sets of explanatory variables. In the first, we use age and
age-squared (measured as $Age/10$ and $(Age-45)^{2}/{100}$, respectively,
where $Age$ is measured in years). In the second, we also include log-income.

The results are presented in Table \ref{Table: Empirical Results Version2}, which
also presents the estimates from a correlated random effects specification.
We have normalized the $\gamma$-coefficient associated with
\textquotedblleft Good Health\textquotedblright\  and the threshold ($%
\lambda $) just below \textquotedblleft Good Health\textquotedblright\ to be zero.

The most consistent result presented in Table \ref{Table: Empirical Results Version2} is that the
coefficient on age is negative across all specifications and that the coefficient on age-squared is
insignificantly different from 0 in all specifications. 
The point estimates for the effect of income on self-reported health are
positive and statistically significant for all the specifications. The most puzzling aspect of Table 
\ref {Table: Empirical Results Version2} is that the standard error of the correlated random effects estimator of
the coefficient on age gets much more precise when log-income is included. The method of moments estimator does not display this pattern. 
A comparison of the other standard errors reveals that the correlated random effects estimator is less variable 
than the method of moments estimator. This is in line with the interquartile ranges reported in Tables \ref{Design A}, \ref{Design B}\ and \ref{Design C}.

\begin{table}[h!]
\caption{Empirical Results: Application to Self-Reported Health Status}
\label{Table: Empirical Results Version2}
\begin{center}
{\footnotesize
\begin{tabular}{lrrrrcrrrr}
& \multicolumn{4}{c}{Three Outcomes} &  & \multicolumn{4}{c}{Four Outcomes}
\\
& \multicolumn{1}{c}{MoM} & \multicolumn{1}{c}{CRE} & \multicolumn{1}{c}{MoM}
& \multicolumn{1}{c}{CREs} & \multicolumn{1}{c}{} & \multicolumn{1}{c}{MoM}
& \multicolumn{1}{c}{CRE} & \multicolumn{1}{c}{MoM} & \multicolumn{1}{c}{CRE}
\\
       $Age/10$ & $ -1.070 $ & $ -1.110 $ & $ -1.271 $ & $ -0.179 $ &  & $ -1.136 $ & $ -0.961 $ & $ -1.078 $ & $ -0.147 $ \\
                 & $( 0.296\rlap{)}$ & $( 0.145\rlap{)}$ & $( 0.470\rlap{)}$ & $( 0.035\rlap{)}$ &  & $( 0.214\rlap{)}$ & $( 0.137\rlap{)}$ & $( 0.331\rlap{)}$ & $( 0.032\rlap{)}$ \vspace{5pt} \\
{\tiny ($Age-45)^2/100$} & $  0.396 $ & $ -0.251 $ & $  0.804 $ & $  0.109 $ &  & $  0.278 $ & $ -0.226 $ & $ -0.293 $ & $  0.025 $ \\
                 & $( 0.925\rlap{)}$ & $( 0.548\rlap{)}$ & $( 2.047\rlap{)}$ & $( 0.544\rlap{)}$ &  & $( 0.749\rlap{)}$ & $( 0.503\rlap{)}$ & $( 1.037\rlap{)}$ & $( 0.499\rlap{)}$ \vspace{5pt} \\
      log-income &          &          & $  0.054 $ & $  0.283 $ &  &          &          & $  0.130 $ & $  0.217 $ \\
                 &          &          & $( 0.198\rlap{)}$ & $( 0.048\rlap{)}$ &  &          &          & $( 0.089\rlap{)}$ & $( 0.045\rlap{)}$ \vspace{25pt} \\
      $\gamma_1$ &          &          &          &          &  & $ -1.477 $ & $ -1.594 $ & $ -1.427 $ & $ -1.674 $ \\
                 &          &          &          &          &  & $( 0.179\rlap{)}$ & $( 0.109\rlap{)}$ & $( 0.247\rlap{)}$ & $( 0.111\rlap{)}$ \vspace{5pt} \\
      $\gamma_2$ & $  0.269 $ & $ -0.810 $ & $ -0.753 $ & $ -0.824 $ &  & $ -0.759 $ & $ -0.618 $ & $ -0.691 $ & $ -0.657 $ \\
                 & $( 0.165\rlap{)}$ & $( 0.067\rlap{)}$ & $( 0.354\rlap{)}$ & $( 0.068\rlap{)}$ &  & $( 0.100\rlap{)}$ & $( 0.063\rlap{)}$ & $( 0.116\rlap{)}$ & $( 0.064\rlap{)}$ \vspace{5pt} \\
      $\gamma_3$ & $  0.000 $ & $  0.000 $ & $  0.000 $ & $  0.000 $ &  & $  0.000 $ & $  0.000 $ & $  0.000 $ & $  0.000 $ \\
                 & $( 0.000\rlap{)}$ & $( 0.000\rlap{)}$ & $( 0.000\rlap{)}$ & $( 0.000\rlap{)}$ &  & $( 0.000\rlap{)}$ & $( 0.000\rlap{)}$ & $( 0.000\rlap{)}$ & $( 0.000\rlap{)}$ \vspace{5pt} \\
      $\gamma_4$ & $  0.147 $ & $  0.484 $ & $  0.559 $ & $  0.590 $ &  & $  0.395 $ & $  0.557 $ & $  0.635 $ & $  0.619 $ \\
                 & $( 0.133\rlap{)}$ & $( 0.063\rlap{)}$ & $( 0.619\rlap{)}$ & $( 0.064\rlap{)}$ &  & $( 0.105\rlap{)}$ & $( 0.064\rlap{)}$ & $( 0.236\rlap{)}$ & $( 0.065\rlap{)}$ \vspace{25pt} \\
     $\lambda_1$ &          &          &          &          &  & $ -2.464 $ & $ -2.560 $ & $ -2.510 $ & $ -2.555 $ \\
                 &          &          &          &          &  & $( 0.057\rlap{)}$ & $( 0.050\rlap{)}$ & $( 0.080\rlap{)}$ & $( 0.050\rlap{)}$ \vspace{5pt} \\
     $\lambda_2$ & $  0.000 $ & $  0.000 $ & $  0.000 $ & $  0.000 $ &  & $  0.000 $ & $  0.000 $ & $  0.000 $ & $  0.000 $ \\
                 & $( 0.000\rlap{)}$ & $( 0.000\rlap{)}$ & $( 0.000\rlap{)}$ & $( 0.000\rlap{)}$ &  & $( 0.000\rlap{)}$ & $( 0.000\rlap{)}$ & $( 0.000\rlap{)}$ & $( 0.000\rlap{)}$ \vspace{5pt} \\
     $\lambda_3$ & $  3.707 $ & $  3.818 $ & $  3.304 $ & $  3.780 $ &  & $  3.374 $ & $  3.758 $ & $  3.340 $ & $  3.728 $ \\
                 & $( 0.085\rlap{)}$ & $( 0.054\rlap{)}$ & $( 0.203\rlap{)}$ & $( 0.054\rlap{)}$ &  & $( 0.060\rlap{)}$ & $( 0.053\rlap{)}$ & $( 0.082\rlap{)}$ & $( 0.054\rlap{)}$ \vspace{5pt} \\
\\ \hline
\end{tabular}
}
\end{center}
\par
{\footnotesize \singlespacing
The dependent variable is self-reported health status and the data is a
balanced panel from the first five waves of the British Household Panel
Survey. We restrict the sample to individuals who are between 26 and 70 years old in the first wave. In columns 1-4, the dependent variable takes the values ``Poor, Very Poor, or Fair'', ``Good'', and
``Excellent''.
The dependent variable in columns 5-8
takes values ``Poor or Very Poor'',  ``Fair'',
``Good'' and ``Excellent''. MoM refers to the method of moments estimator and CRE is the
correlated random effects estimator.}
\end{table}

\IgnoreThisText{
\begin{table}[h!]
\caption{Empirical Results}
\label{Table: Empirical Results}\footnotesize \centering
\begin{tabular}{lrrcrr}
&\multicolumn{2}{c}{Four Outcomes}&\qquad &\multicolumn{2}{c}{Three Outcomes}\\
$Age/10$ & -1.214 & -1.829 & & -0.980 & -1.261 \\
\vspace{2pt} & (0.306\rlap{)} & (0.309\rlap{)} & & (0.488\rlap{)} & (0.475%
\rlap{)}\vspace{1pt} \\
$(Age-45)^2/1000$\qquad & -0.200 & -0.261 & & -0.304 & -0.236 \\
\vspace{2pt} & (0.074\rlap{)} & (0.080\rlap{)} & & (0.108\rlap{)} & (1.411%
\rlap{)}\vspace{1pt} \\
log-income &  & 0.063 & & & 0.196 \\
\vspace{2pt} &  & (0.070\rlap{)} & & & (0.117\rlap{)} \\
$\delta _{1}$ & -0.310 & -0.339 &  &  \\
\vspace{2pt} & (0.166\rlap{)} & (0.146\rlap{)} &  & \vspace{1pt} \\
$\delta _{2}$ & -0.313 & -0.338 & & -0.535 & -0.489 \\
\vspace{2pt} & (0.130\rlap{)} & (0.129\rlap{)} & & (0.338\rlap{)} & (0.239%
\rlap{)}\vspace{1pt} \\
$\delta _{3}$ & 0.000 & 0.000 & & 0.000 & 0.000 \\
\vspace{2pt} & (0.000\rlap{)} &  (0.000\rlap{)} & & (0.000\rlap{)} & (0.000%
\rlap{)}\vspace{1pt} \\
$\delta _{4}$ & -0.180 & -0.263 & & 0.127 & 0.021 \\
\vspace{2pt} & (0.128\rlap{)} & (0.113\rlap{)} & & (0.146\rlap{)} & (0.115%
\rlap{)}\vspace{1pt} \\
$\lambda _{1}$ & -2.756 & -3.050 &  &  \\
\vspace{2pt} & (0.086\rlap{)} & (0.101\rlap{)} &  & \vspace{1pt} \\
$\lambda _{2}$ & 0.000 & 0.000 & & 0.000 & 0.000 \\
\vspace{2pt} & (0.000\rlap{)} & (0.000\rlap{)} & & (0.000\rlap{)} & (0.000%
\rlap{)}\vspace{1pt} \\
$\lambda _{3}$ & 3.693 & 4.137 & & 3.395 & 3.440 \\
& (0.126\rlap{)} & (0.128\rlap{)} & & (0.155\rlap{)} & (0.120\rlap{)}%
\end{tabular}
\end{table}  
}

\section{Conclusions}
\label{Conclusions}

This paper has extended the analysis in \cite{honore2020dynamic} to provide conditional moment conditions for panel data fixed effects versions of the dynamic ordered logit models like the one considered in \cite{muris2020dynamic}. The moment conditions are interesting in their own right, and the paper also illustrates the potential for systematically deriving moment conditions for nonlinear panel models. The moment conditions presented here can be used for estimation as well as for testing more parametric specifications of the individual-specific effects in dynamic ordered logits. For point-identification, it is important to investigate whether the moment conditions are \textit{uniquely} satisfied at the true parameter values. The paper presents conditions under which this is the case. 
The paper also proposes a practical strategy for turning the derived conditional moment conditions into unconditional moment conditions that can be used for GMM estimation, and it illustrates the use of the resulting estimator in a small Monte Carlo study as well as in an empirical application.
 
More broadly, this paper  contributes to the literature on panel data estimation of nonlinear models with fixed effects. In this context, the main contribution is to illustrate the potential for applying the functional differencing insights of \cite{bonhomme2012functional} to logit-type models.

\ifx\undefined\BySame
\newcommand{\BySame}{\leavevmode\rule[.5ex]{3em}{.5pt}\ }
\fi
\ifx\undefined\textsc
\newcommand{\textsc}[1]{{\sc #1}}
\newcommand{\emph}[1]{{\em #1\/}}
\let\tmpsmall\small
\renewcommand{\small}{\tmpsmall\sc}
\fi

\newpage
\appendix
\section{Appendix}
\allowdisplaybreaks

\subsection{Proof of Theorem~\ref{th:momentsT3} and \ref{th:momentsTgeneralA}}

We first want to establish Lemma~\ref{lemma:MASTER} below, which is key to proving the main text theorems.
In order to state the lemma, we require some additional notation.
Recall that $Q \in \{2,3,\ldots\}$ is the number of values that the observed outcomes $Y_{it}$  can take.
Let  $  \widetilde Y_1, \widetilde  Y_3 \in \{0,1\}   $,  $\widetilde  Y_2 \in \{1,2,3\} $, and $W \in \{1,2,\ldots,Q\}$
be random variables, and let $\widetilde Y = (\widetilde Y_1,\widetilde Y_2,\widetilde Y_3)$.
For the joint distribution of $\widetilde Y$ and $W$ we write
\begin{align*}
    p(\widetilde y, w ) := {\rm Pr}\left(\widetilde Y = \widetilde y \; \; \& \; \;  W = w \right) 
\end{align*}
and we assume that
\begin{align}
  p(\widetilde y, w ) &=  p_3(\widetilde y_3 \,| \,\widetilde y_2, w) \; \,  p_2(\widetilde y_2 \,| \,w) \;\,  f(w \,|\, \widetilde y_1)   \;\,  p_1(\widetilde y_1)   ,
     \label{AssMarkov}
\end{align}
where $\widetilde y = (\widetilde y_1,\widetilde y_2,\widetilde y_3)$, and
\begin{align*}
    p_1(\widetilde y_1) &:=  {\rm Pr}\left( \widetilde Y_1 = \widetilde y_1  \right) ,
 \\   
     f(w \, |\, \widetilde y_1 ) & := \Pr \left(W=w \, \Big| \,  \widetilde Y_1 = \widetilde y_1 \right) ,
\\
   p_2( \widetilde y_2 \, |\, w ) &:=  {\rm Pr}\left( \widetilde Y_2 =  \widetilde y_2  \, \Big| \,   W  = w   \right) ,
\\
   p_3( \widetilde y_3 \, |\, \widetilde y_2, w ) &:=  {\rm Pr}\left( \widetilde Y_3 =  \widetilde y_3  \, \Big| \,    \widetilde Y_2 =  \widetilde y_2 , \, W=w  \right) .
\end{align*}
We do not impose any  assumptions on the 
transition probabilities $f(w \, |\, \widetilde y_1 ) $,  $ p_3(\widetilde y_3 \,| \,1,w) $,
and $ p_3(\widetilde y_3 \,| \,3,w) $. All the other transition  probabilities are assumed to follow an (ordered) logit model:
\begin{align}
     p_1(\widetilde y_1  ) &=  \left\{
    \begin{array}{ll}    
      1 -  \Lambda(  \pi_{1}  )   & \text{if $\widetilde y_1=0$,} 
      \\
      \Lambda(  \pi_{1}  )    & \text{if $\widetilde y_1=1$,}
    \end{array}
    \right.  
\nonumber \\[10pt] 
    p_2(\widetilde y_2 \,|\, w) &=  \left\{
    \begin{array}{ll}    
      1 -  \Lambda[    \pi_{2,1}(w)  ]   & \text{if $\widetilde y_2=1$,} 
      \\
      \Lambda[   \pi_{2,1}(w)  ] -  \Lambda[  \pi_{2,2}(w)  ]   & \text{if $\widetilde y_2=2$,}
      \\
      \Lambda[   \pi_{2,2}(w)  ]    & \text{if $\widetilde y_2=3$,}
    \end{array}
    \right.  
\nonumber \\[10pt] 
     p_3(\widetilde y_3 \,| \, 2 , \, w ) &=  \left\{
    \begin{array}{ll}    
      1 -  \Lambda(   \pi_{3}  )   & \text{if $\widetilde y_3=0$,} 
      \\
      \Lambda(   \pi_{3}  )    & \text{if $\widetilde y_3=1$,}
    \end{array}
    \right.  
  \label{AssLogit}  
\end{align}
where  $ \Lambda(\xi) :=  [1+\exp(-\xi)]^{-1}$ is the cumulative distribution function of the logistic distribution,
$\pi_1, \pi_3 \in \mathbb{R}$ are constants,
and  $ \pi_{2,1}, \pi_{2,2} \,: \, \{1,2,\ldots,Q\} \rightarrow \mathbb{R}$ are functions
such that $\pi_{2,1}(w) \geq \pi_{2,2}(w)$ for all $w \in  \{1,2,\ldots,Q\}$.
Notice that $ p_3(\widetilde y_3 \,| \, 2 , \, w ) $ does not depend on $w$.
Finally, we define
$m : \{0,1\} \times \{1,2,3\} \times \{0,1\} \times \{1,2,\ldots,Q\}  \rightarrow \mathbb{R}$ 
by
 \begin{align}
  m(\widetilde y, w)   &:= 
    \left\{  \begin{array}{ll}
           \exp\left(   \pi_1  - \pi_3  \right)     \frac{ \displaystyle    \exp[\pi_3 - \pi_{2,2}(w)] -  1   }
             {\displaystyle \exp\left[  \pi_{2,1}(w) - \pi_{2,2}(w)      \right]  - 1  }  & \text{if }  \widetilde y_1 = 0, \, \widetilde y_2=2 , \, \widetilde y_3=0, 
            \\[10pt]
         \exp\left(   \pi_1  -  \pi_3  \right)    \frac{ \displaystyle     1 - \exp[\pi_{2,2}(w) - \pi_3]  }
             {\displaystyle 1 -  \exp\left[  \pi_{2,2}(w)  -  \pi_{2,1}(w)  \right]  }  & \text{if }  \widetilde y_1 = 0, \, \widetilde y_2=2 , \, \widetilde y_3=1,            
            \\[10pt]
             \exp\left(  \pi_1 -  \pi_3   \right) & \text{if }  \widetilde y_1 = 0, \,  \widetilde y_2=3 ,
            \\[10pt]
         -1  & \text{if }  \widetilde y_1 = 1, \, \widetilde y_2=1 ,                        
            \\[10pt]
 -   \frac{ \displaystyle     1 - \exp[\pi_3 - \pi_{2,1}(w)  ]   }
             {\displaystyle 1 -  \exp\left[  \pi_{2,2}(w)  -  \pi_{2,1}(w)  \right]  }  & \text{if }  \widetilde y_1 = 1, \, \widetilde y_2=2 , \, \widetilde y_3=0, 
            \\[10pt]
 -   \frac{\displaystyle     \exp[  \pi_{2,1}(w) - \pi_3 ] - 1    }
             {\displaystyle \exp\left[    \pi_{2,1}(w) -  \pi_{2,2}(w) \right] - 1   }
               & \text{if }  \widetilde y_1 = 1, \, \widetilde y_2=2 , \, \widetilde y_3=1, 
            \\[10pt]
                      0 & \text{otherwise}.
       \end{array} \right.
     \label{DefMgeneral}  
\end{align}

\begin{lemma}
    \label{lemma:MASTER}
    Let $\pi_1, \pi_3 \in \mathbb{R}$, and $ \pi_{2,1}, \pi_{2,2} \,: \, \{1,2,\ldots,Q\} \rightarrow \mathbb{R}$
   be such that $\pi_{2,1}(w) \geq \pi_{2,2}(w)$, for all $w \in  \{1,2,\ldots,Q\}$.
    Let the random variables $\widetilde Y \in \{0,1\} \times \{1,2,3\} \times \{0,1\}$ and $W \in \{1,2,\ldots,Q\}$
      be such that their distributions satisfy  \eqref{AssMarkov} and \eqref{AssLogit},
      and let $m : \{0,1\} \times \{1,2,3\} \times \{0,1\} \times \{1,2,\ldots,Q\}  \rightarrow \mathbb{R}$  be defined by \eqref{DefMgeneral}.      
       Then  we have
     \begin{align*}
          \mathbb{E}\left[ m(\widetilde Y, W)   \right] = 0.
     \end{align*}
\end{lemma}

\begin{proof}
     Define
          \begin{align*}
           g(\widetilde y_1,w)
         &  :=    \mathbb{E}\left[ m(\widetilde Y, W) \, \Big| \, \widetilde Y_1 = \widetilde y_1, \, W=w  \right] 
         \\
           &=  \sum_{\widetilde y_2 \in \{1,2,3\}} 
            \,  \sum_{\widetilde y_3 \in \{0,1\}} \, m(\widetilde y, w)  \; \, p_3(\widetilde y_3 \,| \,\widetilde y_2,w) \; \,  p_2(\widetilde y_2 \,| \,w ) ,
     \end{align*}
     where $\widetilde y = ( \widetilde y_1, \widetilde y_2, \widetilde y_3)$.
      Using  the expressions for  $ p_2(\widetilde y_2 \,| \,w) $, 
        $p_3(\widetilde y_3 \,| \,\widetilde y_2,w)$, and $m(\widetilde y, w) $
      in \eqref{AssLogit} and \eqref{DefMgeneral}      
      one finds that  for $ \widetilde y_1 = 1$ we have 
      \begin{align}
             g(1,w)&=
           -  \left\{ 1 -  \Lambda[    \pi_{2,1}(w)  ]  \right\}
          -  \left\{  \Lambda[   \pi_{2,1}(w)  ] -  \Lambda[  \pi_{2,2}(w)  ]   \right\}  \times
       \nonumber   \\ & \qquad  \qquad \times
         \underbrace{
           \left(
              \frac{ \left[  1 -  \Lambda(   \pi_{3}  ) \right] \left\{    1 - \exp[\pi_3 - \pi_{2,1}(w)  ] \right\}   }
             { 1 -  \exp\left[  \pi_{2,2}(w)  -  \pi_{2,1}(w)  \right]  }  
           +    \frac{ \Lambda(   \pi_{3}  )  \left\{    \exp[  \pi_{2,1}(w) - \pi_3 ] - 1  \right\}   }
             {  \exp\left[    \pi_{2,1}(w) -  \pi_{2,2}(w) \right] - 1   }   
             \right) }_{\displaystyle = \frac{\Lambda[\pi_{2,1}(w)] - \Lambda(\pi_3)} {  \Lambda[   \pi_{2,1}(w)  ] -  \Lambda[  \pi_{2,2}(w)  ]   } }
           \nonumber   \\
          &=      -  \left[ 1  -  \Lambda(\pi_3)  \right] , 
           \label{ResultForG1}
      \end{align}
      and analogously one calculates for $ \widetilde y_1 = 0$ that
      \begin{align}
             g(0,w ) &= 
           \exp\left(   \pi_1  - \pi_3  \right)  \Lambda(\pi_3)   .
                      \label{ResultForG0} 
      \end{align}
     Notice that $g(\widetilde y_1,w) $ therefore does not depend on $w$, so we can simply write 
      $g(\widetilde y_1) := g(\widetilde y_1,w) $ in the following.
      Using \eqref{ResultForG1}, \eqref{ResultForG0}, and the expression for $p_1(\widetilde y_1 ) $ in \eqref{AssLogit} 
      we obtain that 
     $$
         \sum_{\widetilde y_1 \in \{0,1\}}   g(\widetilde y_1 ) \, p_1(\widetilde y_1 )  = 0 .
     $$
     Together with $\sum_{w \in \{1,\ldots,Q\}} f(w \,|\, \widetilde y_1 ) = 1$, this gives
     \begin{align*}
           \mathbb{E}\left[ m(\widetilde Y, W)   \right]
      &=      
            \sum_{\widetilde y_1 \in \{0,1\}} 
          \,  \sum_{w \in \{1,\ldots,Q\}}
           \,  \sum_{\widetilde y_2 \in \{1,2,3\}} 
            \,  \sum_{\widetilde y_3 \in \{0,1\}} \, m(\widetilde y, w)  \, p(\widetilde y, w )  
       \\
        &=       \sum_{\widetilde y_1 \in \{0,1\}} 
          \,  \sum_{w \in \{1,\ldots,Q\}}    g(\widetilde y_1 ) \,   f(w \,|\, \widetilde y_1 )   \;  p_1(\widetilde y_1 ) 
       \\
         &=     \sum_{\widetilde y_1 \in \{0,1\}} 
          \,    g(\widetilde y_1 )  \underbrace{\left[ \sum_{w \in \{1,\ldots,Q\}}   f(w \,|\, \widetilde y_1 )  \right]}_{=1}      p_1(\widetilde y_1 )           
        \\
         &=     \sum_{\widetilde y_1 \in \{0,1\}}   g(\widetilde y_1 ) \, p_1(\widetilde y_1  )  = 0 ,
     \end{align*}
     which is what we wanted to show.
\end{proof}

The following lemma is similar to Lemma~\ref{lemma:MASTER} above, but the random variables 
$\widetilde Y_1,\widetilde Y_2,\widetilde Y_3$ and their distributional assumptions are now different, and the lemma should be 
understood independently from any notation established above.
 
\begin{lemma}
\label{lemma:MASTER2}
Let  $  \widetilde Y_1,   \widetilde Y_2,   \widetilde Y_3 \in \{0,1\}   $ and  $W,V \in \{1,\ldots,Q\}$ 
be random variables such that the joint distribution of $\widetilde Y = (\widetilde Y_1,\widetilde Y_2,\widetilde Y_3)$, $W$, and $V$
satisfies 
\begin{align*}
    {\rm Pr}\left( \widetilde Y  = \widetilde y  \; \; \& \; \;   W = w  \; \; \& \; \;   V =v \right) 
     &=  p_3(\widetilde y_3 \,| \,v) \; \,  g(v \,|\, \widetilde y_2,w) \; \, p_2(\widetilde y_2 \,| \,w) \;\,  f(w \,|\, \widetilde y_1)   \;\,  p_1(\widetilde y_1)   ,
\end{align*}
where $\widetilde y = (\widetilde y_1,\widetilde y_2,\widetilde y_3)$, and the functions
$p_3$, $g$, $p_2$, and $f$ are appropriate conditional probabilities,
while $p_1(\widetilde y_1)$  is the marginal distribution of $ \widetilde Y_1$.
For $p_1(\widetilde y_1)$,  $p_2(\widetilde y_2 \,| \,w)$, and $p_3(\widetilde y_3 \,| \,v) $ we assume logistic binary choice models:
\begin{align*}
     p_1(\widetilde y_1  ) &=   \Lambda\left[ (2 \,\widetilde y_1-1) \, \pi_{1}    \right]  ,
&
    p_2(\widetilde y_2 \,|\, w) &=  \Lambda\left[ (2 \,\widetilde y_2-1) \, \pi_2(w)    \right] ,
&
     p_3(\widetilde y_3 \,| \, v ) &=    \Lambda\left[ (2 \,\widetilde y_3-1) \, \pi_3(v)    \right] ,
\end{align*}
where  
$\pi_1 \in \mathbb{R}$ is a constant,
and  $ \pi_{2}, \pi_{3} : \,\{1,\ldots,Q\}\rightarrow \mathbb{R}$  are functions.
The only assumption that we impose on $f(w \,|\, \widetilde y_1) $ and $g(v \,|\, \widetilde y_2,w)$ is that
$
      g(v \,|\, 1,w) = g(v \,|\, 1) 
$;
that is, conditional on $\widetilde Y_2=1$ the distribution of $V$ is independent of $W$.
Furthermore, let
$m : \{0,1\}^3 \times \{1,\ldots,Q\}^2  \rightarrow \mathbb{R}$ 
be given by
 \begin{align*}
  m(\widetilde y, w,v)   &:= 
    \left\{  \begin{array}{ll}
           \exp\left[   \pi_1  - \pi_2(w)  \right]     & \text{if }   \widetilde y =(0,1,0), 
            \\
         \exp\left[   \pi_1  -  \pi_3(v)  \right]     & \text{if }  \widetilde y =(0,1,1),        
            \\
         -1  & \text{if }  (\widetilde y_1, \widetilde y_2)=(1,0) ,                        
            \\
     \exp\left[    \pi_3(v)  - \pi_2(w) \right]  - 1  & \text{if }  \widetilde y =(1,1,0) , 
            \\
                      0 & \text{otherwise}.
       \end{array} \right.
\end{align*}
       We then have
     \begin{align*}
          \mathbb{E}\left[ m(\widetilde Y, W,V)   \right] = 0.
     \end{align*}
\end{lemma}

\begin{proof}
    This lemma is a restatement 
    of Lemma 6 in the
    2021 arXiv version of 
    \cite{honore2020dynamic}, 
    and the proof can be found there.
\end{proof}

\bigskip

Using Lemma~\ref{lemma:MASTER} and~\ref{lemma:MASTER2}, we are now ready to prove the two main text theorems.

\bigskip

\begin{proof}[\bf Proof of Theorem~\ref{th:momentsT3}]
We consider the three cases $q_2 \in \{2,\ldots,Q-1\}$, $q_2=Q$, and $q_2=1$ separately.

  \underline{Case $q_2 \in \{2,\ldots,Q-1\}$:} \; 
   In this case, we define
   \begin{align*}
           W &:= Y_1 ,
       &
        \widetilde Y_1 &:=   \mathbbm{1}\left\{  Y_1 > q_1 \right\} ,
       &
        \widetilde Y_3 &:=  \mathbbm{1}\left\{  Y_3 > q_3 \right\} ,
         &
        \widetilde Y_2 &:= \left\{   
              \begin{array}{ll}  1 & \text{if } Y_2 < q_2 ,
                           \\
                                     2 & \text{if } Y_2 = q_2 ,
                           \\
                                     3 & \text{if } Y_2 > q_2 .
           \end{array} \right.
   \end{align*}
   Our ordered logit model  in \eqref{model} then implies that the joint distribution of
   $\widetilde Y = (\widetilde Y_1,\widetilde Y_2,\widetilde Y_3)$ and $W$ conditional on
     $A=\alpha$, $Y_0=y_0$, $X=(x_1,x_2,x_3)$,  and  $\theta=\theta^0$
   satisfies \eqref{AssMarkov} and \eqref{AssLogit}, as long as we choose
\begin{align*}
      f(y_1 \,|\, 1 ) &= {\rm Pr}\left(  Y_1=y_1\, \big| \, Y_1 > q_1 , \, Y_0=y_0, \, X=x, \, A=\alpha \right) ,
     \\
      f(y_1 \,|\, 0 ) &= {\rm Pr}\left(  Y_1=y_1\, \big| \, Y_1 \leq q_1 , \, Y_0=y_0, \, X=x, \, A=\alpha \right) ,
     \\
      p_3(\widetilde y_3 \,| \, 1, \, y_1) &=  \left\{
    \begin{array}{ll}    
       {\rm Pr}\left(  Y_3 \leq q_3  \, \big| \, Y_2 < q_2  , \, Y_1=y_1, \, Y_0=y_0, \, X=x, \, A=\alpha \right)   & \text{if $\widetilde y_3=0$,} 
      \\
     {\rm Pr}\left(  Y_3 > q_3  \, \big| \, Y_2 < q_2  , \, Y_1=y_1, \, Y_0=y_0, \, X=x, \, A=\alpha \right)    & \text{if $\widetilde y_3=1$,}
    \end{array}
    \right.      
     \\
      p_3(\widetilde y_3 \,| \, 3, \, y_1) &=  \left\{
    \begin{array}{ll}    
       {\rm Pr}\left(  Y_3 \leq q_3  \, \big| \, Y_2 > q_2  , \, Y_1=y_1, \, Y_0=y_0, \, X=x, \, A=\alpha \right)   & \text{if $\widetilde y_3=0$,} 
      \\
     {\rm Pr}\left(  Y_3 > q_3  \, \big| \, Y_2 > q_2  , \, Y_1=y_1, \, Y_0=y_0, \, X=x, \, A=\alpha \right)    & \text{if $\widetilde y_3=1$,}
    \end{array}
    \right.      
\end{align*}
and
\begin{align*}       
      \pi_{1} &= \alpha + z(y_0,x_{1},\theta^0)  - \lambda_{q_1} =   \alpha +  x'_1 \, \beta^0 + \gamma^0_{y_0}  - \lambda_{q_1} ,
      \\
      \pi_{2,1}(y_1) &= \alpha + z(y_1,x_{2},\theta^0)  - \lambda_{q_2-1} =  \alpha +   x'_2 \, \beta^0 + \gamma^0_{y_1} -   \lambda_{q_2-1},
      \\
       \pi_{2,2}(y_1) &= \alpha + z(y_1,x_{2},\theta^0) - \lambda_{q_2}  =  \alpha + x'_2 \, \beta^0 + \gamma^0_{y_1} - \lambda_{q_2} ,
     \\
       \pi_3 &=  \alpha +  z(q_2,x_{3},\theta^0) -  \lambda_{q_3} = \alpha +   x'_3 \, \beta^0 + \gamma^0_{q_2} - \lambda_{q_3} ,
\end{align*}
where $w=y_1$, and $z(y_{t-1},x_t,\theta) $ is defined in \eqref{DefSingleIndex}.
Plugging those expressions for $\pi_{1}$, $  \pi_{2,1}(w)$,  $\pi_{2,2}(w)$ and $\pi_3$
into the moment function $m(\widetilde y, w)$   in \eqref{DefMgeneral}, we find that this moment function exactly 
coincides with $m_{y_0,q_1,q_2,q_3}(y,x,\theta^0) $ in equation \eqref{MomentFunctionsT3} of the main text.
Thus, by applying Lemma~\ref{lemma:MASTER} to the distribution of $Y$ conditional on  $A=\alpha$, $Y_0=y_0$, and $X=x$
(the lemma does not feature those conditioning variables, which is why we are applying the lemma to the conditional distribution), we obtain
\begin{align*}
\mathbb{E} \left[ m_{y_0,q_1,q_2,q_3}(Y,X,\theta^0) \, \big| \, Y_0=y_0, \, X=x, \, A=\alpha \right] &= 0 ,
\end{align*}
which concludes the proof for the case $q_2 \in \{2,\ldots,Q-1\}$.

\medskip

   \underline{Case $q_2 =Q $:} \; 
   In this case, we choose
   \begin{align*}
           W &:= Y_1 ,
           &
           V &:= Y_2 ,           
       &
        \widetilde Y_1 &:=   \mathbbm{1}\left\{  Y_1 > q_1 \right\} ,
       &
        \widetilde Y_2 &:=  \mathbbm{1}\left\{  Y_2 = Q \right\} ,
       &
        \widetilde Y_3 &:=  \mathbbm{1}\left\{  Y_3 > q_3 \right\} .
   \end{align*}
   Our ordered logit model  in \eqref{model} then implies that the joint distribution of
   $\widetilde Y = (\widetilde Y_1,\widetilde Y_2,\widetilde Y_3)$, $W$, and $V$ conditional on
     $A=\alpha$, $Y_0=y_0$, $X=(x_1,x_2,x_3)$,  and  $\theta=\theta^0$
   satisfies  the assumptions of Lemma~\ref{lemma:MASTER2}, as long as we choose
\begin{align*}
      f(y_1 \,|\, 1 ) &= {\rm Pr}\left(  Y_1=y_1\, \big| \, Y_1 > q_1 , \, Y_0=y_0, \, X=x, \, A=\alpha \right) ,
     \\
      f(y_1 \,|\, 0 ) &= {\rm Pr}\left(  Y_1=y_1\, \big| \, Y_1 \leq q_1 , \, Y_0=y_0, \, X=x, \, A=\alpha \right) ,
     \\
      g(y_2 \,|\, 1) &=  \mathbbm{1}\left\{  y_2 = Q \right\} ,
     \\
      g(y_2 \,|\, 0, y_1) &= {\rm Pr}\left(  Y_2=y_2\, \big| \, Y_2 < Q , \, Y_1=y_1, \, X=x, \, A=\alpha \right) ,
   \end{align*}
and
\begin{align*}       
      \pi_{1} &= \alpha + z(y_0,x_{1},\theta^0)  - \lambda_{q_1} =   \alpha +  x'_1 \, \beta^0 + \gamma^0_{y_0}  - \lambda_{q_1} ,
      \\
      \pi_{2}(y_1) &= \alpha + z(y_1,x_{2},\theta^0)  - \lambda_{Q-1} =  \alpha +   x'_2 \, \beta^0 + \gamma^0_{y_1} -   \lambda_{Q-1},
     \\
       \pi_3(y_2) &=  \alpha +  z(y_2,x_{3},\theta^0) -  \lambda_{q_3} = \alpha +   x'_3 \, \beta^0 + \gamma^0_{y_2} - \lambda_{q_3} ,
\end{align*}
where $w=y_1$ and $v=y_2$, and $z(y_{t-1},x_t,\theta) $ is defined in \eqref{DefSingleIndex}.
Plugging those expressions for $\pi_{1}$, $  \pi_{2}(y_1)$,  and $\pi_3(y_2)$
into the moment function $m(\widetilde y, w,v)$   in  Lemma~\ref{lemma:MASTER2},  we find that this moment function exactly 
coincides with $m_{y_0,q_1,Q,q_3}(y,x,\theta^0) $ in equation \eqref{MomentFunctionsT3c} of the main text.
Thus, by applying Lemma~\ref{lemma:MASTER2} to the distribution of $Y$ conditional on  $A=\alpha$, $Y_0=y_0$, and $X=x$
(the lemma does not feature those conditioning variables, which is why we are applying the lemma to the conditional distribution), we obtain
\begin{align*}
\mathbb{E} \left[ m_{y_0,q_1,Q,q_3}(Y,X,\theta^0) \, \big| \, Y_0=y_0, \, X=x, \, A=\alpha \right] &= 0 ,
\end{align*}
which concludes the proof for the case $q_2=Q$.

\medskip

   \underline{Case $q_2 =1 $:} \; 
   The result for this case follows from the result for $q_2=Q$ by applying the transformation $Y_t \mapsto Q+1-Y_t $,
   $\lambda_q \mapsto - \lambda_{Q-q}$, $\beta \mapsto - \beta$, $\gamma_q \mapsto - \gamma_{Q+1-Y_t}$, $A_i \mapsto - A_i$.
   This transformation leaves the model probabilities in \eqref{DefProb} unchanged but transforms the moment function in \eqref{MomentFunctionsT3c}
   into the one in \eqref{MomentFunctionsT3b}, implying that this is also a valid moment function.
\end{proof}

\bigskip

\begin{proof}[\bf Proof of Theorem~\ref{th:momentsTgeneralA}]
As was the case in the proof of Theorem~\ref{th:momentsT3}, we consider the three cases $q_2 \in \{2,\ldots,Q-1\}$, $q_2=Q$, and $q_2=1$ separately.

   \underline{Case $q_2 \in \{2,\ldots,Q-1\}$:} \; 
   In this case, we define
   \begin{align*}
           W &:= Y_{s-1} ,
       &
        \widetilde Y_1 &:=   \mathbbm{1}\left\{  Y_t > q_1 \right\} ,
       &
        \widetilde Y_3 &:=  \mathbbm{1}\left\{  Y_{s+1} > q_3 \right\} ,
         &
        \widetilde Y_2 &:= \left\{   
              \begin{array}{ll}  1 & \text{if } Y_s < q_2 ,
                           \\
                                     2 & \text{if } Y_s = q_2 ,
                           \\
                                     3 & \text{if } Y_s > q_2 .
           \end{array} \right.
   \end{align*}
   Let $Y^{t-1}=(Y_{t-1},Y_{t-2},\ldots,Y_0)$.
   Our ordered logit model  in \eqref{model} then implies that the joint distribution of
   $\widetilde Y = (\widetilde Y_1,\widetilde Y_2,\widetilde Y_3)$ and $W$ conditional on
     $A=\alpha$, $Y^{t-1}=y^{t-1}$, $X=(x_1,x_2,x_3)$,  and  $\theta=\theta^0$
   satisfies \eqref{AssMarkov} and \eqref{AssLogit}, as long as we choose
\begin{align*}
      f(y_1 \,|\, 1 ) &= {\rm Pr}\left(  Y_t=y_1\, \big| \, Y_t > q_1 , \, Y^{t-1}=y^{t-1}, \, X=x, \, A=\alpha \right) ,
     \\
      f(y_1 \,|\, 0 ) &= {\rm Pr}\left(  Y_t=y_1\, \big| \, Y_t \leq q_1 , \, Y^{t-1}=y^{t-1}, \, X=x, \, A=\alpha \right) ,
     \\
      p_3(\widetilde y_3 \,| \, 1, \, y_{s-1}) &=  \left\{
    \begin{array}{ll}    
       {\rm Pr}\left(  Y_{s+1} \leq q_3  \, \big| \, Y_s < q_2  , \, Y_{s-1}=y_{s-1}, \, Y^{t-1}=y^{t-1}, \, X=x, \, A=\alpha \right)   & \text{if $\widetilde y_3=0$,} 
      \\
     {\rm Pr}\left(  Y_{s+1} > q_3  \, \big| \, Y_s < q_2  , \, Y_{s-1}=y_{s-1}, \, Y^{t-1}=y^{t-1}, \, X=x, \, A=\alpha \right)    & \text{if $\widetilde y_3=1$,}
    \end{array}
    \right.      
     \\
      p_3(\widetilde y_3 \,| \, 3, \, y_{s-1}) &=  \left\{
    \begin{array}{ll}    
       {\rm Pr}\left(  Y_{s+1} \leq q_3  \, \big| \, Y_s > q_2  , \, Y_{s-1}=y_{s-1}, \, Y^{t-1}=y^{t-1}, \, X=x, \, A=\alpha \right)   & \text{if $\widetilde y_3=0$,} 
      \\
     {\rm Pr}\left(  Y_{s+1} > q_3  \, \big| \, Y_s > q_2  , \,  Y_{s-1}=y_{s-1}, \, Y^{t-1}=y^{t-1}, \, X=x, \, A=\alpha \right)    & \text{if $\widetilde y_3=1$,}
    \end{array}
    \right.      
\end{align*}
and
\begin{align*}       
      \pi_{1} &= \alpha + z(y_{t-1},x_t,\theta^0)  - \lambda_{q_1} =   \alpha +  x'_t \, \beta^0 + \gamma^0_{y_{t-1}}  - \lambda_{q_1} ,
      \\
      \pi_{2,1}(y_1) &= \alpha + z(y_{s-1},x_s,\theta^0)  - \lambda_{q_2-1} =  \alpha +   x'_s \, \beta^0 + \gamma^0_{y_{s-1}} -   \lambda_{q_2-1},
      \\
       \pi_{2,2}(y_1) &= \alpha + z(y_{s-1},x_s,\theta^0) - \lambda_{q_2}  =  \alpha + x'_s \, \beta^0 + \gamma^0_{y_{s-1}} - \lambda_{q_2} ,
     \\
       \pi_3 &=  \alpha +  z(q_s,x_{s+1},\theta^0) -  \lambda_{q_3} = \alpha +   x'_{s+1} \, \beta^0 + \gamma^0_{q_2} - \lambda_{q_3} ,
\end{align*}
where $w=y_{s-1}$, and $z(y_{t-1},x_t,\theta) $ is defined in \eqref{DefSingleIndex}.
Plugging those expressions for $\pi_{1}$, $  \pi_{2,1}(w)$,  $\pi_{2,2}(w)$ and $\pi_3$
into the moment function $m(\widetilde y, w)$   in \eqref{DefMgeneral} we find that this moment function exactly 
coincides with $ m^{(t,s,s+1)}_{y_0,q_1,q_2,q_3}(y,x,\theta^0)$ in equation \eqref{MomentFunctionsTgeneral} of the main text.
Thus, by applying Lemma~\ref{lemma:MASTER} to the distribution of $Y$ conditional on  $A=\alpha$, $Y^{t-1}=y^{t-1}$, and $X=x$
(the lemma does not feature those conditioning variables, which is why we are applying the lemma to the conditional distribution), we obtain
\begin{align*}
\mathbb{E} \left[ m^{(t,s,s+1)}_{y_0,q_1,q_2,q_3}(Y,X,\theta^0) \, \big| \, Y^{t-1}=y^{t-1}, \, X=x, \, A=\alpha \right] &= 0 .
\end{align*}
Applying the law of iterated expectations, we thus also find that
\begin{align*} 
\mathbb{E}\left[ w(Y_{1},\ldots ,Y_{t-1})\,m^{(t,s,s+1)}_{y_0,q_1,q_2,q_3}(Y,X,\theta^0)\,\Big|\,Y_{0}=y_{0},\,X=x,\,A=\alpha \right] & =0 ,
\end{align*}
which concludes the proof for the case $q_2 \in \{2,\ldots,Q-1\}$.

\medskip

   \underline{Case $q_2 =Q $:} \; 
   In this case, we choose
   \begin{align*}
           W &:= Y_{s-1} ,
           &
           V &:= Y_{r-1} ,           
       &
        \widetilde Y_t &:=   \mathbbm{1}\left\{  Y_1 > q_1 \right\} ,
       &
        \widetilde Y_s &:=  \mathbbm{1}\left\{  Y_2 = Q \right\} ,
       &
        \widetilde Y_r &:=  \mathbbm{1}\left\{  Y_3 > q_3 \right\} .
   \end{align*}
   Our ordered logit model  in \eqref{model} then implies that the joint distribution of
   $\widetilde Y = (\widetilde Y_1,\widetilde Y_2,\widetilde Y_3)$, $W$, and $V$  conditional on
     $A=\alpha$, $Y^{t-1}=y^{t-1}$, $X=(x_1,x_2,x_3)$,  and  $\theta=\theta^0$ 
   satisfies  the assumptions of Lemma~\ref{lemma:MASTER2}, as long as we choose
\begin{align*}
      f(y_{s-1} \,|\, 1 ) &= {\rm Pr}\left(  Y_{s-1}=y_{s-1} \, \big| \, Y_t > q_1 , \, Y^{t-1}=y^{t-1}, \, X=x, \, A=\alpha \right) ,
     \\
      f(y_{s-1} \,|\, 0 ) &= {\rm Pr}\left(  Y_{s-1}=y_{s-1} \, \big| \, Y_t \leq q_1 , \, Y^{t-1}=y^{t-1}, \, X=x, \, A=\alpha \right) ,
     \\
      g(y_{r-1}  \,|\, 1) &=  \mathbbm{1}\left\{  y_{r-1} = Q \right\} ,
     \\
      g(y_{r-1} \,|\, 0, y_{s-1}) &= {\rm Pr}\left(  Y_{r-1}=y_{r-1}\, \big| \, Y_s < Q , \, Y_{s-1}=y_{s-1}, \, X=x, \, A=\alpha \right) ,
   \end{align*}
and
\begin{align*}       
      \pi_{1} &= \alpha + z(y_{t-1},x_t,\theta^0)  - \lambda_{q_1} =   \alpha +  x'_t \, \beta^0 + \gamma^0_{y_{t-1}}  - \lambda_{q_1} ,
      \\
      \pi_{2}(y_{s-1}) &= \alpha + z(y_{s-1},x_s,\theta^0)  - \lambda_{Q-1} =  \alpha +   x'_s \, \beta^0 + \gamma^0_{y_{s-1}} -   \lambda_{Q-1},
     \\
       \pi_3(y_{r-1}) &=  \alpha +  z(y_{r-1},x_r,\theta^0) -  \lambda_{q_3} = \alpha +   x'_r \, \beta^0 + \gamma^0_{y_{r-1}} - \lambda_{q_3} ,
\end{align*}
where $w=y_{s-1}$ and $v=y_{r-1}$, and $z(y_{t-1},x_t,\theta) $ is defined in \eqref{DefSingleIndex}.
Plugging those expressions for $\pi_{1}$, $   \pi_{2}(y_{s-1})$,  and $ \pi_3(y_{r-1}) $
into the moment function $m(\widetilde y, w,v)$   in  Lemma~\ref{lemma:MASTER2}  we find that this moment function exactly 
coincides with $m_{y_0,q_1,Q,q_3}(y,x,\theta^0) $ in equation \eqref{MomentFunctionsTgeneralBC} of the main text.
Thus, by applying Lemma~\ref{lemma:MASTER2} to the distribution of $Y$ conditional on  $A=\alpha$, $Y^{t-1}=y^{t-1}$, and $X=x$
(the lemma does not feature those conditioning variables, which is why we are applying the lemma to the conditional distribution) we obtain
\begin{align*}
\mathbb{E} \left[ m^{(t,s,r)}_{y_0,q_1,Q,q_3}(Y,X,\theta^0) \, \Big| \, Y^{t-1}=y^{t-1}, \, X=x, \, A=\alpha \right] &= 0 .
\end{align*}
Applying the law of iterated expectations, we thus also find that
\begin{align*} 
\mathbb{E}\left[ w(Y_{1},\ldots ,Y_{t-1})\,m^{(t,s,r)}_{y_0,q_1,Q,q_3}(Y,X,\theta^0)\,\Big|\,Y_{0}=y_{0},\,X=x,\,A=\alpha \right] & =0 ,
\end{align*}
which concludes the proof for the case $q_2=Q$.

\medskip

   \underline{Case $q_2 =1 $:} \; 
   The result for this case again follows from the result for $q_2=Q$ by applying the transformation $Y_t \mapsto Q+1-Y_t $,
   $\lambda_q \mapsto - \lambda_{Q-q}$, $\beta \mapsto - \beta$, $\gamma_q \mapsto - \gamma_{Q+1-Y_t}$, $A_i \mapsto - A_i$.
\end{proof}

\subsection{Proof of Proposition~\ref{prop:IDgamma}, \ref{prop:IDbeta}, and \ref{prop:IDlambda}}

The following lemma is useful for the proof of Proposition~\ref{prop:IDgamma}.

\begin{lemma}
     \label{lemma:HelpProp1}
     Let $Q\geq2$. Let $B$ be a $Q \times Q$ matrix for which all non-diagonal elements are positive (i.e.\ $B_{q,r}>0$ for $q \neq r$).
     Let $g^0,g \in (0,\infty)^Q$ be two vectors with only positive entries.
     Assume that $B g^0=0$ and $B g = 0$.
     Then there exists $\kappa >0$ such that $g = \kappa \, g^0$.    
\end{lemma}
 
 \begin{proof}[\bf Proof]
     This is a proof by contradiction. Let all assumptions of the lemma be satisfied, 
     and assume that 
     there does \underline{not} exist a $\kappa >0$ such that
      $g = \kappa \, g^0$. Define the vector $h \in [0,\infty)^Q$ 
     and the two  sets ${\cal Q}_+,{\cal Q}_0 \subset \{1,\ldots,Q\}$ by
     \begin{align*}
           h &:=  g^0   -    \left(  \min_{q \in \{1,\ldots,Q\}} \frac{g^0_q} {g_q}  \right) \, g ,
           &
           {\cal Q}_+  &:=   \left\{  q   \, : \,   h_q > 0 \right\} ,
           &
           {\cal Q}_0  &:=   \left\{  q   \, : \,   h_q = 0 \right\}  .
     \end{align*}     
     All elements of $h$ are non-negative by construction,
     and we have $h \neq 0$, because otherwise we would have $g = \kappa \, g^0$ for some $\kappa >0$.
     Therefore, neither $ {\cal Q}_+$ nor ${\cal Q}_0$ are empty sets.
     Furthermore, since
     $h$ is a linear combination of   $g^0$ and $g$, and we have $B g^0= B g = 0$,
     we also have $Bh=0$. This can equivalently be written as
\begin{align*}
      \sum_{r \in  {\cal Q}_+} \,  B_{q,r}  \;    h_r = 0 ,    \qquad \text{for all } q \in \{1,\ldots,Q\},
\end{align*} 
where we dropped the  indices $r$ from the sum for which we have $h_r=0$.

Now, let $q \in {\cal Q}_0$.
We then have  $q \notin  {\cal Q}_+$, and therefore
$B_{q,r} >0$ for all $r \in  {\cal Q}_+$, according to our assumption on $B$.
 We have argued that  ${\cal Q}_+$ is non-empty, and by construction we have $h_q>0$ for $q \in  {\cal Q}_+$. 
We therefore have
\begin{align*}
     \sum_{r \in  {\cal Q}_+} \,  B_{q,r}  \;    h_r > 0.
\end{align*}
The last two displays are the contradiction that we wanted to derive here.
 \end{proof}
 
\bigskip
 
 \begin{proof}[\bf Proof of Proposition~\ref{prop:IDgamma}]
 Let $x_{(1)}=(x_1,x_1,x_1) $.
 For $y_0,q \in \{1,\ldots,Q\}$ we define  
\begin{align*}
     B_{y_0,q}  :=  \left\{  \begin{array}{ll}
                     {\rm Pr}\left( y_1 > 1 \; \& \;  y_2=1 \; \& \;   y_3= 1 \mid Y_0=y_0, \, X=x_{(1)} \right) & \text{if  $y_0 \neq q$ and $q=1$,}
                     \\
                     {\rm Pr}\left( y_1 = q \; \& \;  y_2=1 \; \& \;   y_3> 1 \mid Y_0=y_0, \, X=x_{(1)}  \right) & \text{if  $y_0 \neq q$ and  $q>1$,}
                     \\
                     {\rm Pr}\left( y_1 > 1 \; \& \;  y_2=1 \; \& \;   y_3= 1 \mid Y_0=y_0, \, X=x_{(1)}  \right)
                   \\ \qquad  \qquad  \qquad 
                     -  {\rm Pr}\left( y_1 = 1 \; \& \;  y_2>1   \mid Y_0=y_0, \, X=x_{(1)}  \right)
                      & \text{if  $y_0 = q$ and $q=1$,}
                     \\
                     {\rm Pr}\left( y_1 = q \; \& \;  y_2=1 \; \& \;   y_3> 1 \mid Y_0=y_0, \, X=x_{(1)}  \right) 
                   \\ \qquad  \qquad  \qquad   
                                      -  {\rm Pr}\left( y_1 = 1 \; \& \;  y_2>1   \mid Y_0=y_0, \, X=x_{(1)}  \right)
                     & \text{if  $y_0 = q$ and  $q>1$.}
       \end{array} \right.
\end{align*}
Let $B$ be the $Q \times Q$ matrix with entries $ B_{y_0,q} $.
Our assumptions guarantee that all the conditional probabilities  that enter into the definition of $B_{y_0,q} $ are non-negative,
and we therefore have
\begin{align}
    B_{y_0,q}  &> 0 , \qquad \text{for all $y_0 \neq q$}    .
    \label{Aconditions}
\end{align}
Applying Theorem~\ref{th:momentsT3}
we find that  
the moment function in \eqref{MomentFunctionID}  satisfies
\begin{align}
   \mathbb{E}\left[ m_{y_0}(y,\gamma^0) \mid Y_0=y_0, \, X=(x_1,x_1,x_1)  \right]=0,
        \qquad \text{for all $y_0 \in \{1,\ldots,Q\}$},
    \label{MomentConditionID}
\end{align}    
where $\gamma^0$ is the true parameter that generates the data.
In the proposition, we assume that  $\gamma \in \mathbb{R}^Q$ is an alternative parameter that satisfies the
same moment conditions.
Let $g^0$ and $g$ be the $Q$-vectors with entries
$g^0_q := \exp\left(  \gamma^0_q   \right)>0$
and $g_q := \exp\left(  \gamma_q   \right)$.
Using the definition of the matrix $B $
we can rewrite the two systems of $Q$ equations in \eqref{MomentConditionID} and \eqref{MomentConditionIDprop1} as
\begin{align}
     B \, g^0 &= 0 ,
     &
     B \, g  = 0 .
     \label{SystemIdentificationX0}
\end{align}
Since we have \eqref{Aconditions} and \eqref{SystemIdentificationX0} we can apply
Lemma~\ref{lemma:HelpProp1}
to find that there exists  $\kappa>0$ such that $g = \kappa \, g^0$. Taking logarithms, 
we thus have $\gamma = \gamma^0+ c$, where $c=\log(\kappa) $.
This is what we wanted to show.
\end{proof}

\bigskip

The following lemma is useful for the proof of Proposition~\ref{prop:IDbeta}.

\begin{lemma}
\label{lemma:INVERSION} Let $K \in \mathbb{N}$. For every $s =
(s_1,\ldots,s_{K-1},+) \in \{-,+\}^{K}$, let $g_s : \mathbb{R}^{K}  
\rightarrow \mathbb{R}$ be a continuous function such that for all $%
\beta \in \mathbb{R}^{K} $, we have

\begin{itemize}
\item[(i)] $g_s(\beta)$ is strictly increasing in $\beta_K$.

\item[(ii)] For all $m \in \{1,\ldots,K-1\}$: If $s_m = +$, then $g_s(\beta)$ is strictly increasing in $\beta_m$.

\item[(iii)] For all  $m \in \{1,\ldots,K-1\}$: If $s_k = -$, then $g_s(\beta)$ is strictly decreasing in $\beta_m$.
\end{itemize}
Then, the system of $2^{K-1}$ equations in $K$ variables
\begin{align*}
g_s(\beta)=0    \qquad \text{for all $s \in  \{-,+\}^{K}$ with $s_K=+$,}
\end{align*}
has at most one solution.
\end{lemma}

\begin{proof}[\bf Proof]
    This is the same as Lemma~2 in  \cite{honore2020dynamic}, only presented using slightly different notation here.
\end{proof}

\bigskip

\begin{proof}[\bf Proof of Proposition~\ref{prop:IDbeta}]
 For  $s =(s_1,\ldots,s_{K-1},+) \in \{-,+\}^{K}$ we define
\begin{equation}
    g_{s}(\beta) =  \mathbb{E}\left[ m_{y_0,1,1,1}(Y,X,\beta,\gamma_0) \,\Big|\,Y_{0}=y_{0},\;X\in \mathcal{X}_s\right],
    \label{eq:prop2_expectation}
\end{equation}
where $m_{y_0,1,1,1}(y,x,\beta,\gamma)$ is the moment function in \eqref{momIdentify}.
Our assumptions guarantee that the  conditioning sets in \eqref{eq:prop2_expectation}
have positive probability, 
which, together with the definition of $m_{y_0,1,1,1}(y,x,\beta,\gamma_0)$ and $\mathcal{X}_s$, 
guarantee that the functions $g_{s}(\beta)$ satisfy the  monotonicity requirements (i), (ii) and (iii) of Lemma~\ref{lemma:INVERSION}. 
Theorem~\ref{th:momentsT3} guarantees that
\begin{align*}
g_s(\beta^0)=0 ,   \qquad \text{for all $s \in  \{-,+\}^{K}$ with $s_K=+$,}
\end{align*}
where $\beta^0$ is the true parameter value that generates the data.
Equation~\eqref{Condition:prop:IDbeta} in the proposition can equivalently be written as
\begin{align*}
g_s(\beta)=0 ,   \qquad \text{for all $s \in  \{-,+\}^{K}$ with $s_K=+$.}
\end{align*}
According to Lemma~\ref{lemma:INVERSION}, the system of equations in the last two displays can have at most one solution, 
and we, therefore, must have $\beta=\beta_0$.
\end{proof}

\bigskip

\begin{proof}[\bf Proof of Proposition~\ref{prop:IDlambda}]
    The 
    definition of $m_{y_0,q_1,1,1}(y,x,\beta,\gamma,\lambda)$ in \eqref{Moments:prop:IDlambda}
    together with
    the    
    assumptions of the proposition guarantee that
    $g(\lambda) :=  \mathbb{E}\left[ m_{y_0,q_1,1,1}(Y,X,\beta^0,\gamma^0,\lambda) \,\Big|\,Y_{0}=y_{0} \right] $
    is strictly increasing in $\lambda_{q_1} - \lambda_1$ for all $q_1  \in \{2,\ldots,Q-1\}$.
    Theorem~\ref{th:momentsT3} guarantees that $g(\lambda^0)=0$ for the true parameter $\lambda^0$ that generates the data.
    For any $\lambda \in \mathbb{R}^{Q-1}$ that satisfies  $g(\lambda)=0$,
    we therefore must have $\lambda_{q_1} - \lambda_1 = \lambda^0_{q_1} - \lambda^0_1$,
    which implies $\lambda = \lambda^0 + c$ for $c = \lambda_1 -  \lambda^0_1$.    
\end{proof}

\subsection{Lower bound on the number of moment conditions}
\label{app:LowerBound}

\cite{dobronyi2021identification} point out that it is sometimes possible to easily
derive a lower bound on the number of linearly independent
valid moment conditions in a given panel data model. 
The key insight (to be verified below) is that one can write the model
probabilities defined in equation \eqref{DefProb}
as
\begin{equation}
p_{y_{0}}(y,x,\theta ,\alpha)
 =\kappa \left( a \right) \,
\sum_{k=1}^{K}a^{k-1}c_{k}\left( y \right) ,
   \label{Polynomial}
\end{equation}%
for some constant $K \in \{1,2,\ldots\}$, some positive function $\kappa$ of $%
a=\exp \left( \alpha \right) $ that does not depend on $y$,
and some functions $c_k$ of $y$ that do not depend on $a$.
Here, the functions $\kappa$
and $c_k$ also depend on
$y_0$, $x$, $\theta$,
but  those arguments are dropped to focus more clearly on the
dependence on $\alpha$ and $y$. In other words,
$y_0$, $x$, $\theta$ are simply assumed fixed here. The dependence
of all functions and constants on $T$
is also not made explicit
(here or anywhere else in the paper).

Once we have shown \eqref{Polynomial}, then a
a valid moment function must satisfy
\begin{equation}
\sum_{y\in \mathcal{Y}}m(y)\,\sum_{k=1}^{K}a^{k-1}c_{k}\left(
y\right) =0,\qquad \text{for all }a\in \left( 0,\infty \right) ,
\label{EQ: Polynomial}
\end{equation}%
which is equivalent to
\begin{equation*}
\sum_{y\in \mathcal{Y}}\,m(y)\,c_{k}(y)=0,\qquad \text{for all $k\in
\{1,\ldots ,K\}$.}
\end{equation*}%
These are $K$ linear conditions in $Q^T$
unknown parameters $m(y)$. We, therefore, have at least $Q^T -K$ linearly independent solutions $m(y)$. In other words, the
model must have at least $Q^T -K$ linearly independent conditional
moment conditions.

What is left to do now is to show that
\eqref{Polynomial} indeed holds for our
the dynamic order choice model, for
some $K<Q^T$.
For that purpose,
remember that $ \Lambda(\varepsilon) =  [1+\exp(-\varepsilon)]^{-1}$, and also define
$
    \omega_{t,q}  =  X_{t}^{\prime }\,\beta  - \lambda_{q}  ,
$ 
and 
\begin{align*}
    a&=\exp (\alpha ) ,
    &
     w_q &= \exp( \omega_{t,q}  ) ,
     & 
     g_q &=  \exp( \gamma_{q}  ) .
\end{align*}    
Also remember that
$\lambda_{0} = -\infty$ and $\lambda_{Q} = \infty$, which implies that
\begin{align*}
    w_0 &= \infty ,
    &
    w_Q &= 0.
\end{align*}  
According to \eqref{DefProb},
the probability of observing $Y=y$,
conditional on $Y_{0}=y_0$, $X=x$, $A=\alpha$, $\theta=\theta_0$, is then given by
\begin{align*}
p_{y_{0}}(y,x,\theta ,\alpha)
  &=
\prod_{t=1}^{T}
\left\{  \Lambda\Big[  \omega_{t,(y_{t}-1)} +  \gamma_{y_{t-1}} + \alpha \Big] -  \Lambda\Big[   \omega_{t,y_{t}} +  \gamma_{y_{t-1}} + \alpha   \Big]
\right\}
   \\ &= 
   \prod_{t=1}^{T} \left\{   \frac { w_{t,(y_{t}-1)} \, g_{y_{t-1}} \, a  } {1+w_{t,(y_{t}-1)} \, g_{y_{t-1}} \, a }
                          -   \frac { w_{t,y_{t}} \, g_{y_{t-1}} \, a  } {1+w_{t,y_{t}} \, g_{y_{t-1}} \, a } \right\}
   \\ &= \kappa \left( a \right) \, \widetilde p(y,a)           ,
\end{align*}                          
where  
\begin{align*}
 \kappa \left( a \right)  &=   
     \prod_{q=1}^{Q-1} \, \frac 1  {1+w_{t,q} \, g_{y_0} \, a }    \,
      \prod_{t=2}^{T}   \prod_{r=1}^{Q}  \frac 1  {1+w_{t,q} \, g_{r} \, a }  ,
\end{align*}
and
 \begin{align*}
    \widetilde p(y,a)  &=     
    \bigg\{  \mathbbm{1}\{y_1 \neq 1\} \,  w_{1,(y_1-1)} \, g_{y_0} \, a   \prod_{q \in \{1,\ldots,Q-1\} \setminus \{y_1 - 1\}} \, \left(1+w_{1,q} \, g_{y_0} \, a \right) 
        \\ & \qquad \qquad\qquad\qquad 
          -    \mathbbm{1}\{y_1 = Q\}  
                          -    \mathbbm{1}\{y_1 \neq Q\} \,    w_{1,y_1} \, g_{y_0}  \, a \,     \prod_{q \in \{1,\ldots,Q-1\} \setminus \{y_1\}} \, \left(1+w_{1,q} \, g_{y_0} \, a \right)  \bigg\}
                          \\ & \quad 
                   \times       \prod_{t=2}^{T}
                             \bigg\{  \mathbbm{1}\{y_t \neq 1\} \,  w_{t,(y_t-1)} \, g_{y_{t-1}} \, a  
                              \prod_{(q,r) \in \{1,\ldots,Q-1\} \times  \{1,\ldots,Q\} \setminus \{(y_t - 1,y_{t-1}\}} 
                       \, \left(1+w_{t,q} \, g_{r} \, a \right) 
        \\ & \quad
          -    \mathbbm{1}\{y_t = Q\}  
                          -    \mathbbm{1}\{y_t \neq Q\} \,    w_{t,y_t} \, g_{y_{t-1}}  \, a \,   
                              \prod_{(q,r) \in \{1,\ldots,Q-1\} \times  \{1,\ldots,Q\} \setminus \{(y_t ,y_{t-1}\}} 
                             \,   \left(1+w_{t,q} \, g_{y_r} \, a \right)  \bigg\}      .
\end{align*}
Here, $\kappa(a)$ does not depend on $y$,   and $   \widetilde p(y,a) $ is a polynomial in $a$ 
 of order
$$K - 1 =  (Q-1)+(T-1) Q (Q-1) . $$
This implies that a lower bound on the number of linearly independent valid moment conditions in this model is given by
\begin{align*}
    Q^T - K  &= Q^T - Q - (T-1) Q (Q-1) 
    \\
    &= Q^T - (T-1) Q^2 + (T-2) Q  ,
\end{align*}
which is exactly the number of linearly independent valid moment conditions we found in the main text.

\subsection{An alternative specification}\label{app:alternative_specification}

The model analyzed in this paper specifies that the latent variable evolves according to
\begin{align*}
    Y^*_{it} = X_{it}^{\prime }\,\beta  +  \sum_{q=1}^{Q} \gamma_q  \, \mathbbm{1}\left\{ Y_{i,t-1} = q \right\} + A_i + \varepsilon_{it},
\end{align*}
see \eqref{ModelYstar}. An alternative specification would have the dynamics in the lagged latent variable, for example
\begin{align}
    Y^*_{it} = X_{it}^{\prime }\,\beta  +  \gamma Y_{i,t-1}^* + A_i + \varepsilon_{it}, \label{eq:dynamics_latent}
\end{align}

In many cases, one is interested in the dynamics in the actual discrete outcome as in our specification \eqref{ModelYstar}. For example, the outcome of a soccer game is clearly ordered (win, draw, defeat) and there may be strategic or physiological effects of past results on actual outcomes. In such a case, the dynamics are likely to be through the actual outcome rather than some underlying index.

A second example is the empirical application in \cite{muris2020dynamic} on self-reported health status: ``While there are multiple levels of self-reported health, being above or below a threshold at which medical care is demanded in period $t-1$ is essential in determining health status in period $t$. \cite{galamakapteyn} formalized this important issue by building on the seminal work by \cite{grossmanJPE}.''
A third example comes from the analysis of government bonds (see also the discussion in \cite{muris2020dynamic}. Government bonds are ordered ratings (20 tiers), which can be grouped into two categories: Investment Grade (IG) and Non-Investment Grade (NIG). IG Governments can easily raise credit, whereas NIG Governments cannot. This has implication for next period's credit rating for these governments. Thus, the actual lagged value of the ordered variable -- not the latent index -- affects a country’s current credit rating, cf. \cite{rigobon2002}. Additional examples may include the analysis of demand for ordered goods. For example, goods could be available in discrete quantities (8 oz, 12 oz, or 16 oz). In that case, it is more reasonable to model today's choice in terms of the actual discrete choice last period. Goods can also be ordered in terms of quality. In that case, one might model this period's choice  (low, medium, high quality variety of a product) in terms of the actual choice in the previous periods. In conclusion, there are many settings of applied interest where interest is in the relationship between the current and lagged actual value of the discrete variable.

However, in some applications, \eqref{eq:dynamics_latent} may be a more natural specification. With that in mind, we performed a numerical experiment to explore what happens when the techniques in our paper are applied to data from a DGP with dynamics that are linear in the lagged latent variable $Y_{i,t-1}^*$.

Specifically, we consider the case of binary choice, which would be the worst case in terms of approximation error. Specifically, we consider the dynamic binary choice model with fixed effects, with $Q=2$ and $T=3$, with $(X_1,X_2,X_3,Y_0) \in \{0,1\}^4$ with equal probabilities, and set $\beta_{1,0} = 1$. We let the autoregressive parameter range over $$\gamma_0 \in \{-0.9,-0.75,\cdots,0.6,0.75,0.9\},$$ and let the fixed effects takes on the values $A \in \{-2,-1,0,1,2\}$.

We generate data from a DGP1 that is consistent with our model \eqref{ModelYstar}, and from a DGP2 with outcome equation \eqref{eq:dynamics_latent}. For DGP2, we draw $Y_0^* \sim \mathcal N(0,1)$. 

We construct optimal instruments under DGP1 evaluated at the true value of the parameters,
$$
Z^*(y_0,x_1,x_2,x_3;\beta_{10},\gamma_0) = D'(y_0,x_1,x_2,x_3,\beta_{10},\gamma_0) \Omega(y_0,x_1,x_2,x_3;\beta_{10},\gamma_0)
$$
and form the unconditional moments
$$ 
h(b,c) = E\left[Z^*(Y_0,X_1,X_2,X_3;\beta_{10},\gamma_0) m(Y_0,X_1,X_2,X_3;b,c)\right].
$$
To explore the sensitivity of our results to DGP2, we approximate $E(h(b,c))$ using 500000 draws for each value of the conditioning variables. We determine the values of $(b,c)$ that set the moments equal to zero under both DGP1 and DGP2. 
     
Figure \ref{fig:approximation_lagged_latent} reports the results for the regression coefficient. Each panel corresponds to a value of $A$. The autoregressive parameter $\gamma_0$ varies along the horizontal axis. Results for DGP1 are covered by the horizontal line at $\beta_{10} = 1$, as expected. Results for DGP2 are given by the dashed line. Even for large values of $\gamma_0$, the pseudo-true values of $b$ are close to $\beta_{10}=1$, deviating by at most 0.25. Thus, the moment conditions seem to be robust against severe misspecification.
\begin{figure}
\centering
\includegraphics[scale = 0.75]{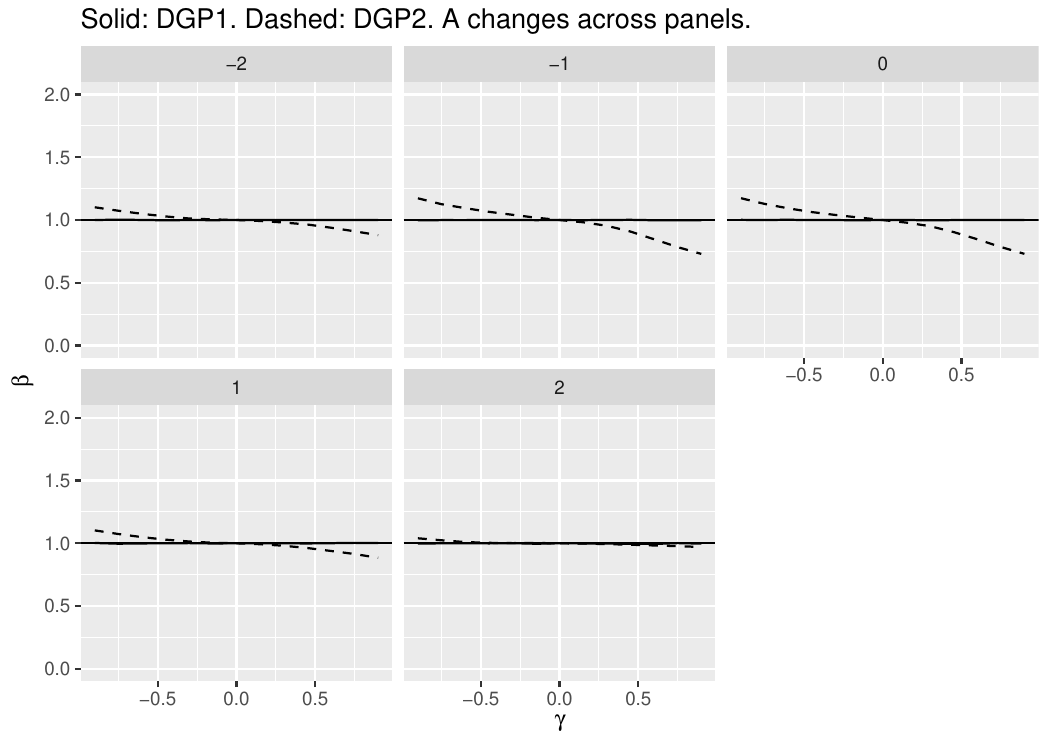}
\caption{Results sensitivity analysis, $\beta$.}
\label{fig:approximation_lagged_latent}
\end{figure}

We also report results for $\gamma$, in Figure \ref{fig:approximation_lagged_latent_gamma}. Some care is required in interpreting these results, because the interpretation of $\gamma$ differs across the two DGPs. When the true $\gamma$ (horizontal axis) is 0, the (pseudo-)true values of $\gamma$ are 0 for both DGPs. Results for DGP1 (solid line) are on the 45 degree line, as expected. Results for DGP2 always have the right sign, and are typically close to the true $\gamma$.
 
\begin{figure}
\centering
\includegraphics[scale = 0.75]{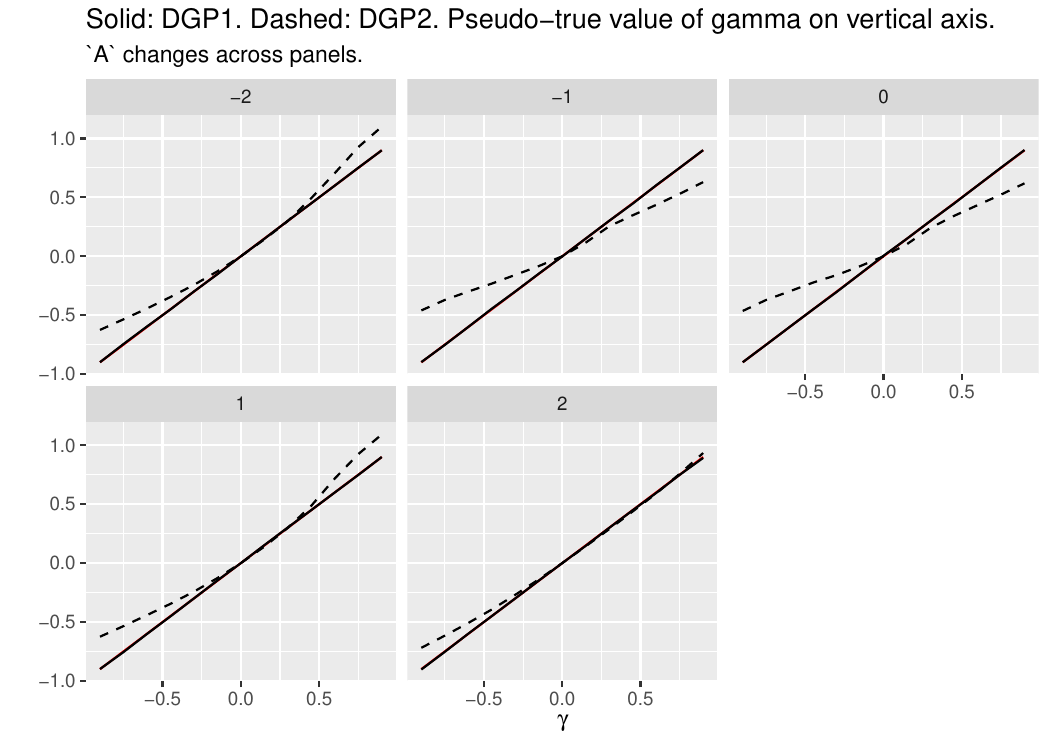}
\caption{Results sensitivity analysis, $\gamma$.}
\label{fig:approximation_lagged_latent_gamma}
\end{figure}

\end{document}